%% file: constraint_freezeml.tex
\newif\ifcomments\commentsfalse 
\newif\iflegacydoc\legacydocfalse 
\newif\ifnewranksection\newranksectionfalse

\documentclass[acmsmall,nonacm]{acmart}

\usepackage{amsthm}
\usepackage{mathtools}
\usepackage[shortcuts]{extdash} 
\usepackage{xcolor}             
\usepackage{xspace}             
\usepackage{mathpartir}         
\usepackage{stmaryrd}
\usepackage{placeins}           
\usepackage{enumitem}           
\usepackage{cleveref}

\usepackage{tikz}
\usetikzlibrary{positioning}
\usetikzlibrary{calc}
\usetikzlibrary{arrows}
\usetikzlibrary{decorations.pathreplacing}

\crefname{appendix}{Appendix}{Appendix}
\Crefname{appendix}{Appendix}{Appendix}


\input{macros}

\ifcomments
\newcommand{\note}[1]{{\color{red} {\textsc{NOTE:}} #1}}
\else
\newcommand{\note}[1]{}
\fi

\begin{document}
\title{Constraint-based type inference for FreezeML}
\author{Frank Emrich}
\email{frank.emrich@ed.ac.uk}

\author{Jan Stolarek}
\email{jan.stolarek@ed.ac.uk}


\additionalaffiliation{%
  \institution{Lodz University of Technology}
  \country{Poland}
}

\author{James Cheney}
\email{james.cheney@ed.ac.uk}


\additionalaffiliation{%
  \institution{The Alan Turing Institute}
  \country{UK}
}

\author{Sam Lindley}
\email{sam.lindley@ed.ac.uk}

\affiliation{%
  \institution{The University of Edinburgh}
  \country{UK}
}

\renewcommand{\shortauthors}{Emrich et al.}

\begin{abstract}
  FreezeML is a new approach to first-class polymorphic type inference that
  employs term annotations to control when and how polymorphic types are
  instantiated and generalised. It conservatively extends Hindley-Milner type
  inference and was first presented as an extension to Algorithm W. More modern
  type inference techniques such as \HMX and OutsideIn($X$) employ constraints to
  support features such as type classes, type families, rows, and other
  extensions. We take the first step towards modernising FreezeML by presenting
  a constraint-based type inference algorithm. We introduce a new constraint
  language, inspired by the Pottier/R\'emy presentation of \HMX, in order to allow
  FreezeML type inference problems to be expressed as constraints. We present a
  deterministic stack machine for solving FreezeML constraints and prove its
  termination and correctness.
\end{abstract}

\maketitle
\bibliographystyle{ACM-Reference-Format}
\citestyle{acmauthoryear}   

\section{Introduction}

Hindley-Milner type inference is well-studied, yet extending it to provide full
support for polymorphism (``first-class'' polymorphism a la System F) remains an
active research topic---characterised in one recent paper as ``a deep, deep
swamp''~\citep{SerranoHVJ18}.
A term such as $\lambda f. f\: f$ , which would be rejected by Hindley-Milner
type inference, may be accepted by a type system permitting first-class
polymorphism, by assigning a sufficiently polymorphic type to $f$, such as
$\forall a. a \to a$.
However, type inference for System F is undecidable~\citep{Wells94}, meaning that
some restrictions must be imposed.
Choosing $\forall a. a$ as the type for $f$ also allows the example above to
type-check, but no System F type can be given to the function that subsumes both
choices.
A wide range of solutions has emerged to explore the resulting design space,
yielding systems that go beyond System F types, employ elaborate heuristics that
determine the system's behaviour, require type annotations for certain terms, or
rely on additional syntax, or give up on completeness or principal typing, to
name a few~\citep{%
BotlanR03,%
Leijen08,%
RussoV09,%
SerranoHVJ18,%
GarrigueR99,%
SerranoHJV20,%
VytiniotisWJ06}.

Recently, \citet{EmrichLSCC20} proposed a new approach called FreezeML that has
several desirable properties: it conservatively extends ML type inference,
allows expressing arbitrary System F types and computations, and retains
decidable, complete type inference. The key ingredient of FreezeML is the
``freezing'' operation, an annotation on term-level variables that \emph{blocks}
automatic instantiation of any quantifiers in that variable's type. FreezeML
also includes let- and lambda-bindings with ascribed types (which are standard
in other systems).
Unlike other approaches to first-class polymorphism that err on the side of
explicitness~\citep{RussoV09,GarrigueR99}, FreezeML uses just System F types
instead of introducing different, incompatible sorts of polymorphic types.

Freezing enables the programmer to control when instantiation happens instead of
requiring the type inference algorithm to guess or employ some heuristic that
the programmer must then work around.
For instance, suppose we have defined functions $\dec{single} : \forall
a.a \to \List\;a$, which creates a singleton list, and $\dec{choose} : \forall
a.a \to a \to a$, which returns one of its arguments (for type inference
purposes it does not matter which).
In FreezeML we can define $f_1 \; () = \dec{single}\;\dec{choose}$ and $f_2 \; ()
= \dec{single}\;\freeze{\dec{choose}}$. In the former (following the usual ML
convention) both $\dec{single}$ and $\dec{choose}$ are fully instantiated before
the body is generalised, hence $f_1 : \forall a. \dec{unit} \to \List\;(a \to a \to a)$.
In the latter, however, instantiation of $\dec{choose}$ is frozen, hence $f_2 :
\dec{unit} \to \List\;(\forall a.a \to a \to a)$.


The original presentation of FreezeML type inference was given as an extension
to Algorithm W~\citep{DamasM82}. Although Algorithm W is well-understood, many
modern type inference implementations, notably Haskell,
employ \emph{constraint-based type
inference}~\citep{Pottier14,PottierR05,VytiniotisJSS11} instead in which type
inference is split into two stages, mediated by an intermediate logical language
of \emph{constraints}~\citep{OderskySW99,PottierR05}. In the first stage,
programs $M$ are translated to constraints $C$ such that $C$ is solvable if and
only if $M$ is typable, and the solutions to $C$ are the possible types of
$M$. In the second stage, the constraint $C$ is solved (or shown to be
unsatisfiable), without further reference to $M$. Adopting a constraint-based
inference strategy has several potential benefits over the traditional Algorithm
W-style, including separating the core logic of type inference from the details
of the surface language, leveraging already-known efficient techniques for
constraint solver implementation, and supporting extensions such as type classes
and families, subtyping, rows, units of measure,
GADTs~\citep{OderskySW99,SimonetP07,VytiniotisJSS11}, etc.

One influential approach to constraint-based type inference
is \HMX~\citep{OderskySW99,PottierR05}, that is, Hindley-Milner type inference
``parameterised over X'', where $X$ stands for a constraint domain that can be
used in types.  For example, if $X$ is a theory of type equality, one obtains
standard Hindley-Milner type inference \HM{=}; if $X$ is a theory of row types
one obtains row type inference; if $X$ is a theory of subtyping one gets type
inference with subtyping.

In this paper, we take a first step towards such a constraint-parametric system,
a version of FreezeML parameterised in a constraint domain $X$, in the spirit of
\HMX. Specifically, we introduce a constraint language for FreezeML, inspired
by \HMX, in which type expressions can include arbitrary polymorphism, and which
provides suitable constraints to encode type inference for FreezeML
programs. (We have not yet explored parameterising the system over the
constraint domain $X$, but even adapting FreezeML to a constraint-based approach
turns out to require surmounting significant technical obstacles.) We also
provide a deterministic stack machine for solving these constraints (again
inspired by the presentation of constraint solving for \HMX by Pottier and
R\'emy). Full correctness proofs for both contributions are included in an
appendix.

Formulating a suitable constraint language for FreezeML and a (provably) correct
translation and sound and complete solver involves several subtleties. Handling
the freeze operator itself turns out to be straightforward by adding a
constraint that checks that the type of a variable exactly matches an expected
type. Unification of types needs to account for polymorphism occurring anywhere in a type, and constraints need to be
extended with universal quantifiers as well, in order to deal with polymorphism
in ascribed types. We also add a ``monomorphism constraint'' to enforce
FreezeML's requirement that certain types are required to be
monomorphic. Finally, to deal with FreezeML's approach to the value
restriction (in which let-bindings of non-values are allowed but not
generalised)
we introduce an additional constraint form to handle type inference of
non-generalisable expressions. However, the most challenging problem is to
design a constraint language and semantics that preserves the necessary
invariants to ensure that FreezeML type inference remains sound, complete, and
principal: specifically, to ensure that flexible type variables occurring in the
inferred types of variables are always monomorphic, which is necessary in
FreezeML to avoid the need to ``guess polymorphism'' when a polymorphic type is
instantiated.

%

Like certain other systems~\citep{DynamicsML,VytiniotisWJ06,Leijen08}, typing
derivations in FreezeML require principal types to be assigned to certain
subterms.
To the best of our knowledge, the inference algorithm shown in this paper is the
first one based on constraint solving for such a type system, requiring similar
principality conditions in the semantics of the constraint language.
Detailed proofs of correctness are provided in an appendix.

We characterise these contributions as a first step towards a longer-term goal:
parameterising FreezeML type inference over other constraint domains X.  This is
a natural next step for future work, and would enable experimentation with
combining FreezeML-style polymorphism with features found in other modern type
systems, such as Haskell's type classes and families (and the numerous libraries that rely on them), higher-kinded types, and
GADTs~\citep{VytiniotisJSS11}, row types as found in Links~\citep{LindleyC12},
Koka~\citep{Leijen14} or Rose~\citep{MorrisM19}, units of measure as found
in F\#~\citep{Kennedy09} and some Haskell libraries~\citep{Gundry15}. To
summarise, in this paper we:
\begin{itemize}
\item
  present background on FreezeML (Section~\ref{section:freezeml});

\item
  introduce a constraint language inspired by Pottier and R\'emy's presentation
  of \HMX and give a translation from FreezeML programs to constraints
  representing type inference problems (Section~\ref{section:language});

\item
  present a stack machine for solving the constraints (again inspired by Pottier and R\'emy's) and show that
  it is correct, deterministic, and terminating (Section~\ref{section:solving});

\item
  discuss extensions (Section~\ref{sec:discussion}), related and future work
  (Section~\ref{section:related}), and conclude (Section~\ref{section:concl}).
\end{itemize}

\section{FreezeML}
\label{section:freezeml}

In this section we summarise the syntax and typing rules of FreezeML.
(We omit the dynamic semantics, given by elaboration into System
F~\cite{EmrichLSCC20}, as it is not relevant to the current paper.)

\paragraph{Lists as sets.}
We write $\many{X}$ for a (possibly empty) set $\{X_1, \dots, X_n\}$ and
$\omany{X}$ for a (possibly empty) sequence $X_1, \dots, X_n$.
%
%
We overload comma for use as a union / concatenation operator for sets and
sequences, writing $\many{X}, \many{Y}$ for the set $\{X_1, \dots, X_m,
Y_1, \dots, Y_n\}$ where $\many{X} = \{X_1, \dots, X_m\}$ and $\many{Y}
= \{Y_1, \dots, Y_n\}$, and writing $\omany{X}, \omany{Y}$ for the sequence
$X_1, \dots, X_m, Y_1, \dots, Y_n$ where $\omany{X} = X_1, \dots, X_m$ and
$\many{Y} = Y_1, \dots, Y_n$.
%
%
Given $\omany{X}$, we may write $\many{X}$ for the set containing the same
elements.
We sometimes indicate that sets or sequences are required to be disjoint using
the $\disjoint$ relation, e.g. $\Delta \disjoint \Delta'$ means that $\Delta$
and $\Delta'$ are disjoint.
%



\paragraph{Types.}

The syntax of types, instantiations, and contexts is as follows.
\[
\ba[t]{@{}l@{}r@{~}c@{~}l@{}}
\slab{Type Variables}         &a, b, c \\
\slab{Type Constructors}      &D &::= &
                                        \mathord{\to}
                                   \mid \times
                                   \mid \Int
                                   \mid \dots \\
\slab{Types}                  &A, B      &::= &
                                         a
                                    \mid D\,\omany{A}
                                    \mid \forall a.A \\
\slab{Monotypes}              &S, T      &::= &
                                         a
                                    \mid D\,\omany{S} \\
\slab{Guarded Types}          &G, H &::= &
                                         a
                                    \mid D\,\omany{A} \\
\slab{Type Instantiation}     &\rsubst  &::= &
                                          \emptyset \mid \rsubst[a \mapsto A] \\
\slab{Type Contexts}      &\Delta, \Xi    &::= &
                                        \cdot \mid \Delta, a\\
\slab{Term Contexts}      &\Gamma    &::= &
                                          \cdot \mid \Gamma, x : A \\
\ea
\]
Types are assembled from type variables ($a, b, c$) and type constructors ($D$).
Type constructors include at least functions ($\to$), products ($\times$), and
base types.
FreezeML uses System F types ($A, B$), and the only syntactic form for expressing type polymorphism is $\forall a. A$.
A type is either a type variable ($a$), a data types ($D\,\omany{A}$) with type
constructor $D$ and type arguments $\omany{A}$, or a polymorphic type $(\forall
a.A$) that binds type variable $a$ in type $A$.
We consider types equal modulo alpha-renaming, but not up to reordering of
quantifiers or the addition/removal of superfluous (i.e., unused) quantified
variables.  For example, the following types are all different: $\forall
a. \forall b. a \to b$, $\forall b. \forall a. a \to b$, $\forall a. \forall
b. \forall c. a \to b$.
Monotypes ($S,T$) disallow any polymorphism.
Guarded types ($G,H$) disallow polymorphism at the top-level.
A type instantiation ($\rsubst$) maps type variables to types.
Unlike traditional presentations of ML, we explicitly track type variables in a
type context ($\Delta$).
By convention we reserve $\Xi$ for flexible type contexts which we will not need
until we treat constraints in Section~\ref{section:language}.
Term contexts ($\Gamma$) ascribe types to term variables.
Contexts are unordered and duplicates are disallowed.
As such, we will frequently take advantage of the fact that a type context
$\Delta$ is a set of type variables $\many{a}$ and use both notations
interchangeably.
This means that we impose the same disjointness conditions when writing
$\Delta, \Delta'$.


\paragraph{Typing judgements.}

FreezeML typing judgements have the form $\Delta; \Gamma \vdash M : A$, stating
that term $M$ has type $A$ in type context $\Delta$ and term context $\Gamma$.
We assume standard well-formedness judgements for types and term contexts:
$\Delta \vdash \wf{A}$ and $\Delta \vdash \wf{\Gamma}$, which state that only
type variables in $\Delta$ can appear in $A$ and $\Gamma$ respectively.
Moreover, the term well-formedness judgement $\Delta;\Gamma \vdash \wf{M}$
states that all free term variables of $M$ appear in $\Gamma$ and type
annotations are well-formed.
This judgement also implements the scoping rules of FreezeML, where
certain let bindings bring type variables in scope such that they become
available in type annotations~\cite{EmrichLSCC20}.
The scoping behaviour interacts with the value restriction adopted by FreezeML,
we therefore introduce $\Delta; \Gamma \vdash \wf{M}$ formally when discussing
let bindings later in this section.

The typing rules are given in Fig.~\ref{fig:freezeml-typing}.
As usual, in these rules we implicitly assume that types and term contexts are
well-formed with respect to the type context and that the term is well-formed
with respect to the type and term context (i.e., $\Delta;\Gamma \vdash \wf{M}$).
In the following running examples, we assume that the function $\dec{id}$ is in scope and has
type $\forall a. a \to a$.
\input{figures/freezeml-typing-rules.tex}

\paragraph{Variables and instantiation.}

A frozen variable ($\freeze{x}$) can only have the exact 
type as given by the term environment $\Gamma$ (rule \textsc{VarFrozen}). This means meaning that the
\emph{only} type of $\freeze{\dec{id}}$ is $\forall a. a \to a$.
In contrast, plain variables ($x$) can be instantiated, as in algorithmic
presentations of ML (rule \textsc{VarPlain}).
In fact, plain variables are the only terms in FreezeML that eliminate
polymorphic types.
This means that if we have $\Gamma(x) = \forall \omany{a}.H$, then the possible types of
$x$ are all results of instantiating all $\omany{a}$ in $H$, using arbitrarily
polymorphic types.
Potential nested quantifiers inside $H$ are not instantiated, however.
As a result, for any well-formed $B$, the type $B \to B$ is a possible type of
$\dec{id}$, whereas $\forall a. a \to a$ is not.

%
Formally, the \textsc{VarPlain} typing rule relies on an \emph{instantiation} $\delta$.
Each instantiation is parameterised by
a \emph{restriction}\footnote{\citet{EmrichLSCC20} called these \emph{kinds},
but we prefer to avoid potential confusion with other uses of this overloaded
term.} $R$ which can be either monomorphic ($\mono$) or polymorphic ($\poly$),
indicating whether type variables may be substituted with monotypes or arbitrary
types.
The instantiation judgement $\Delta \vdash \rsubst
: \Delta' \Rightarrow_R \Delta''$ states that instantiation $\rsubst$
instantiates type variables in $\Delta'$ with types subject to
restriction $R$ using the type context $\Delta, \Delta''$.
Variables in $\Delta$ are considered to be mapped to themselves.
This means that if $R$ is $\poly$, then for all $a \in \Delta'$, $\delta(a)$
must be a well-formed type in context $\Delta, \Delta''$, which may therefore be
arbitrarily polymorphic.
Otherwise, if $R$ is $\mono$, then each such $\delta(a)$ must be a monomorphic
type.
Note that due to all variables $b \in \Delta$ being mapped to themselves,
$\delta(b)$ is always a (monomorphic) well-formed type in context
$\Delta, \Delta''$ for all such $b$.

Intuitively, the variables in $\Delta$ correspond to those appearing in the surrounding context, whereas $\Delta'$ corresponds to variables being instantiated and $\Delta''$ contains new type variables appearing in the instantiation.  In \textsc{VarPlain}, the $\Delta''$ environment is empty, but in the $\meta{principal}$ operation discussed below, $\Delta''$ need not be empty.
In order for this interpretation to make sense the judgement has an implicit
precondition that
$\Delta' \# \Delta$ and $\Delta'' \# \Delta$.
It is defined as follows.

{
\vspace{8pt}
\raggedright
$\boxed{\Delta \vdash \rsubst : \Delta' \Rightarrow_R \Delta''}$
\begin{mathpar}
\inferrule
  { }
  {\Delta \vdash \emptyset : \cdot \Rightarrow_R \Delta'}

\inferrule
  {\Delta \vdash \rsubst : \Delta' \Rightarrow_R \Delta'' \\
   \Delta, \Delta'' \vdash_R \wf{A}}
  {\Delta \vdash \rsubst[a \mapsto A] : (\Delta', a) \Rightarrow_R \Delta''}
\end{mathpar}
}%
We write $\Delta \vdash_R \wf{A}$ for the well-formedness judgement for types.
It is standard except for the presence of $R$; if $R$ is $\mono$ then
$\Delta \vdash_\mono A$ only holds if $A$ is a monotype $S$.
In other words, the judgement enforces the restriction $R$, where any
well-formed type satisfies the restriction $\poly$.
Therefore, the well-formedness judgement $\Delta \vdash \wf{A}$ introduced
earlier is a shorthand for $\Delta \vdash_\poly \wf{A}$.

\paragraph{Functions.}
Function applications ($M\,N$) are standard and oblivious to polymorphism.
The parameter type $A$ of the function $M$ must exactly match that of the
argument $N$, where $A$ may be arbitrarily polymorphic.
In particular, $\freeze{id} \: 3$ is ill-typed because $\freeze{id}$'s type
$\forall a. a \to a$ is not a function type.
Conversely, $\dec{id} \: \freeze{\dec{id}}$ has type $\forall b. b \to b$.
The first occurrence of $\dec{id}$ is instantiated, by picking the type
$\forall b. b \to b$ of $\freeze{\dec{id}}$ for the quantified type variable.
This showcases the impredicative nature of FreezeML, with alpha-renaming
performed for the sake of clarity.

Plain (i.e., unannotated) lambda abstractions ($\lambda x.M$) restrict the
domain to be monomorphic.
This is a simple way to keep type inference tractable, in line with
other systems~\cite{Leijen08,SerranoHVJ18}.
Annotated lambda abstractions ($\lambda(x : A).M$) allow the domain
to be polymorphic, at the cost of a type annotation.
As a result, the example term $\lambda f. f \: f$ given in the introduction
is rejected in FreezeML, unless $f$ is annotated with an appropriate type.
Writing $\lambda (f : \forall a. a \to a). f \: f$ yields a function of type
$(\forall a. a \to a) \to (B \to B)$ for any well-formed $B$.
The return types of both forms of lambda abstractions may be arbitrarily
polymorphic: both $\lambda (f : \forall a. a \to a). f \: \freeze{f}$ and
$\lambda x . \freeze{\dec{id}}$ yield functions with polymorphic return types.

\paragraph{Principality}
The \textsc{LetPlain} rule has a \emph{principality} side condition that
requires that the type inferred for $x$ is a principal one.  Terms cannot
arbitrarily be generalised in FreezeML while retaining typability.
The term $\dec{id}$ has type $A \to A$ for any type $A$, and in particular
$a \to a$ for any type variable $a$.
However, it does not have type $\forall a. a \to a$.
As in System F, there is no direct relationship between the types
$\forall a. a \to a$ and $A \to A$ in FreezeML; generalisation only happens when a let-bound generalised value is encountered, and instantiation only happens if
triggered by a plain variable occurrence.

The fact that FreezeML typing judgements carry type contexts specifying all
in-scope type variables makes it possible to characterise principal types
without universally quantifying additional type variables.
Principal types are always given in their context $\Delta; \Gamma$ and may use
free type variables not present in $\Delta$.
For example, the principal types of term $\dec{id}$ in the context
$\Delta;\Gamma$ are exactly the types $b \to b$ for any $b \not\in \Delta$.
This presentation is syntax-directed, in contrast to declarative presentations of ML that allow generalisation
at any point, and would typically refer to the \emph{type scheme} $\forall
a. a \to a$ (not to be confused with the corresponding System F type) as the
principal type of $\dec{id}$.

We formalise the notion of principal type using the predicate
$\meta{principal}(\Delta, \Gamma, M, \Delta', A')$.
It first asserts that $A'$ is a possible type of $M$ in the context
$\Delta,\Delta'; \Gamma$, where $\Delta'$ includes all type variables of $A'$ not in
$\Delta$.
A minimality condition must also be satisfied:
 any other possible type of $M$ can be obtained by instantiating the
variables in $\Delta'$, possibly using some new type variables in context
$\Delta''$, but (crucially) \emph{not} instantiating any of those in $\Delta$.
Formally, a \emph{principal type} of $M$ in the context of $\Delta$ and $\Gamma$
is any pair $(\Delta',A')$ such that
$\meta{principal}(\Delta,\Gamma,M,\Delta',A')$ holds.  This is unique up to safe
renaming of the variables in $\Delta'$ (that is, avoiding already-known
variables in $\Delta$) and the occurrence of superfluous variables in $\Delta'$
(i.e. variables not actually occuring in $A'$).\footnote{We could impose
minimality of $\Delta'$ and $\Delta''$ in the definition of $\meta{principal}$,
but any superfluous variables in either context simply have no effect.}
Given a particular $A'$ we can recover $\Delta'$, so
we say \emph{the} principal type to refer to a particular $A'$ when $\Delta'$ is
clear from context.
\begin{mathpar}
\bl
\meta{principal}(\Delta, \Gamma, M, \Delta', A') = \\
\qquad
\Delta, \Delta'; \Gamma \vdash M : A' \;\text{ and } \\
\qquad\bl
(\text{for all }\Delta'', A'' \mid
       \bl
       \text{if }
       \Delta, \Delta''; \Gamma \vdash M : A'' \\
       \text{then there exists }
       \rsubst
       \text{ such that } \\
       \;\Delta  \vdash \rsubst : \Delta' \Rightarrow_\poly \Delta''
       \text{ and }
       \rsubst(A') = A'' )\\
       \el \\
\el
\el
\end{mathpar}
\begin{remark}[Abuse of notation]
Notice that the definition of $\mathsf{principal}$ refers to typing derivations in the ``if'' part of the condition.  The reader may be concerned about whether the typing judgement is well-defined 
given that it appears in a negative position in the definition of
$\meta{principal}$.  
As \citet{EmrichLSCC20} explain we can see that the definition is nevertheless
well founded by indexing by untyped terms or the height of derivation
trees. Likewise, proofs involving typing derivations are typically by induction
on $M$ rather than by rule induction.  Although it is a slight abuse of
notation, we prefer to present the typing rules using inference rule notation
for ease of comparison with other systems.  Formally however the rules in
Figure~\ref{fig:freezeml-typing} are implications that happen to hold of the
typing relation, not an inductive definition of it.  We refer to the extended
version of \citet{EmrichLSCC20} for the precise definition.
\end{remark}


\paragraph{Plain let bindings.}
\label{part:plain-let-terms}

Following ML, FreezeML adopts a syntactic value restriction~\cite{wright95lsc},
distinguishing two subcategories of terms.
\[
\ba[t]{@{}l@{}r@{~}c@{~}l@{}}
\slab{Values}                 &\dec{Val} \ni V, W       &::= &
                                        \freeze{x}
                                   \mid x
                                   \mid \lambda x.M
                                   \mid \lambda (x : A). M
                           \mid \Let\; x = V \;\In\; W
                                   \mid  \Let\; (x : A) = V \;\In\; W \\
\slab{Guarded Values}         &\dec{GVal} \ni U       &::= &
                                   \phantom{\freeze{x} \mid\;}
                                        x
                                   \mid \lambda x.M
                                   \mid \lambda (x : A). M
                           \mid \Let\; x = V \;\In\; U
                                   \mid  \Let\; (x : A) = V \;\In\; U
\ea
\]
Values disallow applications. Guarded values disallow frozen variables, and thus
must have guarded type.\footnote{The only guarded value with a top-level
polymorphic type is a plain variable $x$ of type $\forall a_1, \dotsc,
a_n. a_i$.
This special case is handled gracefully by FreezeML. }

Plain let bindings ($\Let\; x = M \;\In\; N$) generalise -- subject to the
value restriction -- the principal type $A'$ of $M$ and ascribe it to $x$.
Here, the predicate $\meta{principal}$ is used to determine
the type $A'$, using fresh variables $\omany{a}$.
Note that the free type variable operator $\ftv$ returns a sequence rather than
a set when applied to a type, returning variables in the order of their
appearance.
This reflects the fact that the order of quantifiers matters in FreezeML.

If $M$ is a guarded value, the type $A$ of $x$ is then $\forall \omany{a}.A'$,
performing the actual generalisation step.
This is achieved using the $\Updownarrow$ auxiliary judgement that enforces the value restriction, which we return
to shortly.

Generalising the principal type $A'$ rather than an arbitrary type of $M$ is
necessary to ensure the existence of principal types in the overall
system~\cite{EmrichLSCC20}.
First, recall that a principal type of an expression $M$ with
respect to $\Delta$ and $\Gamma$ is formally a \emph{pair} $(\Delta',A')$ such that
$\meta{principal}(\Delta,\Gamma,M,\Delta',A')$ hold.  Generalisation then quantifies the variables in
$\Delta'$ yielding a type $\forall a_1,\ldots,a_n.A$ where $a_1,\ldots,a_n$ are
in the order in which the variables first appear in $A'$.
Now consider the term
\[ \Let\; f = \lambda x.x \;\In\; \freeze{f}
\]
which is an example of FreezeML's ``explicit generalisation'' operation
$\$V \equiv \Let\; f = V \;\In\; \freeze{f}$ and allows capturing the
generalised principal type of a value.  The principal type of $\lambda x.x$ is
$a \to a$ (provided $a \not\in \Delta$), and by generalising this we
obtain $\forall a. a \to a$ as the type of $f$, which then becomes the type of
the overall let term due to $f$ being frozen in its body.  Note that the type
$\forall a. a \to a$ is \emph{not} the principal type of $\lambda x. x$ (formally: no pair $(\Delta'',\forall a. a \to a)$ is a principal type).
If the typing rule permitted assigning other, non-principal, types to $\lambda x.x$, such as
$\Int \to \Int$, then generalisation would have no effect.
This would make $\Int \to \Int$ another possible type of the overall let term
(as freezing a variable with a guarded or monomorphic type has no effect).
However, this would mean that the overall let term has no principal type.
The two types $\Int \to \Int$ and $\forall a. a \to a$ don't have a shared
more general type in FreezeML, as discussed earlier.

As mentioned before, the auxiliary judgement
$(\Delta, \omany{a}, M, A') \Updownarrow A$ enforces the value restriction.
Given $\Delta, \many{a} \vdash M : A'$, the judgement determines $A$ to be
$\forall \omany{a}.A'$ if $M$ is a guarded value.
Otherwise, $A$ is obtained from $A'$ by instantiating all of $\many{a}$
with \emph{monotypes}.

{
\vspace{8pt}
\raggedright
$\boxed{(\Delta, \omany{a}, M, A') \Updownarrow A}$
\begin{mathpar}
\inferrule
  {M \in \dec{GVal}}
  {(\Delta, \omany{a}, M, A') \Updownarrow \forall \omany{a}. A'}

\inferrule
  {\Delta' = \many{a} \\
   \Delta \vdash \rsubst : \Delta' \Rightarrow_\mono \cdot \\ M \not\in \dec{GVal}}
  {(\Delta, \omany{a}, M, A') \Updownarrow \rsubst(A')}
\end{mathpar}
}

As is well-known, type inference for System F is undecidable, even with nontrivial restrictions~\cite{Pfenning93,Wells94}.
The condition to instantiate monomorphically is one of several design choices in  FreezeML's to keep type inference decidable and tractable.
%
Along with the monomorphic restriction on the arguments to plain lambda
abstractions, FreezeML ensures that polymorphism can only ever appear in the term
context if it was written explicitly by a programmer in a type annotation or
inferred as a principal type of a plain let binding.

\paragraph{Annotated let bindings.}
\label{part:definition-of-split-function}

Annotated let bindings ($\Let\; (x : A) = M \;\In\; N$) also generalise, subject
to the value restriction, but ascribe the type $A$ to $x$.
The splitting operation $\meta{split}(A, M)$ enforces the value restriction for
annotated let terms.
It decomposes $A$ into a collection of top-level quantifiers and another type.
The first component of the returned pair is maximal if $M$ is a guarded value
and empty otherwise due to the value restriction.
\begin{mathpar}
\bl
\meta{split}(\forall \omany{a}.H, M) =
\left\{
\ba{@{~}l@{}@{\quad}l}
  (\many{a}, H)              & \text{if $M \in \dec{GVal}$} \\
  (\cdot,\forall \omany{a}.H) & \text{otherwise}
\ea
\right.
\el
\end{mathpar}
It is also important to note that in the generalising case (i.e. when the
let-bound expression is a guarded value $U$), the top-level quantifiers in type
annotations are in scope and can be used in $U$ (e.g. in other type
annotations). This is reflected in the $\meta{split}$ operation which returns
these variables in its first argument. In contrast, in the non-generalising case
where $M$ is not a guarded value, these variables are not in scope in $M$.
Since $M$'s type is not being generalised, the only way it can end up with a
polymorphic type is by referencing (frozen) variables with polymorphic types.

Note that this scoping behaviour also needs to be reflected in the term
well-formedness judgement $\Delta;\Gamma \vdash \wf{M}$ mentioned earlier.
To this end, the well-formedness rule for annotated let bindings also uses the
$\meta{split}$ operation, as shown in \cref{fig:term-wellformedness}.
The judgement $\Delta;\Gamma \vdash \wf{M}$ only requires the
presence of a binding for all free term variables, but ignores the associated
types.
As a result, the rules for unannotated lambda functions and let bindings add
arbitrary types $A$ to the the term context.

\begin{figure}[htp]
\raggedright
$\boxed{\Delta;\Gamma \vdash \wf{M}}$
\begin{mathpar}
  \inferrule
    {x \in \Gamma}
    {\Delta;\Gamma \vdash \wf{\freeze{x}}}

  \inferrule
    {x \in \Gamma}
    {\Delta;\Gamma \vdash \wf{x}}

  \inferrule
    {\Delta;(\Gamma, x : A) \vdash \wf{M}}
    {\Delta;\Gamma \vdash \wf{\lambda x.M}}

  \inferrule
    {\Delta \vdash \wf{A} \\\\
     \Delta;(\Gamma, x : A) \vdash \wf{M}
    }
    {\Delta;\Gamma \vdash \wf{\lambda (x : A).M }}

  \inferrule
    {\Delta;\Gamma \vdash \wf{M} \\\\
     \Delta;\Gamma \vdash \wf{N}
    }
    {\Delta;\Gamma \vdash \wf{M\,N}}

  \inferrule
    {\Delta;\Gamma \vdash \wf{M } \\
     \Delta;(\Gamma, x : A) \vdash \wf{N}
    }
    {\Delta;\Gamma \vdash \wf{\Let \; x = M\; \In \; N }}

  \inferrule
    {\Delta \vdash \wf{A}  \\
     (\Delta', A') = \msplit(A, M) \\
     (\Delta, \Delta');\Gamma \vdash \wf{M} \\
     \Delta;(\Gamma, x :A) \vdash \wf{N}
    }
    {\Delta;\Gamma \vdash \wf{\Let \; (x : A) = M\; \In \; N}}
\end{mathpar}
\caption{Well-formedness of terms.}
\label{fig:term-wellformedness}
\end{figure}




\section{Constraint language}
\label{section:language}
In this section, we present the constraint language and a function for generating
typing constraints from terms.
%
%
%
%
Following \citet{PottierR05}, our constraint language uses both term
variables and type variables.
Following \citet{EmrichLSCC20}, we distinguish rigid and flexible
type variables. The former arise in the object language from universal
quantification. The latter are used to represent unknown types.

The syntax and satisfiability judgement for constraints is given
in \cref{paper-fig:constraints-semantics}.
The judgement $\Delta; \Xi; \Gamma; \delta \vdash C$ states that in rigid type
context $\Delta$, flexible type context $\Xi$, term context $\Gamma$, using
instantiation $\rsubst$, the constraint $C$ is satisfied.
Note that rigid and flexible type contexts follow the same grammar, but we use
the convention that $\Delta$ is used for rigid variables, whereas $\Xi$ contains
flexible ones.
Therefore, using both environments in the judgement allows us to distinguish the
flexible variables in scope from those that are rigid.
In the judgement $\Delta; \Xi; \Gamma; \rsubst \vdash C$ we implicitly assume
that the term environment $\Gamma$ is well-formed and contains no flexible
variables ($\Delta \vdash \wf{\Gamma}$) and that type instantiations close over
the flexible type variables ($\Delta \vdash \rsubst
: \Xi \Rightarrow_\poly \cdot$).

To support composition of constraints we start with the always true constraint
($\meta{true}$) and conjunction ($C_1 \wedge C_2$).
The equality constraint $A \ceq B$ asserts that $A$ and $B$ are equivalent.
The frozen constraint $\freeze{x : A}$ asserts that $x$ has type $A$.
The instance constraint $x \preceq A$ asserts that top-level quantifiers of $x$'s
type can be instantiated to yield $A$.
The universal constraint $\forall a.C$ binds rigid type variable $a$ in $C$.
The existential constraint $\exists a.C$ binds flexible type
variable $a$ in $C$.
Monomorphism constraints $\monoc{}(a)$ assert that the flexible variable $a$
must only be instantiated with monotypes.
The definition constraint $\Def\; (x : A) \;\In\; C$ binds term variable $x$ in
$C$. (It also includes a side-constraint which we will return to shortly.)
The polymorphic let constraint $\Let_\poly\; x = \letexists{a} C_1 \;\In\; C_2$
and monomorphic let constraint $\Let_\mono\; x = \letexists{a} C_1 \;\In\; C_2$
are used to bind $x$ in $C_2$, subject to the restrictions imposed on $a$ in $C_1$.
The two forms differ in how the type of $x$ is obtained from solving $C_1$ for $a$:
either by generalisation $(\poly)$ or monomorphic instantiation $(\mono$).
These constraints are somewhat involved, so we defer a full explanation until we
present the constraint-generation function.

We consider constraints equivalent modulo alpha-renaming of all binders, of both
type and term variables.




\begin{figure}[tb!]
\[
\bl
\ba[t]{@{}l@{\quad}r@{~}c@{~}l@{}}
&C &::= &
                                        \dec{true}
                                   \mid C \wedge C
                                   \mid A \ceq B
                                   \mid \freeze{x : A}
                                   \mid x \preceq A
                                   \mid \forall a. C
                                   \mid \exists a.C
                                  \mid \monoc{}(a)
\\
                                 &&\mid&\Def\; (x : A) \;\In\; C
                                 \mid \Let_\poly\; x = \letexists{a} C \;\In\; C
                                   \mid \Let_\mono\; x = \letexists{a} C \;\In\; C\\
\ea \\
\el
\]
\begin{mathpar}

\inferrule[\lab{Sem-True}]
  { } {\Delta; \Xi; \Gamma; \rsubst \vdash \dec{true}}

\inferrule[\lab{Sem-And}]
  {
    \Delta; \Xi; \Gamma; \rsubst \vdash C_1 \\
    \Delta; \Xi; \Gamma; \rsubst \vdash C_2 \\
  }
  {\Delta; \Xi; \Gamma; \rsubst \vdash C_1 \wedge C_2}

\inferrule[\lab{Sem-Equiv}]
  {
    (\Delta, \Xi) \vdash_\poly \wf{A} \\
    (\Delta, \Xi) \vdash_\poly \wf{B} \\\\
    \rsubst(A) = \rsubst(B)
  }
  {\Delta; \Xi; \Gamma; \rsubst \vdash A \ceq B}
\\
\inferrule[\lab{Sem-Freeze}]
  {
    \Gamma(x) = \rsubst(A)
  }
  {\Delta; \Xi; \Gamma; \rsubst \vdash \freeze{x : A}}

\inferrule[\lab{Sem-Instance}]
  {
    \Delta' = \many{a} \\
    \Delta \vdash \rsubst' : \Delta' \Rightarrow_\poly \emptydelta \\\\
    \Gamma(x) = \forall \omany{a}.H \\
    \rsubst'(H) = \rsubst(A)
  }
  {\Delta; \Xi; \Gamma; \rsubst \vdash x \preceq A}
\\
\inferrule[\lab{Sem-Forall}]
  {
    (\Delta, a); \Xi; \Gamma; \rsubst \vdash C
  }
  {\Delta; \Xi; \Gamma; \rsubst \vdash \forall a.C}

\inferrule[\lab{Sem-Exists}]
  {
    \Delta; (\Xi, a ); \Gamma; \rsubst[a \mapsto A] \vdash C
  }
  {\Delta; \Xi; \Gamma; \rsubst \vdash \exists a.C}
\\
\inferrule[\lab{Sem-Mono}]
  {
    \Delta \vdash_\mono \wf{\rsubst(a)}
  }
  {\Delta; \Xi; \Gamma; \rsubst \vdash \monoc(a)}

  \inferrule[\lab{Sem-Def}]
  {
    \text{for all } a \in \ftv(A) - \Delta \;\mid\;  \Delta; \Xi; \Gamma; \rsubst \vdash \monoc(a) \\\\
    \Delta; \Xi; (\Gamma,x :  \rsubst A); \rsubst \vdash C
  }
  {\Delta; \Xi; \Gamma; \rsubst \vdash \Def\; (x : A) \;\In\; C}
\\
\inferrule[\lab{Sem-LetPoly}]
  {
    \meta{mostgen}(\Delta, (\Xi, a), \Gamma, C_1, \Delta_\mathrm{m}, \rsubst_\mathrm{m}) \\\\
    \Delta_\mathrm{o} = \ftv(\rsubst_\mathrm{m}(\Xi)) - \Delta \\
    \omany{b} = \ftv(\rsubst_\mathrm{m}(a)) - \Delta, \Delta_\mathrm{o} \\\\
    \Delta \vdash \rsubst' : \Delta_\mathrm{o} \Rightarrow_\mono \emptydelta \\
    A = \rsubst'(\rsubst_\mathrm{m}(a)) \\\\
    (\Delta, \many{b}); (\Xi, a); \Gamma; \rsubst[a \mapsto A] \vdash C_1 \\
    \Delta; \Xi; (\Gamma,x : \forall \omany{b}. A); \rsubst \vdash C_2
  }
  {\Delta; \Xi; \Gamma; \rsubst \vdash \Let_\poly\; x = \letexists{a} C_1 \;\In\; C_2}

\inferrule[\lab{Sem-LetMono}]
  {
    \meta{mostgen}(\Delta, (\Xi, a), \Gamma, C_1, \Delta_\mathrm{m}, \rsubst_\mathrm{m}) \\\\
    \Delta \vdash \rsubst' : \Delta_\mathrm{m} \Rightarrow_\mono \emptydelta \\
    A = \rsubst'(\rsubst_\mathrm{m}(a)) \\\\
    \Delta; (\Xi, a); \Gamma; \rsubst[a \mapsto A] \vdash C_1 \\
    \Delta; \Xi; (\Gamma,x : A); \rsubst \vdash C_2
  }
  {\Delta; \Xi; \Gamma; \rsubst \vdash \Let_\mono\; x = \letexists{a} C_1 \;\In\; C_2}
\end{mathpar}

\caption{Satisfiability judgement for constraints.}
\label{paper-fig:constraints-semantics}
\end{figure}

\subsection{Constraint generation}
\label{paragraph:term-to-constraint-translation}
We now introduce the function $\congen{M : A}$, which translates a term $M$ and type
$A$ to a constraint.
The only free type variables in the resulting constraint are those appearing in
$A$ and type annotations in $M$.
Assuming that $M$ is well-formed under $\Delta$ and $\Gamma$
($\Delta; \Gamma \vdash \wf{M}$) and that $A$ is well-formed under
$\Delta, \Xi$, the constraint $\congen{M : A}$ is well-formed under
$\Delta, \Xi$ and $\Gamma$ ($\Delta; \Xi; \Gamma \vdash \wf{\congen{M : A}}$).
The latter judgement is given in \cref{fig:constraint-wellformedness}.
Note that this judgement ignores the types in $\Gamma$ and uses it to track
bound term variables, just like the well-formedness judgement on terms
introduced in \cref{section:freezeml}.

\begin{figure}
\begin{mathpar}

\inferrule
  { }
  {\Delta; \Xi; \Gamma \vdash \wf{\dec{true}}}

\inferrule
  { a \in (\Delta, \Xi)}
  {\Delta; \Xi; \Gamma \vdash \wf{\monoc(a)}}

\inferrule
  {\Delta; \Xi; \Gamma \vdash \wf{C_1}\\
   \Delta; \Xi; \Gamma \vdash \wf{C_2}}
  {\Delta; \Xi; \Gamma \vdash \wf{C_1 \wedge C_2}}

\inferrule
  {\Delta; (\Xi, a); \Gamma \vdash \wf{C}}
  {\Delta; \Xi; \Gamma \vdash \wf{\exists a. C}}

\inferrule
  {(\Delta, a); \Xi; \Gamma \vdash \wf{C}}
  {\Delta; \Xi; \Gamma \vdash \wf{\forall a. C}}

\inferrule
  {(\Delta, \Xi) \vdash \wf{A}\\
   (\Delta, \Xi) \vdash \wf{B}}
  {\Delta; \Xi; \Gamma \vdash \wf{A \ceq B}}

\inferrule
  {x \in \Gamma\\
   (\Delta, \Xi) \vdash \wf{A}}
  {\Delta; \Xi; \Gamma \vdash \wf{x \preceq A}}

\inferrule
  {x \in \Gamma\\
   (\Delta, \Xi) \vdash \wf{A}}
  {\Delta; \Xi; \Gamma \vdash \wf{\freeze{x : A}}}

\inferrule
  {(\Delta, \Xi) \vdash \wf{A} \\
   \Delta; \Xi; (\Gamma, x : A) \vdash \wf{C}}
  {\Delta; \Xi; \Gamma \vdash \wf{\Def\; (x : A) \;\In\; C}}

\inferrule
  {\Delta;  (\Xi, a); \Gamma \vdash \wf{C_1}\\\\
   \Delta; \Xi; (\Gamma, x : A) \vdash \wf{C_2}}
  {\Delta; \Xi; \Gamma \vdash \wf{\Let_R \; x = \letexists{a} C_1 \;\In\; C_2}}

\end{mathpar}

\caption{Well-formedness of constraints.}
\label{fig:constraint-wellformedness}
\end{figure}

If $\Xi$ is empty (i.e., $A$ contains no flexible variables) then this
constraint is satisfiable in context $\Delta; \Gamma$ if and only if $M$ has type $A$
in context $\Delta; \Gamma$.
However, if $A$ does contain flexible variables, then the models of
$\congen{M : A}$ are exactly those that instantiate $A$ to valid types of $M$.
We formalise these properties in
\cref{section:constraint-generation-metatheory}.
Concretely, we perform type inference for $M$ by choosing $A$ to be a single
flexible variable.

%
%
%
\begin{figure}[htp]
\small
\renewcommand{\parlab}[1]{}


\[
\ba{@{}r@{~\;}c@{~\;}l@{}l@{}r@{}}
\congen{\freeze{x} : A}                     &=& \freeze{x : A} && \parlab{C-Freeze} \\
\congen{x : A}                              &=& x \preceq A && \parlab{C-Inst} \\
\congen{M\,N : A}                           &=&
 \exists a_1. (\congen{M : a_1 \to A} \wedge \congen{N : a_1} )&& \parlab{C-App} \\

\congen{\lambda x.M : A}                    &=&
 \exists a_1, a_2 . (a_1 \to a_2 \ceq A \wedge \Def\; (x : a_1)\; \In \;\congen{M : a_2}) & \parlab{C-Abs} \\

\congen{\lambda (x : B).M : A}              &=&
 \exists a_1. B \to a_1 \ceq A \wedge \Def\; (x: B) \;\In\; \congen{N : a_1} & \parlab{C-AbsAscribe} \\

\congen{\Let \; (x : \forall \omany{a}. H) = U\; \In \; N : A} &=&
 (\forall \omany{a} . \congen{U : H}) \;\wedge\; \Def\; (x: \forall \omany{a}. H) \;\In\; \congen{N : A}
 & \parlab{C-LetValAscribe} \\

\congen{\Let \; (x : B) = M\; \In \; N : A} &=&
  \congen{M : B} \;\wedge\; \Def\; (x: B) \;\In\; \congen{N : A}
 & \;\text{(if $M \not\in \dec{GVal}$)}
 & \parlab{C-LetAscribe} \\

\congen{\Let \; x = U\; \In \; N : A}       &=&
  \Let_\poly \; x = \letexists{a} \congen{U : a} \; \In \;\congen{N : A} && \parlab{C-LetVal} \\

\congen{\Let \; x = M\; \In \; N : A}       &=&
 \Let_\mono \; x = \letexists{a}\congen{M : a} \; \In \;\congen{N : A}
   \qquad\qquad\qquad\qquad&\; \text{(if $M \not\in \dec{GVal}$)}
   & \parlab{C-Let} \\

\ea
\]
\caption{Translation from terms to constraints.}
\label{fig:translation}
\end{figure}

The function $\congen{-}$ is defined in \cref{fig:translation}.

%



Frozen variables and plain variables generate the corresponding atomic constraints.
An application generates an existential constraint that binds a fresh flexible
type variable for the argument type.
A plain lambda abstraction generates a constraint that binds fresh flexible type
variables for argument and return types and uses a definition constraint to bind
the argument in the constraints generated for the body of the lambda
abstraction.
An annotated lambda abstraction generates a similar constraint to a plain lambda
abstraction, but the argument type is fixed by the type annotation.
The remaining four cases of $\congen{-}$ account for the four different
combinations arising from the two choices between plain or annotated and between
guarded value or not.
An annotated let binding $\Let\; (x : B) = M \;\In\; N$ generates a conjunction of
constraints: one for $M$ and the other for $N$.
Following the definition of $\meta{split}$ in the \lab{LetAnn} rule in
\cref{fig:freezeml-typing}, if $M$ is a guarded value $U$ then its type can be
generalised to obtain $B$ as witnessed by the universal constraints.
Notice in particular that the quantified type variables introduced in the
annotation are in scope in $U$ in the sub-constraint
$\forall \omany{a}.\congen{U:H}$.
Otherwise the types must match on the nose without any generalisation, and in
this case the quantified variables are not in scope in $M$.

\subsection{Def constraints}\label{sec:def-constraints}

The side condition in the \lab{Sem-Def} rule ensures that the argument type can
only be instantiated with a monomorphic type.
In general, the side condition preserves the invariant that no undetermined (or
``guessed'') polymorphism exists in the term context $\Gamma$.
This is crucial to ensure the existence of most general solutions for our
constraint language.
Consider the constraint $\Def\; (x : a) \;\In\; x \preceq b \wedge c \ceq (a \to
b)$ with free flexible variables $a,b,c$.
Without the extra condition on def constraints, different solutions could for
instance include $c \mapsto (\Int \to \Int)$ or $c \mapsto ((\forall
a. a) \to \Int)$ for $c$.
However, there is no more general solution subsuming both.
Note that the monomorphism condition on def constraints does not impose the
type annotation to be monomorphic itself, it only imposes conditions on free
flexible variables appearing within it.
Consequently, the constraint $\Def\; (x : (\forall a.a) \to b) \;\In\; \true$ is
satisfiable as long as the flexible variable $b$ is instantiated
monomorphically.
We cannot avoid the monomorphism condition imposed on def constraints simply
by using $\monoc$ constraints.
The constraint $\monoc(b){} \wedge \Def\;f (x : (\forall a.a) \to
b) \;\In\; \true$ would be equivalent to the previous one in terms of its
solutions, even if we dropped the monomorphism condition built into def
constraints.
However, this system would exhibit the same lack of most general solutions for
def constraints discussed earlier, showing that the monomorphism condition needs
to be imposed on def constraints directly.


%



\subsection{Let constraints}
Plain let bindings are translated to let constraints.
A plain let binding of a guarded value $\Let\; x = U \;\In\; N$ generates a
polymorphic let constraint.
In general, such a polymorphic let constraint
$\Let_{\poly}\; x = \letexists{a} C_1 \;\In\; C_2$ binds the flexible variable
$a$ in $C_1$, much like an existential constraint.
The type assigned to $x$ in $C_2$ is then obtained by generalising type
variables appearing in the solution for $a$.

We make several observations motivating the overall semantics of let
constraints.

\paragraph{Need to generalise principal solution.}
We first observe that let \emph{constraints} require a principality condition
similar to the one imposed on plain let \emph{terms}.
Consider the constraint $C$ defined as
$\Let\; x = \letexists{a} \exists b. a \ceq (b \to b) \;\In\; \freeze{x : c}$,
appearing in a rigid context $\Delta$. It has a single free type variable $c$
and we refer to its first subconstraint (i.e., $\exists b. a \ceq (b \to b)$) as
$C_1$ in the following.
Allowing arbitrary solutions for $a$ in $C_1$ to be generalised to yield the
type for $x$ would lead to the following pathological situation.

For any well-formed type $A$, $[a \mapsto (A \to A)]$ is a model of $C_1$.
As usual, we must not generalise any type variables already bound in the
surrounding scope, namely those variables in $\Delta$.
However, we may generalise fresh variables appearing in $A$.
This means that if we choose $[a \mapsto (\Int \to \Int)]$ there is nothing to
generalise and we have $x : (\Int \to \Int)$ in $C_2$ , whereas
$[a \mapsto (b' \to b')]$ for some fresh $b'$ does allow us to generalise, meaning
that we have $x : (\forall b'. b' \to b')$ in $C_2$.
Any solution of the overall constraint must use the type of $x$ for $c$.
This leads to a problem very similar to the one discussed for let terms in
\cref{part:plain-let-terms}: the two solutions $[c \mapsto (\Int \to \Int)]$ and
$[c \mapsto (\forall b'. b' \to b')]$ of $C$ would have no shared most general
solution in our system.
%

We avoid this problem by demanding that only the most general solution for $a$
in $C_1$ must be generalised to yield the type for $x$ in $C_2$.
In our example, this means choosing $[a \mapsto (b' \to b')]$, where $b'$ is
fresh, which means that in our example only
$[c \mapsto (\forall b'. b' \to b')]$ is a valid solution of the overall let
constraint.

The rule $\lab{Sem-LetPoly}$ in \cref{paper-fig:constraints-semantics} enforces
this using the premise
$\meta{mostgen}(\Delta, (\Xi, a), \Gamma, C_1, \Delta_\mathrm{m}, \rsubst_\mathrm{m})$,
which asserts that $\rsubst_\mathrm{m}$ is the most general model of $C_1$ in
the context $\Delta; \Xi; \Gamma$.
Here, $\Delta_\mathrm{m}$ contains fresh variables that are used in place of
flexible type variables for which no further substitution/solution is currently
known.
Note that this premise (and subsequently, $\delta_\mathrm{m}$) is independent
from the ambient instantiation $\rsubst$; the latter does not appear as an
argument.
The predicate $\meta{mostgen}$ is defined as follows, stating that
$\rsubst_\mathrm{m}$ is a model of $C_1$ and every other one can be obtained by
refining $\rsubst_\mathrm{m}$ by composition.
\[
\bl
\meta{mostgen}(\Delta, \Xi, \Gamma, C, \Delta_\mathrm{m}, \rsubst_\mathrm{m}) = \\
\qquad
(\Delta, \Delta_\mathrm{m}); \Xi; \Gamma; \rsubst_\mathrm{m} \vdash C \;\text{ and }\;\\
\qquad\bl
(\text{for all }\Delta'', \rsubst'' \mid
       \bl
       \text{if }
       (\Delta, \Delta''); \Xi; \Gamma; \rsubst'' \vdash C \\
       \text{then there exists }
       \rsubst'
       \text{ such that } \\
       \;\Delta  \vdash \rsubst' : \Delta_\mathrm{m} \Rightarrow_\poly \Delta''
       \text{ and }
       \rsubst' \comp \rsubst_\mathrm{m} = \rsubst'' )\\
       \el \\
\el
\el
\]

The rule $\lab{Sem-LetPoly}$ then defines two subsets\footnote{%
In general, $\Delta_\mathrm{m}$ may contain useless variables not appearing in
the codomain of $\rsubst_\mathrm{m}$. Otherwise, if all variables in
$\Delta_\mathrm{m}$ appear in the range of $\rsubst_\mathrm{m}$, then
$\Delta_\mathrm{m}$ and $\omany{b}$ denote a partitioning of
$\Delta_\mathrm{m}$.
} of $\Delta_\mathrm{m}$:
The variables in $\Delta_\mathrm{o}$ are those appearing in the range of
$\rsubst_\mathrm{m}$ restricted to $\Xi$ (i.e., not considering the mapping for
$a$ in $\rsubst_\mathrm{m}$).
This means that the variables in $\Delta_\mathrm{o}$ are related to the
\emph{outer} context, namely by being part of the instantiations of the
variables $\Xi$ in the surrounding scope.

The rule then determines the variables $\omany{b}$ to be generalised as the
flexible ones appearing freely in $\rsubst_\mathrm{m}(a)$ (i.e., the most
general solution for $a$) and disregarding the variables from
$\Delta_\mathrm{o}$, as the latter variables are related to the outer scope
$\Xi$.

\paragraph{Safe interaction with outer scope}

We have discussed that the rule $\lab{Sem-LetPoly}$ forces solutions for
constraints of the form $\Let_{\poly}\; x = \letexists{a} C_1 \;\In\; C_2$ to
use the most general solution for $a$ in $C_1$ -- using fresh rigid variables
$\Delta_\mathrm{m}$ -- and quantifying over
variables $\omany{b} \subseteq \Delta_\mathrm{m}$ to obtain the type for $x$ in
$C_2$.

We now show how flexible variables from the outer scope that appear in $C_1$ may
influence the type of $x$ and how we prevent this from introducing undetermined
polymorphism in the term context.
Consider the constraint $\exists a. \Let_{\poly}\; x = \letexists{b} a \ceq
b \;\In\; C_2$ appearing in rigid context $\Delta$.
The semantics of $\exists$ constraints (cf.\ \lab{Sem-Exists} in
\cref{paper-fig:constraints-semantics}) necessitates choosing a type $B$ for $a$ such
that $\Delta \vdash \wf{B}$.
The first subconstraint of the let constraint then equates $a$ and $b$, making
any kind of generalisation impossible when determining the type of $x$ (i.e.,
the type of $x$ is just $B$ without further quantification).
However, this means that the choice of $B$ influences the polymorphism of $x$,
meaning that the constraint above may introduce undetermined polymorphism in the
term context if arbitrarily polymorphic types were permitted for $a$.
Thus we must restrict the possible choices for $B$.
The rule $\lab{Sem-LetPoly}$ does so by imposing a relationship between the most
general solution $\rsubst_\mathrm{m}(a)$ for $a$ and the type $A$ actually
chosen as the instantiation of $a$.
In our example above, each most general solution $\rsubst_\mathrm{m}$ of $C_1$
has the form $[a \mapsto c, b \mapsto c]$, where $c \in \Delta_\mathrm{m}$.
Therefore, we have $\Delta_\mathrm{o} = c$ and $\omany{b}$ is empty.
The rule $\lab{Sem-LetPoly}$ then imposes that the actual type $A$ for $a$
results from monomorphically instantiating all non-generalisable variables in
$\rsubst_\mathrm{m}(a)$ (namely, $\Delta_\mathrm{o}$).
In the example above, this means that $a$ (and therefore also $b$) must be
instantiated with a monotype.
Observe that in general, $\delta'(\rsubst_\mathrm{m}(a))$ may not be a feasible
choice for $A$ for any well-formed monomorphic instantiation $\delta'$.
Consider the constraint:
\[
\exists a. a \ceq (\Int \to \Int) \wedge \left(\Let_{\poly}\; x = \letexists{b} \exists c. a \ceq b \wedge a \ceq (c \to c) \;\In\; C_2\right)
\]
Here, $\delta'(\rsubst_\mathrm{m}(a))$ may yield any type of the form $S \to S$
(recall that $S$ denotes monotypes).
However, the premise
$(\Delta, \many{b}); (\Xi, a); \Gamma; \rsubst[a \mapsto A] \vdash C_1$
of
$\lab{Sem-LetPoly}$
forces $A$ to be compatible with any prior choices made by the ambient instantiation $\rsubst$.
In our examples, this ensures that $A = (\Int \to \Int)$.

\paragraph{Monomorphic let constraints}

To accommodate the value restriction, the function $\congen{-}$ translates a
plain let binding of a term $M$ that is not a guarded value,
$\Let\; x = M \;\In\; N$, to a monomorphic let constraint of the form
$\Let_\mono\; x = \letexists{a} C_1 \;\In\; C_2$.
The only difference between a polymorphic let constraint and a monomorphic one
is that all variables that would be generalised by the former are instantiated
monomorphically by the latter.

The rule $\lab{Sem-LetMono}$ in \cref{paper-fig:constraints-semantics} achieves
this by instantiating all of $\Delta_\mathrm{m}$ monomorphically to obtain $A$
from $\rsubst_\mathrm{m}(a)$.
An equivalent, yet slightly more verbose version of $\lab{Sem-LetMono}$
highlighting the symmetry between $\lab{Sem-LetPoly}$ and $\lab{Sem-LetMono}$
could define $\Delta_\mathrm{o}$ and $\omany{b}$ just like the former rule, and then
impose
$\Delta \vdash \rsubst' : (\Delta_\mathrm{o}, \omany{b}) \Rightarrow_\mono \emptydelta$.
Observe that the variables in
$\Delta_\mathrm{m} - (\Delta_\mathrm{o}, \omany{b})$ are irrelevant in
$\lab{Sem-LetPoly}$.

\subsection{Metatheory}
\label{section:constraint-generation-metatheory}
We can now formalise the relationship between terms and the constraints obtained
from them.
Firstly, if $M$ has type $A$, then $\congen{M : a}$ is satisfiable by a
substitution that maps $a$ to $A$.
\input{thm-constraint-gen-sound}

Secondly, if a constraint $\congen{M : a}$ is satisfied using an instantiation
$\rsubst$, then $\rsubst(a)$ is a valid type for $M$.
\input{thm-constraint-gen-complete}

Both properties are proved by structural induction on $M$; proof details are provided in the appendix (see supplementary material).

\section{Constraint solving}
\label{section:solving}

We present a stack machine for solving constraints in our language, similar to
the \HMX solver by \citet{PottierR05}.
Our machine is defined in terms of a transition relation on states of the form
$(F, \Theta, \theta, C)$, consisting of a \emph{stack} $F$, a \emph{restriction
context} $\Theta$, a \emph{type substitution} $\theta$, and an \emph{in-progress
constraint} $C$, each of which we elaborate on below.

\paragraph{Stacks}
In a state $(F, \Theta, \theta, C)$, $C$ denotes the constraint to be solved next.

The stack $F$ denotes the context in which $C$ appears, containing bindings for
type variables (rigid and flexible) and term variables that may appear in $C$.
Further, the stack indicates how to continue after $C$ has been solved.
Our stack machine operates on closed states, meaning that the stack contains
bindings for all free variables of $C$.

Formally, stacks are built from stack frames as follows.
\[
\bl
\ba[t]{@{}l@{\quad}r@{~}c@{~}l@{}}
\slab{Frames}                 & f         &::= &
                                        \Box \wedge C
                                   \mid \forall\; a
                                   \mid \exists\; a
                                   \mid \Let_R \; x = \letexists{a} \Box \;\dec{in}\; C
                                   \mid \Def\; (x : A) \\
\slab{Stacks}                  &F         &::= &
                                        \emptystack
                                   \mid F :: f
\ea \\
\el
\]

The different forms of stack frames directly correspond to those constraints
with at least one sub-constraint.
The overall stack can then be seen as a constraint with a hole in which $C$ is
plugged.
We use holes $\Box$ in frames for constraints with two sub-constraints and store
the second sub-constraint to which we must return after solving the first one.

\paragraph{Restriction Contexts and Type Substitutions}

The components $\Theta$ and $\theta$ of a state $(F, \Theta, \theta, C)$ encode
the \emph{unification context}.
Their syntax is defined as follows.
%
\[
\bl
\ba[t]{@{}l@{\quad}r@{~}c@{~}l@{}}
\slab{Restriction Contexts}  &\Theta &::= & \cdot \mid \Theta, a : R \\
\slab{Type Substitutions}  &\theta &::= & \emptyset \mid \theta[a \mapsto A] \\
\slab{States}                 &s          &::= &
                                        (F, \kenv, \subst, C)
\ea \\
\el
\]

The restriction context $\Theta$ contains exactly the flexible variables bound
by the stack $F$ and stores the restriction imposed on each such variable.
Again, restrictions $R$ determine which types a flexible variable may be unified
or instantiated with: monomorphic only ($\mono$) or arbitrary polymorphic types
($\poly$).

Type substitutions $\theta$ are similar to type instantiations $\delta$.
However, they apply only to flexible variables, their codomain may contain
flexible variables, and must respect the restriction imposed on each individual
variable in the domain.
Note that this is in contrast to instantiations, where
$\Delta \vdash \delta : \Delta' \Rightarrow_R \Delta''$ fixes a single restriction $R$
for all variables in the domain of $\delta$.

To this end, we formalise what it means for a type $A$ to obey a restriction $R$
using the judgement $\Delta;\Theta \vdash_R \wf{A}$, shown in
\cref{fig:types-wf-and-substitution-wf}.
Rigid variables are monomorphic.
Flexible variables have their restriction determined by the restriction context.
The restriction of a data type is determined inductively.
A universally quantified type is polymorphic.
Every monomorphic type is also a polymorphic type.
Observe that the well-formedness judgement $\Delta \vdash_R \wf{A}$ used in
\cref{section:language} can now be considered as a shorthand for
$\Delta; \emptykenv \vdash_R \wf{A}$.

We can now formally state what it means for a substitution $\fsubst$ to be
well-formed, mapping flexible variables in $\Theta'$ to well-formed types over
variables from $\Delta, \Theta$ via the judgement $\Delta \vdash \fsubst
: \Theta' \Rightarrow \Theta$, which is also shown in
\cref{fig:types-wf-and-substitution-wf}.
As for substitutions, we additionally require that
$\Theta \disjoint \Delta$ and $ \Theta' \disjoint \Delta$ (however, $\Theta$ and $\Theta'$ need not be disjoint).
In summary, this means that in any solver state, the substitution $\theta$
contains the current knowledge about unification variables, respecting the
restrictions imposed by $\Theta$.

\begin{figure}[htp]
\raggedright
$\boxed{\Delta;\Theta \vdash_R \wf{A}}$
\begin{mathpar}
  \inferrule*
    {a \in \Delta}
    {\Delta;\Theta \vdash_\mono a }

\inferrule*
    {a : R \in \Theta}
    {\Delta;\Theta \vdash_R a }

  \inferrule*
    {\meta{arity}(D) = n \\\\
     \Delta;\Theta \vdash_R \wf{A_1} \\\\ \cdots \\\\ \Delta;\Theta \vdash_R \wf{A_n}}
    {\Delta;\Theta \vdash_R \wf{D \, \omany{A}}}

  \inferrule*
    {(\Delta, a); \Theta \vdash_R \wf{A}}
    {\Delta; \Theta \vdash_\poly \wf{\forall a.A} }

  \inferrule*
    {\Delta; \Theta \vdash_\mono \wf{A}}
    {\Delta; \Theta \vdash_\poly \wf{A}}
\end{mathpar}

\raggedright
$\boxed{\Delta \vdash \theta : \Theta' \Rightarrow \Theta}$
\begin{mathpar}
\inferrule
  { }
  {\Delta \vdash \emptyset : \cdot \Rightarrow \Theta}

\inferrule
  {\Delta \vdash \fsubst : \Theta' \Rightarrow \Theta \\
   \Delta;\Theta \vdash_R \wf{A}
  }
  {\Delta \vdash \fsubst[a \mapsto A] : (\Theta', a : R) \Rightarrow \Theta}
\end{mathpar}

\caption{Well-formedness of types and substitutions.}
\label{fig:types-wf-and-substitution-wf}
\end{figure}

We write $\bv(F)$ and $\btv(F)$ for the term variables and type variables
(flexible or rigid) bound by $F$, respectively.
Moreover, we write $\Deltaof{F}$, $\Xiof{F}$, and $\Gammaof{F}$ for the rigid context,
flexible context, and term context synthesised from a stack $F$, respectively.
The latter operators consider $\forall$ frames ($\Deltaof{F}$), let and
$\exists$ frames ($\Xiof{F}$), and def frames ($\Gammaof{F}$), as shown in
\cref{fig:extraction-operators}.

\begin{figure}[htb]
\[
\ba{cc}

\Deltaof{F} =
\begin{cases}
\emptydelta &\text{if } F = \emptystack \\[2pt]
\Deltaof{F'}, a &\text{if } F = F' :: \forall a \\
\Deltaof{F'}    &\text{otherwise } (F = F' :: \_)
\end{cases}

&\Xiof{F} =
\begin{cases}
\emptydelta &\text{if } F = \emptystack \\
\Xiof{F'}, a &\text{if } \begin{aligned}[t] &F = F' :: \exists a \;\text{ or } \\ &F = F' :: \Let_R\; x = \letexists{a} \Box \;\In\; C_2 \end{aligned}\\
\Xiof{F'}    &\text{otherwise } (F = F' :: \_)   \\
\end{cases} \\[3em]

\multicolumn{2}{c}{
\Gammaof{F} =
\begin{cases}
\emptygamma &\text{if } F = \emptystack \\
\Gammaof{F'}, (x : A)  &\text{if } F = F' :: \Def\; (x : A) \\
\Gammaof{F'}           &\text{otherwise } (F = F' :: \_)   \\
\end{cases}
}
\el
\]
\caption{Extracting components from stacks.}
\label{fig:extraction-operators}
\end{figure}


In order for a state $(F, \Theta, \theta, C)$ to be well-formed
($\vdash \wf{(F, \Theta, \theta, C)}$), we require that
$\Deltaof{F} \vdash \theta : \Theta \Rightarrow \Theta$, that $\theta$ is
idempotent, that $C$ is well-formed ($\Deltaof{F}; \Xiof{F}; \Gammaof{F} \vdash \wf{C}$), and that $F$ is well-formed with respect to
$\Theta$ ($\Theta \vdash \wf{F}$).
The latter judgement is defined in \cref{fig:stack-wellformedness}.
%
%
%
In addition to basic well-formedness conditions on the involved types and
constraints, the judgement $\Theta \vdash \wf{F}$ imposes the following invariants: all type and term
variables bound by $F$ are pairwise disjoint and all free type variables
appearing in annotations on def constraints are monomorphic.
Moreover, $\Theta$ must contain exactly the flexible variables bound by $F$.

\begin{remark}[Idempotent substitutions]
  Our requirement that $\theta$ be idempotent (i.e. $\theta \circ \theta = \theta$) concretely means that each binding $a \mapsto A$ in $\theta$ either maps $a$ to itself (an ``undetermined variable'') or to a type $A$ whose flexible variables are all undetermined.  This has some helpful consequences, for example in the definition of $\meta{partition}$, discussed later (Remark~\ref{rmk:partition}).
\end{remark}

To check the well-formedness of constraints embedded in stack frames, the
corresponding rules of $\Theta \vdash \wf{F}$ in \cref{fig:stack-wellformedness}
synthesise term contexts from the stack under consideration.
As with earlier well-formedness judgements, the judgement $\Theta \vdash \wf{F}$
checks that all term variables are in scope, but ignores the associated types.

\begin{figure}
\input{figures/stack-wellformedness}
\caption{Stack well-formedness.}
\label{fig:stack-wellformedness}
\end{figure}

\subsection{Stack Machine Rules}
\input{figures/stack-machine-rules-new.tex}

We now introduce the rules of the constraint solver itself
(\cref{paper-fig:constraints-stack-machine}).
These rules are deterministic in the sense that at most one
rule applies at any point.
Moreover, after each step the resulting state is unique up to the names of
binders added to the stack
%
%
and the order of adjacent existential frames in the frame (e.g., a step may
yield either $F :: \exists a :: \exists b :: \dots$ or $F :: \exists b
:: \exists a :: \dots$).
A constraint is satisfiable in a given context if the machine reaches a state of the
form $(\forall \Delta :: \exists \Xi, \Theta, \theta, \true)$ from an
initial configuration built from $C$ and the context.
From such a final configuration we can also read off a most general solution for
the constraint.
If the constraint is unsatisfiable, the machine gets stuck before reaching such
a final state.
We formalise the properties of the solver in
\cref{section:constraint-solving-metatheory}.
%
%

\paragraph{Unification}

The rule $\lab{S-Eq}$ in \cref{paper-fig:constraints-stack-machine} handles
equality constraints of the form $A \ceq B$.
We apply $\theta$ to both types as this may refine the types prior to invoking the
unification procedure $\substunifier$.
It remains unchanged as compared to the original type inference system for
FreezeML based on Algorithm W~\cite{EmrichLSCC20}.
\begin{figure}
\newcommand{\mlet}{\text{let}}
\newcommand{\mreturn}{\text{return}}
\newcommand{\massert}{\text{assert}}
\newcommand{\mfresh}{\text{assume fresh}}
\newcommand{\idsubst}{\theta_\mathrm{id}}
\renewcommand{\unify}{\substunifier}
\newenvironment{equations}{\begin{displaymath}\ba{@{}r@{~}c@{~}l@{}}}{\ea\end{displaymath}\ignorespacesafterend}
\raggedright

\[
\ba{c  c}

\bl
\unify(\Delta, \Theta, a, a ) = \\
\quad\mreturn\; (\Theta, \idsubst) \bigskip \\

\unify(\Delta, \Theta, D\,\omany{A}, D\,\omany{B}) = \\
\quad\bl
     \mlet\; (\Theta_1, \fsubst_1) = (\Theta, \idsubst) \\
     \mlet\; n = \dec{arity}(D) \\
     \text{for } i \in 1\ldots n \\
     \quad \mlet\; (\Theta_{i+1}, \fsubst_{i+1}) = \\
     \qquad\bl
           \mlet\; (\Theta', \fsubst') = \unify(\Delta, \Theta_i, \fsubst_i(A_i), \fsubst_i(B_i)) \\
           \mreturn\; (\Theta', \fsubst' \circ \fsubst_i) \\
           \el \\
     \mreturn\; (\Theta_{n+1}, \fsubst_{n+1}) \\
     \el 

\el

&\qquad\bl
\unify(\Delta, (\Theta, a : R), a, A) = \\
%
\unify(\Delta, (\Theta, a : R), A, a) = \\
\quad\bl
     \mlet\;\Theta_1 = \dec{demote}(R, \Theta, \ftv(A) - \Delta) \\
     \massert\;\Delta, \Theta_1 \vdash_R \wf{A} \\
     \mreturn\; (\Theta_1, \idsubst[a \mapsto A]) \\
     \el \bigskip\\

\unify(\Delta, \Theta, \forall a. A, \forall b. B)= \\
\quad\bl
     \mfresh\; c \\
     \mlet\; (\Theta_1, \fsubst') = \unify((\Delta, c), \Theta, A[c/a], B[c/b]) \\
     \massert\;c \notin \ftv(\fsubst') \\
     \mreturn\; (\Theta_1, \fsubst') \\
     \el
\el
\ea
\]

\begin{equations}
\dec{demote}(\poly, \Theta, \Delta)      &=& \Theta \\
\dec{demote}(\mono, \cdot, \Delta)      &=& \cdot \\
\dec{demote}(\mono, (\Theta,a:R), \Delta)       &=&
\dec{demote}(\mono,\Theta,\Delta),a:\mono  \quad (a \in \Delta)\\
\dec{demote}(\mono, (\Theta,a:R), \Delta) &=&
\dec{demote}(\mono, \Theta,\Delta),a:R \quad (a \not\in \Delta)  \\
\end{equations}
\caption{Unification algorithm.}
\label{fig:unifcation-algorithm}
\end{figure}
The unification algorithm is largely standard, supporting unification of
polymorphic types without reordering of quantifiers or the removal/addition of
unneeded quantifiers, as per FreezeML's notion of type equality.
It returns updated versions of the restriction context and substitution, named $\Theta'$ and
$\theta'$.
The algorithm is sound, complete, and yields most general unifiers~\cite[Theorem 4 and 5]{EmrichLSCC20}.

One notable feature is the unification algorithm's treatment of restrictions.
The restriction context $\Theta'$ returned by the algorithm contains the same
flexible variables as the original $\Theta$, but some variables therein may have
been demoted from a polymorphic restriction to a monomorphic one.
Unifying a flexible, monomorphic variable $a$ with a type $A$ only succeeds if
making all free flexible variables in $A$ monomorphic makes $A$ itself
monomorphic.
Therefore, assuming $a : \mono, b : \poly \in \Theta$, unifying $a$ with
$b \to b$ yields $(b : \mono) \in \Theta'$, whereas unification of $a$ with
$\forall c. c \to c)$ fails.

The unification algorithm is shown in \cref{fig:unifcation-algorithm}. On each
invocation, the first applicable clause is used; $\theta_\mathrm{id}$ denotes
the identity substitution on $\Theta$.

\paragraph{Basic constraints.}
The rules for constraints $\freeze{x : A}$ and $x \preceq A$ yield corresponding
equality constraints.
For instantiation constraints, the solver instantiates all top-level quantifiers $\omany{a}$
of $x$'s type by existentially quantifying them.
Note that the rule imposes picking variables $\omany{a}$ that are fresh with
respect to the bound type variables (rigid and flexible) of $F$.
In both rules, $\Gammaof{F}$ denotes the term context synthesised from all
$\Def$ constraints in $F$.

A monomorphism constraint $\monoc{}(a)$ is handled by demoting all flexible variables
in $\theta(a)$.
The step fails if doing so does not make $\theta(a)$ a monomorphic type.
Note that demoting the involved restrictions means that later unification steps
cannot make $a$ polymorphic --- even if, say, $\theta(a) = a$ holds at the time
of applying $\lab{S-Mono}$, recording $(a : \mono)$ in $\Theta'$ ensures that it
stays monomorphic.

The rules $\lab{S-ConjPush}$ and $\lab{S-ConjPop}$ handle conjunctions.
When encountering $C_1 \wedge C_2$, the first rule pushes a corresponding frame
on the stack.
Once $C_1$ is solved and the state's in-progress constraint becomes $\true$, the latter
rule pops this frame from the stack and continues solving $C_2$.

\paragraph{Binding of type variables.}

When encountering $\exists a. C$ or $\forall a. C$, a corresponding frame is
added by the rules $\lab{S-ExistsPush}$ and $\lab{S-ForallPush}$, respectively.

In general, once the in-progress constraint in a state becomes $\true$ and the
topmost stack frame binds a flexible variable $a$, the binding frame is either
moved downwards in the stack, generalised when handling $\Let$ constraints, or
dropped if no variable further down in the stack depends on $a$.
Note that lowering binders of flexible variables in the stack this way
effectively increases the syntactic scope of the bound variable as it moves
outwards in the constraint representing the stack.
The rule $\lab{S-ExistsLower}$ implements part of this lowering mechanism; it acts on stacks of the
form $F :: f :: \exists \omany{a}$, a shorthand for
$F :: f :: \exists a_0 :: \dots :: \exists a_n$.
The rule requires that $f$ is neither a let frame (whose rules treat adjacent
existentials directly) or another existential frame (to make the rule
deterministic by making $\omany{a}$ exhaustive).
It uses the helper function $\meta{partition}$, which returns a tuple and is defined
as follows.
\label{inline-def:partition}
\[
\dec{partition}(\Xi, \theta, \Theta) = \Xi' ; \Xi'' \text{ where }
\Xi', \Xi'' = \Xi \text{ and }
\text{for all } a \in \Xi:
  a \in \Xi'' \text{ iff } a \in \ftv(\restriction{\theta}{(\ftv(\Theta) - \Xi)})
\]
It partitions $\Xi$ into two sets $\Xi'$ and $\Xi''$ such that the latter
contains exactly those variables appearing in the range of $\theta$ restricted
to the flexible variables bound further down in the stack (i.e., $\theta$
restricted to $\ftv(\Theta) - \Xi$).
This formalises the notion that no variable further down in the stack depends on
a variable in $\Xi'$.
Therefore, the bindings for $\Xi'$ can simply be removed altogether; the
bindings for $\Xi''$ must be kept and are lowered within the stack.

\begin{remark}[Idempotence and $\meta{partition}$]\label{rmk:partition}
  If we did not require substitutions to be idempotent, the existing definitions
  would break in subtle ways (or need to be made more complicated).  For
  example,. let $a, b, c \in \Theta$ and $\Xi = \{a, b\}$ and
  $\theta = [a \mapsto a, b \mapsto a, c \mapsto b]$.  Note that this $\theta$
  is not idempotent since $\theta(\theta(c))) = a \neq b = \theta(c)$.  The
  definition of $\meta{partition}(\Xi, \theta, \Theta)$ would then yield
  $\Xi' = \{a\}$ and $\Xi'' = \{b\}$, even though $c$ depends on both variables.
\end{remark}

When popping a frame $\forall a$ from the stack, the rule $\lab{S-ForallPop}$
checks that $a$ does not escape its scope by still being present in the range of
$\theta$.
Note that together with the lowering of existentials mentioned before, removing
the unneeded variables $\many{b}$ in $\lab{S-ExistsLower}$ is not simply an
optimisation, but necessary for completeness.
Consider the state $(F :: \forall a :: \exists b, \Theta, \theta, \true)$ where
$\theta(b) = a$.
If this variable was lowered by $\lab{S-ExistsLower}$ -- yielding a new state
$(F :: \exists b :: \forall a, \Theta, \theta, \true)$ -- instead of being
removed, this would cause $\lab{S-ForallPop}$ to erroneously detect an escaping
quantifier.

\paragraph{Binding of term variables}

Similarly to the other $\lab{*Push}$ rules,
$\lab{S-DefPush}$ moves a constraint $\Def\; (x : A) \;\In\; C$ to the stack and
makes $C$ the next in-progress constraint.
However, it also forces all flexible variables found in $\theta(A)$ to be
monomorphic and checks that doing so does not make the substitution ill-formed.
This ill-formedness would arise if the substitution maps one of the type
variables to be monomorphised to a polymorphic type.
The monomorphisation is crucial in order to maintain the invariant that the term
context does not contain unknown polymorphism in the form of unrestricted (i.e.,
polymorphic) unification variables.
Note that the checks performed by \lab{S-DefPush} are equivalent to adding
$\bigwedge_{a \in \ftv(A) - \Deltaof{F} } \; \monoc(a)$ as a conjunct to $C$,
but doing so may create an ill-formed intermediate state before failing when
solving one of the $\monoc{}$ constraints.
The rule $\lab{S-DefPop}$ is the counterpart of $\lab{S-DefPush}$ and simply
pops the $\Def$ frame.

$\lab{S-LetPush}$ handles constraints
$\Let_{R}\; x = \letexists{b} C_{1} \;\In\; C_{2}$ by adding a stack frame and
bringing $b$ into scope while solving $C_{1}$.
Once $C_{1}$ has been solved, the rules $\lab{S-LetPolyPop}$ and
$\lab{S-LetMonoPop}$ handle the different semantics of $\Let_{\mono}$ and
$\Let_{\poly}$ regarding how they determine the type of $x$.
We first consider the former rule.
Note that the rule is applicable with zero or more existential frames on top of
the let frame, binding $\many{a}$, followed by the actual let frame.
These existential frames are either the result of existential constraints at the
top-level of the original constraint $C_1$ (the first subconstraint of the $\Let$
constraint under consideration), or were lowered while solving $C_{1}$.

Similarly to $\lab{S-ExistsLower}$, the variables $\many{a}$ and $b$ are
partitioned by the $\meta{partition}$ function into $\wmany{a'}$ and
$\wmany{a''}$.
Note that we include $b$ here because
$\Let_{R}\; x = \letexists{b} C_{1} \;\In\; C_{2}$ binds $b$ existentially in
$C_{1}$.
By definition of $\meta{partition}$, we again have that no unification variable
bound below the $\Let$ frame depends on any of the variables in $\wmany{a'}$ (as
indicated by $a$ not appearing in the image of $\theta$ restricted to the
variables bound in the lower frames).
Similarly to $\lab{S-ExistsLower}$, the variables $\wmany{a''}$ must be preserved
and are lowered in the stack.
Note that $\wmany{a''}$ may or may not contain $b$.

The type $A$ for $x$ is then determined by generalising $\theta(b)$.
The variables $\omany{c}$ to be generalised are obtained from taking those
free type variables of $\theta(b)$ that also appear in $\wmany{a'}$.
Recall that $\ftv$ applied to a type yields an ordered sequence, and we assume
that the ordering is preserved under intersection.

By definition, $\wmany{a'}$ contains those variables from $\many{a}, b$ that do
not appear in the codomain of $\theta$ restricted to the variables from lower
stack frames (i.e. no such variable from a lower stack frame directly depends
on a variable in $\wmany{a'}$).
Nevertheless, there may be variables $a \in \wmany{a'}$ such that
$\theta(a) = B$ where $B$ contains (or is) a variable from a lower frame (i.e.
from $\Theta - \many{a}, b$), meaning that $a$ must not be generalised.
However, due to the idempotency of $\theta$, we have that $\theta(a) = a$ for
all $a \in \theta(b)$.
In other words, intersecting $\wmany{a'}$ with the free type variables of
$\theta(b)$ evokes that $\omany{c}$ only contains ``undetermined'' variables
mapped to themselves that are not referenced by lower stack frames, either.

Note that rewriting the let frame to a def constraint also evokes that
solving the latter monomorphises any flexible variables in $A$ that were not
generalised.
This reflects the monomorphic instantiation imposed in the semantics of
$\Let_\poly$ constraints (cf.\ $\delta'$ in \lab{Sem-PolyPop} in
\cref{paper-fig:constraints-semantics}).

The only difference between $\lab{S-LetPolyPop}$ and $\lab{S-LetMonoPop}$ is
that while the latter rule also determines the variables $\many{c}$, it does not
generalise them.
This means that the resulting type $A$ assigned to $x$ contains the unification
variables $\many{c}$ freely.
Therefore, these variables must kept in scope and are existentially quantified
further down in stack after the rule is applied.

\subsection{Metatheory}
\label{section:constraint-solving-metatheory}

Our goal is to state a preservation property along the lines that stepping from state $s_0$
to $s_1$ implies that some representation of $s_0$ as a constraint is
equivalent to $s_1$'s constraint representation.
To this end, we first define how to represent the unification context, comprising
$\Theta$ and $\theta$, as a constraint.
Given $\Theta$ and $\theta$, we define:
\[
\ba{r@{\;\;} c@{\;\;} l}
\mathfrak{U}(\Theta) &= &\bigwedge_{(a\, : \,\mono) \in \Theta}\; \monoc(a)  \\[6pt]
\mathfrak{U}(\theta) &= &\bigwedge_{a \in \ftv(\Theta)}\; a \ceq \theta(a) \\[6pt]
\mathfrak{U}(\Theta, \theta) &= &\mathfrak{U}(\Theta) \wedge \mathfrak{U}(\theta)
\ea
\]

Using $\mathfrak{U}$, we may now represent a state $(F, \Theta, \theta, C)$ as
$F[C \wedge \mathfrak{U}(\Theta, \theta)]$, where the $F[-]$ operator plugs a
constraint into the stack's innermost hole:
{
\[
\bl
\ba[t]{@{}r@{\quad}c@{\quad}l@{\quad}r@{}}

\emptystack[C] &=& C \\

(F :: \Box \wedge C_2)[C_1] &=& F[C_1 \wedge C_2] \\

(F :: \forall\; a)[C] &=& F[\forall a. C] \\

(F :: \exists\; a)[C] &=& F[\exists a. C] \\

(F :: \Let_R\, x = \letexists{a} \Box \;\dec{in}\; C_2 )[C_1] &=&
F[\Let_R\, x = \letexists{a} C_1 \;\In\; C_2] \\

(F :: \Def\; (x : A) )[C] &=& F[\Def\; (x : A) \;\In\; C]

\ea
\el
\]
}

Note that if the state is closed (i.e., $F$ binds all variables free in $C$) the
resulting constraint $F[C]$ is closed, too.
In order to reason about constraints that are satisfied by non-empty
instantiations, we assume that there are some rigid and flexible contexts
$\Delta$ and $\Xi$ quantified by the bottom-most stack frames that remain
unchanged by the step.
Therefore, we consider the satisfiability of constraints before and after the
step by an instantiation $\rsubst$ with
$\Delta \vdash \rsubst : \Xi \Rightarrow_\poly \emptydelta$.

\input{thm-preservation}
This preservation property is inspired by a similar one holding for
\HMX~\cite[Lemma~10.6.9]{PottierR05}.

The following progress property states that given a well-formed, non-final state
whose representation as a formula is satisfiable, the stack machine can take a
step.
\input{thm-progress}

Termination is another crucial property.
\begin{restatable}[Termination]{thm}{termination}
\label{theorem:termination}
The constraint solver terminates on all inputs.
\end{restatable}
The proof relies on the existence of a well-ordering $<$ on states such that
$s \to s'$ implies $s' < s$.
We observe that the well-ordering cannot simply be defined based on the
syntactic size of the in-progress constraint of each state before and after the
step, even when plugging the constraint into each state's stack.
For example, the rule \lab{S-Inst} in \cref{paper-fig:constraints-stack-machine}
may introduce an arbitrary number of nested existential constraints. Other rules
such as $\lab{S-ExistsLower}$ may simply reorder stack frames.
Therefore, given a state $(F, \Theta, \theta, C)$, the well-ordering not only
takes the size of $F[C]$ and $C$ into account but also the number of
instantiation constraints in $C$ and the position of the right-most existential
frame in $F$.

We use the syntax $\Def\; \Gamma \;\In\; C$ to denote a series of nested def
constraints with $C$ as the innermost constraint, where each def constraints
performs a binding from $\Gamma$.
We now state the overall correctness of the solver as follows:
A constraint $C$ is satisfiable in context $\Delta; \Xi; \Gamma$ using
instantiation $\rsubst$ if and only if the solver reaches a final state from the
input constraint $\forall \Delta. \exists \Xi. \Def\; \Gamma \;\In\; C$ and
$\rsubst$ is a refinement of the substitution $\theta$ returned by the solver.
Here, a ``refinement'' of $\theta$ is simply a composition with $\theta$.

\input{thm-solver-overall-correctness}
%
Even though $\theta'$ acts like an instantiation (its codomain only contains
rigid variables), it is crucial for it to be a substitution, meaning that it
respects the individual restrictions in $\Theta$.
An instantiation $\rsubst'$ in place of $\theta'$ may violate the restrictions
in $\Theta$ and introduce polymorphism in places where the type system prohibits
it, which would make the right-to-left direction of the theorem invalid.
Also note the domain of $\theta$ is $(\Xi, \Xi')$, whereas that of $\rsubst$ is
$\Xi$.
Thus, we restrict $\theta$ to $\Xi$ when relating it to $\rsubst$.

Observe that \cref{theorem:solver-correct} also states that
our solver finds most general solutions:
The instantiation $\theta$ returned by the solver is independent from $\rsubst$.
Together with the deterministic nature of our solver, this means that any such
$\rsubst$ can be obtained from $\theta$.

For the purposes of type-checking, we may now relate the correctness of the
solver to constraints resulting from the translation function $\congen{-}$
introduced in \cref{paragraph:term-to-constraint-translation}.
If the solver succeeds on the translation of some term $M$ in some context
$\Delta; \Gamma$, then the term is well-typed in context $\Delta; \Gamma$ for
any well-formed refinement of $\theta(a)$, where $a$ is the placeholder variable
used for the type of $M$.

\input{thm-constraint-tc-sound}

Conversely, if $M$ has type $A$, then $A$ can be obtained from instantiating
$\theta(a)$.
\input{thm-constraint-tc-complete}

\begin{remark}[$\monoc{}(-)$ constraints]
  The constraint $\monoc{}(a)$ constrains a type variable to be instantiated only with a monotype.  Such constraints are not produced during constraint generation and appear in the constraint solving rules only in the rule $\lab{S-Mono}$ which solves them immediately by checking the current instantiation of $a$ and constraining the free variables of $a$ to be monomorphic.  Moreover, because $\Def\; (x : A) \;\In\; C$ constraints also require $A$ to be monomorphic, we considerd leaving out the $\monoc{}(a)$ constraint form altogether and considering it syntactic sugar for $\Def\; (x : a) \;\In\; \true$.  The main reason we have not done this is to simplify the presentation, particularly in the statement and proof of the above results, in order to ensure a simple translation of monomorphism information latent in $\Theta$ back into explicit constraints.

  An alternative design we considered would be to attach monomorphism restrictions to existential quantifiers, and drop the monomorphism requirement on  $\Def$ constraints.  However, we were not able to find a way to make this work that retains most general solutions for $\Def$ constraints, as illustrated in Section~\ref{sec:def-constraints}.
\end{remark}

\section{Discussion}
\label{sec:discussion}

In this section we discuss two extensions: using ranks for efficiency and
unordered quantification. We also compare our approach more directly to Pottier
and R\'emy's presentation of \HMX.

\subsection{Using Ranks}
\label{section:ranks}

In our solver, the lowering of existential frames in the stack as well as
generalisation are controlled by the free type variables in the image of the
substitution $\theta$ in the state under consideration.
Both mechanisms depend on the $\meta{partition}$ function.
A more efficient implementation associates a \emph{rank} with each unification
variable~\cite{RemyD92,KuanM07}, which can then be used instead of
checking what free type variables appear in certain types in the context.

Implementing ranks for our solver requires a similar mechanism to the one
described for the \HMX solver by \citet{PottierR05}; ranks are orthogonal to the
support for first-class polymorphism in our system.

\ifnewranksection
\newcommand{\atv}[1]{\ensuremath{\meta{atv}(#1)}}
\newcommand{\varrank}[1]{\ensuremath{\meta{rank}(#1)}}
\newcommand{\varindex}[1]{\ensuremath{\meta{index}(#1)}}

We briefly outline how ranks are related to our existing definitions and how the
solver can be adapted to use them.

To this end, we define the function $\meta{atv}$ that returns all type variables
(flexible and rigid) bound by a well-formed stack in their order of definition.
Recall that the well-formedness of stacks guarantees the uniqueness of all bound
variables.

\[
\meta{atv}(F) =
\begin{cases}
\emptydelta &\text{if } F = \emptystack \\
\Xiof{F'}, a &\text{if }
  \begin{alignedat}[t]{3}
    F &= F' :: \exists a &\text{or} \\
    F &= F' :: \Let_R\; x = \letexists{a} \Box \;\In\; C_2\;\; &\text{or} \\
    F &= F' :: \forall a
  \end{alignedat}\\
\Xiof{F'}    &\text{otherwise } (F = F' :: \_)   \\
\end{cases} \\[3em]
\]
Therefore, if $\meta{atv}(F) = (a_1, \dots, a_n)$, then $a_1$ is the outermost
type variable bound in $F$ and $\meta{atv}(F)$ is an interleaving of
$\Deltaof{F}$ and $\Xiof{F}$, as defined in \cref{fig:extraction-operators}.

Let a stack $F$ with $\atv{F} = a_1, \dots, a_n$, a variable
$b \in \meta{atv}(F)$, and a substitution $\theta$ with
$\Deltaof{F} \vdash \theta : \Theta \Rightarrow \Theta$ be given, where
$\ftv(\Theta) = \Xiof{F}$.

Adapting the work by \citet{KuanM07}, we may then define the \emph{index} of $b$
based on its location in $\atv{F}$ and its rank as the smallest index of
any variable $b$ such that $a$ appears in $\ftv(\theta(b))$, or $\infty$ if
$b \not\in \ftv(\theta)$.
\[
\ba{l@{} l@{} l@{}}
\varindex{b, F}\: &= \:& i \in \{1, n \} \text{ such that } a_i = b \\[1ex]
\varrank{b, \theta, F} &= &
  \begin{cases}
    \min_{i \in \{ 1, \dots, n \:\mid\: b \,\in\, \ftv(\theta(a_i))\}} i &\text{ if such an $i$ exists}\\
    \infty & \text{ otherwise}
  \end{cases}

\ea
\]
First, we observe that those variables $a$ with rank $\infty$ are those that are
\emph{eliminated} by $\theta$, meaning that $\theta(a) = A$ where $a \neq A$.
Note that our idempotency condition on substitutions then imposes
$a \not\in \ftv(A)$.

We can now relate ranks to the behavior of the partition function introduced in
\cref{inline-def:partition}.
We assume that there exists a stack $F = F_1 :: F_2$ such that
$\Theta \vdash \wf{F}$ and $\Xi = \Xiof{F_2}$, meaing that $\Theta$ contains
exactly the flexible variables bound by $F$ and $F_2$ is a suffix of $F$ binding
exactly the flexible variables in $\Xi$ (and possibly other rigid ones).
We then observe that the following, alternative definition of
$\meta{partition}$ yields the same results.
\[
\bl
\dec{partition'}(\Xi, \theta, F_1, F_2) = \Xi_\mathrm{g} ; \Xi_\mathrm{l} \text{ where } \\
\ba{l@{}l@{}l@{}l@{}}
\quad &\Xi_\mathrm{g} \: &= \: &\{ a \in \Xi \mid \varindex{F_1} + 1 \le \varrank{a, \theta, F}  \} \text{ and} \\
\quad &\Xi_\mathrm{l} &= &\{ a \in \Xi \mid \varrank{a, \theta, F} < \varindex{F_1} + 1 \}
\ea
\el
\]
We can use this alternative version directly in
\cref{paper-fig:constraints-stack-machine}, using $F_1 := F :: f$ and
$F_2 := \exists \omany{a}$ in $\lab{S-ExistsLower}$ as well as $F_1 := F$ and
$F_2 := :: \Let_R\; x = \letexists{b} \Box \:\In\: C :: \exists \omany{a}$ in
the $\lab{S-Let*Pop}$ rules.
Note that we are adding 1 to $\varindex{F_1}$ in the definition of
$\Xi_\mathrm{g}$ and $\Xi_\mathrm{l}$ to obtain the index in the combined stack
$F_1 :: F_2$ of first type variable bound in $F_2$.
Therefore, when used by $\lab{S-ExistsLower}$, we have that is the index of the
first variable $\varindex{F_1} + 1$ from the sequence $\exists{a}$ in the
overall stack.
When used by the $\lab{S-Let*Pop}$ rules, $\varindex{F_1} + 1$ is the index of
the variable $b$ bound by the let frame under consideration.
Using ranks, we can then perform the escape check for the universally quantified
variable $a$ in $\lab{S-ExistsPop}$ by verifying that
$\varrank{a, \theta, F :: \forall a} = \varindex{a, F}$ holds.

The first set returned by $\dec{partition}$ $\dec{partition'}$ and are those
variables that we \emph{may} generalise in $\lab{S-Let*Pop}$.
However, we must only generalise those variables actually appearing in the type
$A$, which is defined as $\theta(b)$ in the rule.
Further, the variables must be generalised in the order in which they appear in
$A$.

Using ranks, we can perform the first step directly within the partition
function, returning only variables appearing in $A = \theta(a)$, by changing the
partition function as follows:
\[
\bl
\dec{partition''}(\Xi, \theta, F_1, F_2) = \Xi_\mathrm{g} ; \Xi_\mathrm{l} \text{ where } \\
\ba{l@{}l@{}l@{}l@{}}
\quad &\Xi_\mathrm{g} \: &= \: &\{ a \in \Xi \mid \varindex{F_1} + 1 = \varrank{a, \theta, F} \} \text{ and} \\
\quad &\Xi_\mathrm{l} &= &\{ a \in \Xi \mid \varrank{a, \theta, F} < \varindex{F_1} + 1 \}
\ea
\el
\]
The only difference is that we require equality in the rank check for
$\Xi_\mathrm{g}$.
Note that in a solver for Unordered FreezeML, discussed in \cref{sec:unordered},
we may generalise the first set returned by this function directly.

\paragraph{Practical Concerns}

In this section, we have discussed how they can be used to define the
$\meta{partition}$ function we introduced in \cref{inline-def:partition} without
the need to inspect the codomain of $\theta$.
Of course, this requires storing the ranks of all type variables during
constraint solving and updating them during unification.
This would be achieved in the standard way in our system. If the unifier detects
that $A$ should be substituted for $a$, then the ranks of all variables in $A$
are lowered to the minimum of their existing rank and that of $a$.

Note that an alternative definition of rank-based generalisation may group
existentials together with the next enclosing $\forall$ or $\Let$ frame, if such
a frame exists.
The choice for defining ranks and indices in a more-fine grained matter was for
illustrative purposes to yield a drop-in replacement for $\meta{partition}$ that
may be used in the rules in \cref{paper-fig:constraints-stack-machine}.

Using the more traditional approach enumerating only $\Let$ and $\forall$ frames
would also be possible with little changes to the rules.
However, we may eschew $\exists$ frames altogether, as discussed in the next
section.

\else
We briefly outline how to adapt their mechanism to implement the escape check
our solver performs for $\forall$ quantifiers.
To this end we associate a rank with each flexible \emph{and} rigid variable
$b$ in a given state $s$, denoted $\meta{rank}(b)$.
We define $\meta{rank}(s)$ as the number of let frames plus the number of
$\forall$ frames appearing in $F$.
When the solver encounters a binder for type variable $b$ in state $s$, the variable's
rank is then initialised to be $\meta{rank}(s)$.
When the unifier detects that some flexible variable $a$ should be substituted
with some type $A$, then the ranks of all variables in $A$ are updated to be no
larger than $\meta{rank}(a)$. The escape check in the rule
\lab{S-ForallPop}, applied to state
$s = (F :: \forall b, \kenv, \subst, \true)$, can then be performed by checking that
$\meta{rank}(b) = \meta{rank}(s)$ holds.  To partition the variables in $\Xi$
using ranks in the rules \lab{S-LetPolyPop} and \lab{S-LetPolyPop} in
\cref{paper-fig:constraints-stack-machine} we use the following modified version
of the function $\meta{partition}$:
\[
\dec{partition'}(\Xi, s) = \Xi' ; \Xi'' \text{ where }
\Xi', \Xi'' = \Xi \text{ and }
(\text{for all } a \in \Xi \mid
  a \in \Xi'' \text{ iff } a \in \meta{rank}(a) < \meta{rank}(s))
\]

\fi

\paragraph{Eschewing existential frames.}
To avoid the need for (inefficiently) lowering existential frames in the stack
by swapping with one non-existential frame at a time (as for example in
\lab{S-ExistsLower}), we may optimise the solver further by not carrying
individual existential frames in the stack at all.
Instead, each state $s$ contains $\meta{rank}(s)$ sets of type variables of
that rank.
We may then remove the rule $\lab{S-ExistsLower}$ altogether; in
$\lab{S-Let*Pop}$ the variables $\wmany{a'}$ are determined by taking exactly
those of rank $\meta{rank}(s)$ (the set $\wmany{a''}$ isn't needed anymore in
this setting).

\subsection{Unordered FreezeML}\label{sec:unordered}

So far we have considered a syntactic equational theory on types that equates quantified types up to alpha-equivalence only and does not
allow for any reordering of quantifiers or the removal/addition of unused ones.
This is in line with the original presentation of FreezeML~\cite{EmrichLSCC20}.

However, this is not a fundamental requirement of the system.
We may define \emph{Unordered FreezeML}, a variant of FreezeML where
quantifiers are unordered, by redefining equality of types to allow
$\forall a b. a \to b = \forall b a. a \to b = \forall a b c . a \to b$ and
consider $\ftv$ to return sets of variables rather than sequences.
The typing rules of Unordered FreezeML can than be obtained from
\cref{fig:freezeml-typing} by replacing every occurrence of $\omany{a}$ by
$\many{a}$.
Likewise, type inference for Unordered FreezeML can be performed using a stack
machine using the same rules as shown in
\cref{paper-fig:constraints-stack-machine}. The only change is to replace the
unification algorithm $\substunifier$ with an alternative one ignoring the order
of quantifiers as well as unnecessary ones.  This is of course not trivial
because unification can no longer assume that when it encounters a $\forall$ on
one side of an equation, the other side must be a $\forall$ binding the same
(modulo alpha-equivalence) variable.  Unification must be modified to handle the
case where one side $\forall$-binds an unused variable (e.g.
$\forall a. int \ceq int$) or where bindings must be reordered (e.g.
$\forall a,b,c. a \to b \to c \ceq \forall a,b,c. c \to b \to a$).  The point is
that this complexity seems to be confined to unification and the rest of the
system is unchanged.

\subsection{Comparison with \HMX solver by \citeauthor{PottierR05}}
\label{section:hmx-comparison}

The solver presented in this section is inspired by the one presented by
\citet{PottierR05} for \HMX, adding support for first-class polymorphism in the
style of FreezeML.

Additional notable differences include the following:
\begin{itemize}
\item
In the \HMX solver, configurations carry a collection $U$ of unification
constraints. It can be interpreted as a conjunction built from a subset of the
constraint language with additional well-formedness restrictions.
This means that when extending the constraint language, the definition of
constraint permitted in $U$ can be adapted accordingly.

Our solver represents the unifier context with two separate
components $\Theta$ and $\theta$.
This is mostly for the purpose of making the system more similar to the original
type inference algorithm of FreezeML.
Our system already provides a mechanism for representing the unification context
as a constraint, in the form of $\mathfrak{U}(\Theta, \theta)$, defined in
\cref{section:constraint-solving-metatheory} for the purposes of our
meta-theory.
It would be straightforward to define unification contexts in our solver in
terms of $\mathfrak{U}(\Theta, \theta)$ (or a more structured representation
thereof using multi-equations) instead of $\Theta$ and $\theta$.

\item
The solver presented by \citeauthor{PottierR05} supports recursive types by
allowing the solver state to contain constraints of the form $a \ceq A$, where
$A \neq a$ and $a \in \ftv(A)$.
In our system, a corresponding state with $\theta(a) = A$ for the same $A$ would
be ill-formed, as we require $\theta$ to be idempotent.

We consider support for recursive types as orthogonal to the issue of supporting
first-class polymorphism, but our reliance on the idempotency of $\theta$ would
require adding explicit $\mu$ types for handling recursive types.
\item
The \HMX solver implements several optimisations, for example mechanisms to
reduce the number of type variables present in states.
Similar optimisations could be performed by our solver, but we eschew them for
the sake of brevity, including the usage of ranks described in
\cref{section:ranks}.

\item
In the \HMX solver, def constraints and term variables efficiently handle sharing the work of type inference for let-generalized values, but they are not strictly necessary: def constraints can be eliminated by a form of constraint inlining.  Doing so would not change the results of type inference but would be disastrous for performance.
In constrast, our system uses term variables as a means of keeping first-class
polymorphism tractable.  Term variables are used as placeholders for polymorphic types that we may instantiate, which is why we require that the polymorphism in these types is always fully determined.  Flexible type variables can also become bound to arbitrarily polymorphic types (e.g. during type inference for $id~\freeze{id}$) but the quantifiers occurring in these types are never instantiated
during constraint solving, they can only be unified with other
quantifiers.
%
By ensuring that all polymorphism in the term context is fully
known, we guarantee that we do not instantiate unknown polymorphism.  On the other hand, it does not appear possible to eliminate def constraints via substitution or to define let constraints in terms of def, as in Pottier and Rémy work.

\end{itemize}

\begin{remark}[Itches we haven't been able to scratch yet]
  An alternative system without
term variables in the constraint language should be possible, in which case instantiation of quantifiers occuring in types bound to type variables would need to take place. This would require an additional
mechanism to guarantee that the polymorphism of those type variables
that could be instantiated in this version is fixed. We haven't
adopted this alternative design as we consider the current design that
uses term variables to be closer to Pottier and Rémy's work and the
original version of FreezeML.

It also may be possible to recover the property that let and def constraints
canbe ``expanded away''.  One possibility (suggested by a reviewer) is to seek a
suitable adjustment of the notion of constraint inlining that accounts for
frozen constraints, perhaps by generalizing them to permit a local type
constraint ($\freeze{x:[C]A}$). As discussed in \ref{section:related}, allowing
for full constrained types in our setting poses challenges, and may also
encounter similar issues to those encountered in equational reasoning about
FreezeML terms, which was considered briefly by \citet{EmrichLSCC20}.  This
general issue deserves further investigation.

Further discussion of these two issues, which are interrelated, is in an appendix.
  \end{remark}

\section{Related Work}
\label{section:related}
Constraint-based type inference for Hindley-Milner and related systems has a
long history \cite{wand87fi}.  Some of the most relevant systems include qualified types~\cite{Jones94},
\HMX~\cite{OderskySW99}, OutsideIn($X$)~\cite{VytiniotisJSS11}, and
GI~\cite{SerranoHVJ18} which present increasingly sophisticated techniques for
solving (generalisations of) constraints generated from ML or Haskell-like
programs.  Our work differs in building on \HMX as presented by 
\citet{PottierR05}, while adapting it to support first-class polymorphism
based on the FreezeML approach.
On the other hand, constraint-based FreezeML does not
so far support constraint solving parameterised by an arbitrary constraint
domain $X$, and extending it to support this is a natural but nontrivial next step.
In particular, FreezeML uses exactly System F types, rather than the type
schemes with constraint components of the form $\forall \omany{a}. C \Rightarrow A$
found in \HMX.
We have compared our constraint solver to the \HMX solver
by \citeauthor{PottierR05} in \cref{section:hmx-comparison}.

FreezeML is also related to PolyML as explained by \citet{EmrichLSCC20}. Unlike
FreezeML, PolyML uses two different sorts of polymorphic types: ML-like type schemes
and first-class polymorphic types. The latter may only be introduced with
explicit type annotations. As a result, the conditions to pick most general
solutions in the semantics of certain constraints in our language are not
necessary in PolyML.

The type system of GI~\cite{SerranoHVJ18} uses carefully crafted rules for
$n$-ary function applications, determining when arguments' types may be generalised or
instantiated.
It does not perform let generalisation.
Its type inference system is built on constraint solving, using a different
approach towards restricting polymorphism. It syntactically distinguishes three
sorts of unification variables, which may only be instantiated with monomorphic,
guarded, or fully polymorphic types.
\fe{One could argue that our idea is better, because it's closer to expressing
these things as residual constraints (because $(a : \mono) \in \Theta$ is just a
$\monoc{a}$ residual constraint).
However, stating that here would lead to the question why we are using
$\Theta$ instead of a unification constraint $U$ in our solver.
On the other hand, their approach allows them to avoid the deeply nested/highly
structured constraints we need to use.
}
While our solver determines the order of constraint solving using a stack, their
system allows individual rules to be blocked until progress has been made
elsewhere, for example waiting until a fully polymorphic variable has been
substituted with a more concrete type.

QuickLook~\cite{SerranoHJV20} combines Hindley-Milner style type inference with
bidirectional type inference in a subtle way, and when typechecking applications
of polymorphically typed variables, performs a ``quick look'' at all of the
arguments; this amounts to a sound but shallow analysis whether there is a
unique type instantiation (possibly involving polymorphism). If there is a
unique type instantiation then that instantiation is chosen, otherwise
quantified variables are instantiated with monomorphic flexible variables.
Type inference for QuickLook follows a two-stage approach: all
first-class polymorphism is resolved at constraint generation time,
and the actual constraint solver does not have to find solutions for
polymorphic type variables. Consequently, their constraint language
and solver are completely standard and oblivious to first-class
polymorphism. Thus, QuickLook requires only small modifications to existing Haskell-style type
inference, including extensions such as qualified types and GADTs, but (like
other recent proposals such as GI) does not support let-bound polymorphism nor come with a
formal completeness result.  In
an appendix the authors discuss approaches to
supporting let generalisation; one is to use let constraints in the
style of Pottier and Rémy, and that is what we do.


Some aspects of our solver are reminiscent of the approach taken in Type Inference in Context by \citet{GundryMM10}, though their approach performs type inference as a traversal of source language terms rather than introducing an intermediate constraint language.  We are interested in adapting their approach to FreezeML type inference, particularly leveraging the insight that type inference monotonically increases knowledge about possible solutions (reflected in the structure of their contexts).

Returning to the motivation for this work mentioned in the introduction, it is a natural to ask what obstacles remain to generalizing our system to handle an arbitrary constraint domain (the ``$X$'' in \HMX).  The immediate obstacle is how to handle constrained or qualified types $C \Rightarrow A$ which are considered equivalent up to reordering constraints $C$.  Such types need not induce, and depending on the theory $X$ such types may have equivalent forms with different numbers of quantified types.  Adopting the unordered quantification approach in Section~\ref{sec:unordered} could help with this, but we leave this and the investigation of generalizing to a ``FreezeML$(X)$'' to future work.

\section{Conclusions}
\label{section:concl}

\citet{EmrichLSCC20} recently introduced FreezeML, a new approach to ML-style type
inference that supports the full power of System F polymorphism using type and
term annotations to control instantiation and generalisation of polymorphic
types. Their initial type inference algorithm was a straightforward extension of
Algorithm W.
%
%
We have introduced \sysname, an alternative constraint-based presentation of
FreezeML type inference, opening up many possibilities for extending FreezeML in
the future.
We extended the constraint language of \HMX with suitable constraints, equipped
with a semantics and translation from FreezeML programs to constraints that
encode type inference problems, and presented a deterministic, terminating state
machine for solving the constraints.
Several potential next steps are opened by this work, including generalising to
support arbitrary constraint domains (the ``$X$'' in \HMX), implementing the
solver efficiently using ranks,
and considering recursive types and higher kinds.

\begin{acks}
  This work was supported by ERC Consolidator Grant Skye (grant number 682315)
  and by an ISCF Metrology Fellowship grant provided by the UK government’s
  Department for Business, Energy and Industrial Strategy (BEIS).  Lindley is
  supported by UKRI Future Leaders Fellowship ``Effect Handler Oriented
  Programming'' (MR/T043830/1).
\end{acks}

\bibliography{bibliography}

\newpage
\appendix

\section{Proofs for Section~\ref{section:constraint-generation-metatheory}}
\label{appendix:constraint-lang-proofs}

The proofs of
\cref{theorem:constraint-generation-soundness,theorem:constraint-generation-completeness}
proceed via mutual induction on the structure of the term $M$.
Both proofs use the following lemma, but only on subterms of
the term $M$ in question.

\begin{restatable}{lem}{principaliffmostgen}
\label{lemma:aux:principal-iff-mostgen}
Let
$\Delta' = \ftv(A) - \Delta$
and $\Delta;\Gamma \vdash \wf{M}$.
%
Then
$\meta{principal}(\Delta, \Gamma, M, \Delta', A)$ iff
$\meta{mostgen}(\Delta,\allowbreak (a),\allowbreak \Gamma,\allowbreak \congen{M : a}, \Delta', [a \mapsto A])$.
\end{restatable}
The proof of \cref{lemma:aux:principal-iff-mostgen} in
turn uses both theorems directly.

We proceed by collecting auxiliary lemmas (including \cref{lemma:aux:principal-iff-mostgen})
in \cref{section:constraint-lang-aux-lemmas}.
The two subsequent subsections contain the proofs of
\cref{theorem:constraint-generation-soundness},
\cref{theorem:constraint-generation-completeness}, respectively.


\subsection{Auxiliary Lemmas}
\label{section:constraint-lang-aux-lemmas}

\begin{lem}
\label{lemma:principal-type-of-gval-is-guarded}
Let
  $M \in \dec{GVal}$ and
  $\meta{principal}(\Delta, \Gamma, M, \Delta', A)$.
Then $A$ is a guarded type.
\end{lem}
\begin{proof}
The only way for $A$ to be a top-level polymorphic type of a guarded value $M$ is
if $M$ is a plain (i.e, not frozen) variable $x$ of type
$\forall a_0, \dotsc, a_n. a_{i}$.
However, the \emph{principal} type of $x$ is $a$ for some fresh polymorphic
variable $a \in \Delta'$, which is a guarded type.
\end{proof}

\begin{lem}[Well-formedness of constraint translation]
\label{lemma:translation-yields-wf-constraint}
\fe{Wf-ness of constraints doesn't have the wf-ness of $\Gamma$ precondition anymore.
Therefore, not imposing wf-ness of $\Gamma$ here, beyond what variables it contains.}
Let $\Delta;\Gamma \vdash \wf{M}$ and
$(\Delta, \Xi) \vdash \wf{A}$.
Then
$\Delta; \Xi; \Gamma \vdash \wf{ \congen{M : A}}$ holds.
\end{lem}
\begin{proof}
By induction on structure of $M$.
We observe that the only free type variables of $\congen{M : A}$ are those
appearing freely in $A$ and in type annotations appearing in $M$.
By $\Delta;\Gamma \vdash M : A$ we have that all such free type variables in the
annotations in $M$ are rigid variables from $\Delta$.
Hence, the only free unification variables of $\congen{M : A}$ are those in $A$.

\end{proof}

\begin{lem}[Satisfiability implies well-formedness]
\label{lemma:sat-implies-wf-ness}
If $\Delta;\Xi;\Gamma;\rsubst \vdash C$ then $\Delta;\Xi;\Gamma \vdash \wf{C}$.
\end{lem}
\begin{proof}
We observe that the rules $\lab{Sem-Equiv}$, $\lab{Sem-Freeze}$, $\lab{Sem-Inst}$
$\lab{Sem-Def}$ all (explicitly or implicitly) require the types found in the constraint under consideration to be well-formed in the context $\Delta;\Xi$.
The rule $\lab{Sem-Mono}$ imposes $a \in (\Delta, \Xi)$.
Finally, the rules $\lab{Sem-Freeze}$, $\lab{Sem-Inst}$ require $x \in \Gamma$.
\end{proof}

See the note at the beginning of \cref{appendix:constraint-lang-proofs},
for an explanation of the dependencies between
\cref{theorem:constraint-generation-soundness},
\cref{theorem:constraint-generation-completeness} and
\cref{lemma:aux:principal-iff-mostgen}.

\input{proofs/principal-vs-mostgen.tex}

\subsection{Proof of \cref{theorem:constraint-generation-soundness}}

\input{proofs/constraint-generation-sound.tex}

\subsection{Proof of \cref{theorem:constraint-generation-completeness}}

\input{proofs/constraint-generation-complete.tex}

\section{Proofs for Section~\ref{section:constraint-solving-metatheory}}

We again proceed by collecting auxiliary lemmas in \cref{section:solver-aux-lemmas}
before proving each theorem from \cref{section:constraint-solving-metatheory} in
an individual subsection.

\subsection{Auxiliary Lemmas}
\label{section:solver-aux-lemmas}

The following lemma states how solutions for constraints $\Uof{\Theta, \theta}$
look like.
The lemma is somewhat specialised for the specific places where we use it, by
introducing an extra type context $\Delta'$ and an existential quantifier around
$\Uof{\Theta, \theta}$.
\begin{lem}
\label{lemma:satisfying-U}
Let
$\Delta \vdash \theta : \Theta \Rightarrow \Theta$
and
$\Delta, \Delta' \vdash \wf{\Gamma}$
and
$\ftv(\Theta) = \Xi, \many{a}$.

Then we have
\[
\ba{c}
(\Delta, \Delta'); \Xi; \Gamma; \rsubst \vdash \exists \many{a}. \Uof{\Theta, \theta} \\
\text{iff} \\
\text{there exists } \theta' \text{ such that } \Delta, \Delta' \vdash \theta' : \Theta \Rightarrow \emptydelta
  \text{ and } \rsubst = \restriction{(\theta' \comp \theta)}{\Xi}.
\ea
\]
\end{lem}
\begin{proof}
By definition, the subconstraint $\Uof{\theta}$ of $\Uof{\Theta, \theta}$
contains constraints of the form $a \ceq \theta(a)$ for all $a$ in $\Theta$.
These constraints are satisfied by exactly those substitutions $\rsubst$ refining $\theta$.
In addition, the $\mono{}$ subconstraints in $\Uof{\Theta}$ are satisfied iff the refinement
respects the restrictions in $\Theta$.
\end{proof}
\fe{commented out, would be good to have}


The next lemma states how \emph{most general} solutions of constraints
$\Uof{\Theta, \theta}$ look like.
\begin{lem}
\label{lemma:mostgen-iff-mapping-free-to-fresh}
Let the following conditions hold:
\begin{itemize}
\item $\Delta \vdash \theta : \Theta \Rightarrow \Theta$
\item $\ftv(\Theta) = \Xi, \many{a}$
\item $\Delta \vdash \wf{\Gamma}$
\item $\omany{b} \approx \ftv(\theta) - \Delta$
\end{itemize}

Then we have
\[
\ba{c}
\meta{mostgen}((\Delta, \Delta'), \Xi, \Gamma, \exists \many{a}. \mathfrak{U}(\Theta, \theta), \Delta_\mathrm{m}, \rsubst_\mathrm{m}) \\
\text{iff} \\
\text{there exists } \omany{c} \subseteq \Delta_\mathrm{m} \text{ s.t. }
\rsubst_\mathrm{m} = \restriction{([\omany{b} \mapsto \omany{c}] \comp \theta)}{\Xi}
\ea
\]
\end{lem}
\begin{proof}
Follows directly from \cref{lemma:satisfying-U} and the observation that
the most general solution of $\Uof{\Theta, \theta}$
is the one that maps all flexible variables in $\ftv(\theta)$ to fresh, pairwise disjoint rigid variables.
Observe that by assumption $\Delta \vdash \theta : \Theta \Rightarrow \Theta$
and the fact that rigid variables are considered monomorphic we have
$\Delta \vdash [\omany{b} \mapsto \omany{c}] \comp \theta : \Theta \Rightarrow \emptydelta$
as well.
\end{proof}

\begin{restatable}{lem}{pluggedconstraintsamestackirrelevantsat}
\label{lemma:aux:plugged-constraint-same-stack-irrelevant-sat}
Let $C_1$, $C_2$, $\Delta, \Xi, \rsubst$ and $F$ be given
and let the following condition hold:
for all
$\Delta', \rsubst'$ we have
$(\Delta, \Deltaof{F}, \Delta'); (\Xi, \Xiof{F}); \rsubst'(\Gammaof{F}); \rsubst' \vdash C_1$ iff
$(\Delta, \Deltaof{F}, \Delta'); (\Xi, \Xiof{F}); \rsubst'(\Gammaof{F}); \rsubst' \vdash C_2$.

Then we have
have $\Delta; \Xi; \emptygamma; \rsubst
\vdash F[C_1] \text{ iff } \Delta; \Xi; \emptygamma; \rsubst \vdash F[C_2]$.
\end{restatable}
\fe{%
Comment on \cref{lemma:aux:plugged-constraint-same-stack-irrelevant-sat}:
We have the $\Delta, \Deltaof{F}$ part here because of constraints of the form
$\Uof{\Theta, \theta}$. If the rigid variables are completely undetermined, the
most general solution of such a constraint differs based on how we choose the
rigid vars.
}
\begin{proof}
By induction on structure of $F$, observing that any judgement involving a
constraint with subconstraint $C_1$ can be replaced by a corresponding judgement
involving $C_2$ (and vice versa).
\end{proof}

\begin{lem}[Stack machine steps preserve well-formedness of states]
\label{lemma:stack-machine-steps-preserve-wf-ness}
If $\wf{s}$ and $s \to s'$ then $\wf{s'}$.
\end{lem}
\begin{proof}
By case analysis over which stack machine rule was applied.
\fe{Need to give more detail?}
\end{proof}

\begin{lem}[Weakening $\meta{mostgen}$]
\label{lemma:weakening-mostgen}
Let
$\meta{mostgen}(\Delta, \Xi, \Gamma, C, \Delta_\mathrm{m}, \rsubst_\mathrm{m})$
and $(\Delta, \Xi, \Delta_\mathrm{m}) \disjoint (\omany{a}, \omany{b})$.
Then we have
$\meta{mostgen}(\Delta, (\Xi, \many{a}), \Gamma, C,
  (\Delta_\mathrm{m}, \omany{b}), \rsubst_\mathrm{m}[\omany{a} \mapsto \omany{b}])$.
\end{lem}
\begin{proof}
By
$\meta{mostgen}(\Delta, \Xi, \Gamma, C, \Delta_\mathrm{m}, \rsubst_\mathrm{m})$
we have $(\Delta, \Delta_\mathrm{m}); \Xi; \Gamma; \rsubst_\mathrm{m} \vdash C$
and therefore $(\Delta, \Delta_\mathrm{m});\Xi;\Gamma \vdash \wf{C}$ (cf.\ \cref{lemma:sat-implies-wf-ness}).
This means that none of the variables in $\omany{a}$ are constrained in any way
by $C$, meaning that the most general solution maps them to pairwise disjoint
variables, like $\omany{b}$.
\end{proof}

\begin{lem}[Refinement]
\label{lemma:refinement-on-constraint-sat}
Let $\Xi \disjoint \Delta'$.
If $\Delta; \Xi; \Gamma; \rsubst \vdash C$ and
$\Delta' \vdash \rsubst' : \Delta \Rightarrow_\mono \emptydelta$ then
$\Delta'; \Xi; \rsubst'(\Gamma); \rsubst' \comp \rsubst \vdash C$.
\end{lem}
\begin{proof}
By structural induction on $C$, observing that for all
$a \in \Xi, R \in \{\mono, \poly \}$ we have $\Delta \vdash_R \wf{\rsubst(a)}$
iff $\Delta' \vdash_R \wf{\rsubst'(\rsubst(a))}$.
\end{proof}

\begin{lem}[Substitution]
\label{lemma:substitution-on-constraint}
If $\Delta; (\Xi, a); \Gamma; \rsubst[a \mapsto A] \vdash C$ then
$\Delta; \Xi; \Gamma; \rsubst \vdash C[a / A]$.
\end{lem}
\begin{proof}
By structural induction on $C$.
Observe that the substitution does not interfere with generalisation: The
variable $a$ is already in scope and therefore not subject to generalisation by
any let constraint within $C$.
\end{proof}

\subsection{Proof of \cref{theorem:preservation}}
\input{proofs/preservation}

\subsection{Proof of \cref{theorem:progress}}
\input{proofs/progress}

\subsection{Proof of \cref{theorem:termination}}

\begin{restatable}[Well-Ordering on States]{lem}{wellorderingonstates}
\label{lemma:well-ordering-on-states}
There exists a strict well-ordering $<$ on the set $\dec{St}$ of stack machine states
such that for all $s, s' \in \dec{St}$ with $s \to s'$ we have $s' < s$.
\end{restatable}
\begin{proof}
First, we define the size of a constraint $C$, denoted $|C|$, s.t.\
\[
\ba{r c l}
|\true| &= &0 \\
|\monoc(a)| &= &1 \\
|A \ceq B| &= &1 \\
|\freeze{x : A}| &= &2 \\
|x \preceq A| &= &2 \\
|\exists a. C| &= &1 + |C| \\
|\forall a. C| &= &1 + |C| \\
|\Def\; (x : A) \;\In\; C| &= &1 + |C| \\
|C_1 \wedge C_2| &= &1 + |C_1| + |C_2| \\
|\Let_R\; x = \letexists{a} C_1 \;\In\; C_2| &= &3 + |C_1| + |C_2| \\
\ea
\]
Next, we define $\meta{insts}(C)$ to be the number of instantiation (sub-)constraints in $C$.

We now define the function $|\cdot|$ that maps  states to elements of
$\N_0 \times \N_0 \times \N_0 \times \N_0$:
\[
\ba{r c l}
|(f_{0} :: \dots :: f_{n}, \Theta, \theta, C)| =
(
\meta{insts}(C), \:
|F[C]|,\:
|C|,\:
\max \{i  \mid 0 \leq i \leq n, \text{ $f_i$ is an $\exists$ frame }\}
)
\ea
\]

We observe that the lexicographic ordering $<_{\mathit{lex}}$ on tuples from
$\N_0 \times \N_0 \times \N_0 \times \N_0$ constitutes a well-ordering on such
tuples and we will show below that for each step $s \to s'$ we have that
$|s'| <_{\mathit{lex}} |s|$ holds.
However, the function $|\cdot|$ on states is surjective, which implies that defining
$s' < s$ iff $|s'| <_{\mathit{lex}} |s|$ would \emph{not} yield a total order on
states.
Hence, let $<_{\mathit{ord}}$ be some arbitrary strict well-order on states.
We then define
\[
s' < s \;\;\text{iff}\;\; |s'| <_{\mathit{lex}} |s| \text{ or } |s'| = |s| \text{ and }
s' <_{\mathit{ord}} s
\]
which is indeed a well-ordering.

It remains to show that each step of the stack machine produces a smaller state w.r.t.\ $|\cdot|$.
Hence, assume $s \to s'$, where
$s = (F, \Theta, \theta, C)$ and $s' = (F', \Theta', \theta', C')$.

\begin{itemize}
\item
If the step is the result of applying the rule
\lab{S-Eq},
\lab{S-Freeze},
or \lab{S-Mono},
then we have $F = F'$ and $|C'| < |C|$, yielding $|s'| <_\mathit{lex} |s|$ via the second component of the tuples.

\item
If the step is the result of applying $\lab{S-Inst}$, we have
$\meta{insts}(C') = \meta{insts}(C) - 1$ and
we have $|s'| <_\mathit{lex} |s|$ via the first component of the tuples.

\item
If the step is the result of applying the rule
\lab{S-ConjPop},
\lab{S-ForallPop}, or
\lab{S-DefPop},
we have $\meta{insts}(C) = \meta{insts}(C')$ and $|F'[C']| < |F[C]|$, yielding $|s'| <_\mathit{lex} |s|$
via the second component of the tuple.

\item
If the step is the result of applying the rule
\lab{S-ConjPush},
\lab{S-ExistsPush},
\lab{S-ForallPush},
\lab{S-DefPush}, or
\lab{S-LetPush},
we have $\meta{insts}(C) = \meta{insts}(C')$ and $F[C] = F'[C']$, but $|C'| < |C|$, yielding $|s'| <_\mathit{lex} |s|$ via the third component of the tuples.

\item
If \lab{S-ExistsLower} got applied,
let
$\many{c}$ and $\many{a}$ be defined as in the rule,
$F = f_0 :: \dots f_n$, and $l = |\many{a}|$.
Note that the rule imposes $l > 0$ and we have $C = C' = \true$.

We observe that $\many{c}$ is a subset of $\many{a}$.
If $|\many{a}| > |\many{c}|$ , we have $|F'[C']| < |F[C]|$ because
the set of frames of $F'$ is a strict subset of the frames in $F$.
We then obtain $|s'| <_\mathit{lex} |s|$ immediately via the second component of the tuples
(the first component remains unchanged).
Otherwise, we have that the two sets are equal.  In that case we have
$|F[C]| = |F'[C']|$ (as there is merely a reordering of stack frames happening) and
$|C| = |C'| = 0$.
The resulting stack $F'$ is of the form $f'_0 :: \dots f'_n $,
where $f'_n = f_{n - l}$ (not an $\exists$ frame), and $f'_{n - 1}$ is an
$\exists$ frame.

Together, we have $|s| = (\meta{insts}(C), |F[C]|, 0, n)$, and
$|s'| = (\meta{insts}(C'), |F[C]|, 0, n - 1)$, and
$\meta{insts}(C) = \meta{insts}(C')$, yielding $|s'| <_\mathit{lex} |s|$.

\item
If rule \lab{S-LetPolyPop} was applied, we use the following reasoning:
$F$ is of the form
$F_0 :: \dec{let}_\poly \, x = \letexists{c} \Box \,\dec{in}\, \hat{C} :: \exists \many{a}$
and $C$ is $\true$.
This yields
\[
\ba{rcl}
|F[C]| &= &|F_0[\Let_\poly\; x = \letexists{c} \exists \many{a}. \true \;\In\; \hat{C}]| \\
&= &|\Let_\poly\; x = \letexists{c} \true \;\In\; \true| + |\exists{\many{a}. \true}| + |\hat{C}| + |F_{0}[\true]| \\
&= &3 + |\many{a}| + |\hat{C}| + |F_0[\true]|,
\ea
\]

Let $\wmany{a''}$ be defined as in the rule.
We have
\[
  |F'[C']| = |F_0[\exists \wmany{a''}. \Def\; (x : A) \;\In\; \hat{C} ]
= 1 + |\wmany{a''}| + |\hat{C}| + |F_0[\true]|
\]
Proving $F'[C'] < F[C]$ is therefore equivalent to proving
\[
\ba{c l}
&1 + |\wmany{a''}| + |\hat{C}| + |F_0[\true]| < 3 + |\many{a}| + |\hat{C}| + |F_0[\true]| \\
\text{equiv. } &\wmany{a''} < 2 + \many{a}.
\ea
\]
We observe that $\wmany{a''}$ is a strict subset of $(\many{a}, c)$ and hence
$|\wmany{a''}| < |\many{a}| + 2$, meaning that the inequality above holds.

We have $\meta{insts}(C) = \meta{insts}(C')$, and therefore $|s'| <_\mathit{lex} |s|$ via the
second component of the tuples.

\item
The reasoning for rule \lab{S-LetMonoPop} is analogous to the previous case. The
only change is that we need to observe that $(\many{c}, \wmany{a''})$ is a subset
of $(\many{a}, c)$.

\end{itemize}

\end{proof}

\termination*
\begin{proof}
Follows immediately from \cref{lemma:well-ordering-on-states}, which guarantees
the absence of infinite sequences of steps.
\end{proof}

\subsection{Proof of \cref{theorem:solver-correct}}

The following lemma is a slight variation of
\cref{theorem:solver-correct}; we use it in the proof of the
latter.
\begin{lem}
\label{lemma:helper-constraints-vs-solver}
Let $\wf{(\forall \Delta \snoc{} \exists \many{a} \snoc{} F, \Theta, \theta, C)}$. Then we have

\[
\bl
\Delta; \many{a} ; \emptygamma; \rsubst
\vdash F[C \wedge \Uof{\Theta,\theta}]   \\
\text{iff} \\
\text{there exist } \kenv', \theta'', \theta', \many{b} \text{ s.t.\ } \\
\ba{cl}
 \;&(\forall \Delta \snoc{} \exists \many{a}  \snoc{} F, \Theta, \theta, C) \to^{*}
    (\forall \Delta :: \exists\; (\many{a},\many{b}), \kenv', \subst', \true) \text{ and }  \\
  &\Delta \vdash \theta'' : \Theta' \Rightarrow \emptydelta \text{ and }  \\
  &\restriction{(\theta'' \comp \theta')}{\many{a}} = \rsubst
\ea
\el
\]

\end{lem}
\begin{proof}
\proofContext{helper-constraints-vs-solver}

We show each direction individually:
\begin{itemize}
\item[$\Longrightarrow$]
\proofContext{constraints-solvable-subst-no-rank-left-to-right}
By transfinite induction on the well-ordering $<$ on stack machine states $s$
whose existence is shown in \cref{lemma:well-ordering-on-states}.
Hence, we assume that the left-to-right direction of the lemma holds for all
$s'$ on the left of the $\to^{*}$ s.t.\ $s' < s$ and show that the left-to-right
direction holds for $s$ on the left of $\to^{*}$, too.

\fe{
If we want to we could spell out the induction hypothesis and principle more clearly:\\
More formally, let $\Delta$ and $\many{a}$ be fixed. Then we define $\meta{Prop}(\hat{s})$  as follows:
\begin{quote}
For all $F, \Theta, \theta, C, \rsubst$: If $\hat{s} = (\forall \Delta \snoc{} \exists \many{a} \snoc{} F, \Theta, \theta, C)$ and
$\wf{\hat{s}}$ and $\Delta; \many{a}; \emptygamma; \rsubst
\vdash F[C \wedge \Uof{\Theta, \theta}]$ then there exist $\kenv', \theta', \theta'', \many{b}$ s.t.\
\[
\ba{cl}
 \;&(\forall \Delta \snoc{} \exists \many{a}  \snoc{} F, \Theta, \theta,  C) \to^{*}
    (\forall \Delta :: \exists\; (\many{a},\many{b}), \kenv', \subst', \true) \text{ and }  \\
  &\Delta \vdash \theta'' : \Theta' \Rightarrow \emptydelta \text{ and }  \\
  &\restriction{(\theta'' \comp \theta')}{\many{a}} = \rsubst.
\ea
\]
\end{quote}

Then we show that if $\meta{Prop}(\hat{s})$ holds for all $\hat{s} < s$, then
$\meta{Prop}(s)$ holds.
By transfinite induction, this is sufficient to show that $\meta{Prop}(s)$ holds
for all $s$.

}

To this end, let $s = (\forall \Delta \snoc{} \exists \many{a} \snoc{}  F, \Theta, \theta, F, C)$ and we assume
$\wf{s}$~\premissNum{premiss:initial-ok} and
$\Delta; \many{a}; \emptygamma; \rsubst
\vdash F[C \wedge \Uof{\Theta, \theta}]$~\premissNum{premiss:sat}.

We first consider the case that $s$ is already a final state in the senses of
this lemma, meaning that $F$ is of the shape $\exists \many{b}$ for some
$\many{b}$ and $C$ is $\true$.
Further, we have $\theta = \theta'$ and $\Theta = \Theta'$, where
$\ftv(\Theta) = \many{a}, \many{b}$.

This makes \premissRef{premiss:sat} equivalent to
$\Delta; \many{a}; \emptygamma; \rsubst \vdash \exists \many{b}. \Uof{\Theta, \theta}$.
Applying \cref{lemma:satisfying-U} then gives us the existence of an appropriate
$\theta''$.

\vspace{0.3cm}
We now consider the case where $s$ is not a final state, i.e., we don't
have $F[C] = \exists \many{b}. \true$ for any $\many{b}$.
We observe that
\premissRef{premiss:sat} implies
$\empty; \emptydelta; \emptygamma; \emptysubst
\vdash \forall \Delta :: \exists \many{a} :: F[C \wedge \Uof{\Theta, \theta}]$.
This allows us to apply \cref{theorem:progress}, showing that the machine can
take a step from $s$ to a new state $s_1$.
We now show that $s_1$ is of the form
$(\forall \Delta \snoc{} \exists \many{a} \snoc{} F_{1}, \Theta_1, \theta_1, C_1)$
for some $F_{1}, \Theta_1, \theta_1$, and  $C_1$:

If $F$ is empty, then $C$ must not be $\true$.
All stack machine rules applicable in this case preserve all existing stack
frames.
Otherwise, if $F$ is not empty, we observe that the only rules of the stack
machine that may replace more than the topmost stack frame are
$\lab{S-ExistsLower}$, $\lab{S-LetMonoPop}$, and $\lab{S-LetPolyPop}$.

If \lab{S-ExistsLower} was applied, we observe that the only way for the
variables $\many{a}$ in the definition of the rule \lab{S-ExistsLower} not to be
disjoint from the variables $\many{a}$ in the statement of this lemma is if the
stack of $s$ is of the form
$\forall \Delta :: \exists \many{a} :: \exists \many{b}$ for some $\many{b}$,
which violates the assumption about the shape of $F[C]$ above.
Therefore, if \lab{S-ExistsLower} was applied, the bottom-most frames
$\forall \Delta :: \exists \many{a}$ of $s$ remained unchanged.
If $\lab{S-LetPolyPop}$ or $\lab{S-LetMonoPop}$ was applied, then $F$ must contain a
$\Let$ frame and any stack frames below that in $s$ (in particular, the frames
$\forall \Delta :: \exists \many{a}$) remain unchanged.

Therefore, the $\forall \Delta \snoc{} \exists \many{a}$ frames at the bottom of
$s$'s stack are preserved by any rule possibly turning $s$ into $s_{1}$.
Using \premissRef{premiss:initial-ok} and the fact that the lower stack frames
of $s_{1}$ are $\forall \Delta \snoc \exists \many{a}$, we may apply
\cref{theorem:preservation} to the step $s \to s_1$, which gives us
\begin{premisses}
 \item $\Delta; \many{a}; \emptygamma; \rsubst
   \vdash F[C \wedge \Uof{\Theta, \theta}] \;\text{ iff }\; \Delta;  \many{a};
   \emptygamma; \rsubst \vdash F_1[C_1 \wedge \Uof{\Theta_1, \theta_1}]$
   \premissLabel{constraints-equivalent}
\end{premisses}
By \cref{lemma:well-ordering-on-states}, we further have
$s_1 < s$~\premissNum{s1-is-smaller} and by
\cref{lemma:stack-machine-steps-preserve-wf-ness} $\wf{s_1}$
\premissNum{final-state-wf}.

Combining \premissRef{premiss:sat} with \premissRef{constraints-equivalent}
gives us
$\Delta; \many{a};
   \emptygamma; \rsubst \vdash F_1[C_1 \wedge \Uof{\Theta_1, \theta_1}]$.
This, together with \premissRef{final-state-wf} and \premissRef{s1-is-smaller}
allows us to apply the induction hypothesis to $s_1$.
This gives us the existence of
$\kenv', \theta'', \theta', \many{a}$ \text{ s.t.\ }
\begin{premisses}
  \item $(\forall \Delta \snoc{} \exists \many{a}  \snoc{} F_1, \Theta_1, \theta_1, C_1) \to^{*}
  (\forall \Delta :: \exists\; (\many{a},\many{b}), \kenv', \subst', \true)$
  \premissLabel{final-run}
  \item $\Delta \vdash \theta'' : \Theta' \Rightarrow \emptydelta$ \premissLabel{theta-prime-s-wf}
  \item $\restriction{(\theta'' \comp \theta')}{\many{a}} = \rsubst$.
  \premissLabel{theta-1-vs-theta-s-composition}
\end{premisses}
The step $s \to s_1$ extends \premissRef{final-run} to
$(\forall \Delta \snoc{} \exists \many{a} \snoc{} F, \Theta, \theta, C) \to^{*} (\forall \Delta :: \exists\; (\many{a},\many{b}), \kenv', \subst', \true)$
and \premissRef{theta-prime-s-wf} as well as
\premissRef{theta-1-vs-theta-s-composition} show us that $\theta''$ has the desired properties.

\item[$\Longleftarrow$]
Let $s$ be the state
$(\forall \Delta \snoc{} \exists \many{a} \snoc{} F, \Theta, \theta, C)$.
We prove this direction by induction on the length $n$ of the sequence
$s \to^{n} (\forall \Delta :: \exists\; (\many{a},\many{b}), \kenv', \subst', \true)$.
By assumption, we also have
$\Delta \vdash \theta'' : \Theta' \Rightarrow \emptykenv$~\premissNum{rtl:theta-pp-wf}
and
$\restriction{(\theta'' \comp \theta')}{\many{a}}= \rsubst$~\premissNum{rtl:comp}.

If $n = 0$ we have
$\Theta = \Theta'$, $\theta = \theta'$, $C = \true$, and $F = \exists \many{b}$.
The property to prove simplifies to
$\Delta; \many{a}; \emptygamma; \rsubst \vdash \exists \many{b}. \Uof{\Theta, \theta}$.
This follows directly from applying \cref{lemma:satisfying-U} to \premissRef{rtl:theta-pp-wf}
and \premissRef{rtl:comp}.


\vspace{0.3cm}
In the inductive step there exists some $s_{1}$  s.t.\
\[
s \to
s_{1} \to^{*}
(\forall \Delta \snoc{} \exists (\many{a}, \many{b}) , \kenv', \subst', \true)
\]

We now assume that $F[C]$ is not of the form $\exists \many{c}. \true$ for any
$\many{c}$ (otherwise, $s$ would already be a final state in the sense of this
lemma and we finish the proof directly using the $n = 0$ case above).

Therefore, using the same reasoning as in the $\Longrightarrow$ direction, we
know that $s_1$ is of the form
$(\forall \Delta \snoc{} \exists \many{a} \snoc{} F_{1}, \Theta_1, \theta_1, C_1)$
for some $F_{1}, \Theta_1, \theta_1$, and $C_1$.
According to \cref{lemma:stack-machine-steps-preserve-wf-ness}, we have $\wf{s_1}$.
We can therefore apply \cref{theorem:preservation} to this single step,
yielding
\begin{premisses}
 \item $\Delta; \many{a}; \emptygamma; \rsubst
   \vdash F[C \wedge \Uof{\Theta,\theta}] \;\text{ iff }\; \Delta; \many{a};
   \emptygamma; \rsubst \vdash F_1[C_1 \wedge \Uof{\Theta_1, \theta_1}]$
\premissLabel{rtl:equiv}
\end{premisses}

We apply the induction hypothesis to the sequence
$s_{1} \to^{*}
(\forall \Delta \snoc{} \exists (\many{a}, \many{b}) , \kenv', \subst', \true)$,
yielding
$\Delta; \wmany{a}; \emptygamma; \rsubst
\vdash F_{1}[C_{1} \wedge \Uof{\Theta_1, \theta_1}]$
By
\premissRef{rtl:equiv}, this gives us the desired property
$\Delta; \wmany{a}; \emptygamma; \rsubst \vdash F[C \wedge \Uof{\Theta, \theta}]$.
\end{itemize}
\end{proof}

\constraintssolvableunifybysubstnorank*

\begin{proof}
Let $\omany{a}$ be an arbitrary ordering of the variables in $\Xi$.
Further, let $\theta_{\mathrm{a}} := [\omany{a} \mapsto \omany{a}]$ and $\Theta_{\mathrm{a}} := (\wmany{a : \poly})$.
We have
\[
(\emptystack, \emptykenv, \emptysubst, \forall \Delta. \exists \many{a}. \Def\; \Gamma \;\In\; C) \;\to^{*}\;
(\forall \Delta :: \exists \many{a}, \Theta_{\mathrm{a}}, \theta_{\mathrm{a}}, \Def\; \Gamma \;\In\; C)
\]
after $|\Delta|$ applications of the rule $\lab{S-ForallPush}$ and $|\many{a}|$
applications of $\lab{S-ExistsPush}$.
Let the former state be defined as $s$, the latter one as $s'$.
Here, due to $\Delta; \Theta;\Gamma \vdash \wf{C}$, we have $\wf{s'}$.

Therefore, for all
$\hat{\Theta}, \hat{\theta}, \many{b}$ we have
\[
\ba{r c l}
s' &\to^{*}
    &(\forall \Delta :: \exists\; (\many{a},\many{b}), \hat{\kenv}, \hat{\subst}, \true) \\
\multicolumn{3}{c}{\text{iff}} \\
s &\to^{*} &(\forall \Delta :: \exists\; (\many{a},\many{b}), \hat{\kenv}, \hat{\subst}, \true) \\
\ea
\]

\fe{Here we use the $\Delta \vdash \Gamma$ precondition of the Theorem}
Now, let $F$ be the empty stack. We then have
\[
\ba{r r@{} l@{} l@{} l@{}}
                     &\Delta; \Xi; \Gamma; \rsubst \;&\vdash \;&C \\
\;\text{ iff }\quad &\Delta; \Xi; \emptygamma; \rsubst \;&\vdash &(\Def\; \Gamma \;\In\; C) &\text{(by $\Delta \vdash \Gamma$: all mono. conditions satisfied)}\\
\;\text{ iff }\quad &\Delta; \Xi\,; \emptygamma; \rsubst \;&\vdash &(\Def\; \Gamma \;\In\; C) \wedge \Uof{\Theta_\mathrm{a}, \theta_{\mathrm{a}}}
  \qquad\quad&\text{($\Uof{\Theta_\mathrm{a}, \theta_{\mathrm{a}}}$ is equivalent to $\true$)}\\
\;\text{ iff }\quad &\Delta; \Xi\,; \emptygamma; \rsubst \;&\vdash &F[(\Def\; \Gamma \;\In\; C) \wedge \Uof{\Theta_\mathrm{a}, \theta_{\mathrm{a}}}]
  &\text{($F$ is empty)}\\
\ea
\]

Using this, the equivalence to prove then follows directly from
\cref{lemma:helper-constraints-vs-solver}.
\fe{As long as wf-ness of states needs wf-ness of the constraint, we use the wf-ness premise about $C$ here.}
\end{proof}

\subsection{Proof of \cref{theorem:constraint-based-type-inference-sound}}
\input{proofs/constraint-based-type-inference-right.tex}

\subsection{Proof of \cref{theorem:constraint-based-type-inference-complete}}
\input{proofs/constraint-based-type-inference-left.tex}

\section{Further discussion of  Let and Def constraints}

As mentioned in the paper, our treatment of let and def constraints differs from \citet{PottierR05} in some ways, and in particular, lacks the nice property that let constraints can be defined in terms of def constraints, which can in turn be eliminated by a form of inlining.  We have investigated alternative designs, and not found one that has these properties and works otherwise.  In this appendix we outline the results of this exploration of the design space.  We use without further explanation notation from \citet{PottierR05}.
\subsection{Def constraints}
\label{subsection:def-constraint-discussion}

In \cite{PottierR05}, qualified types $\sigma$ are of the form
$\forall \omany{X} [C].T$, where $T$ does not contain further quantifiers or
constraints.

 Due to the existence of constraints in types in their system, the only
difference between Let constraints (using their syntax: of the form
$\Let\; x : \forall \omany{X} [C_1]. T \;\In\; C_2$) and Def constraints (of the form
$\Def\; x : \forall \omany{X} [C_1]. T \;\In\; C_2$) is that the former imposes
satisfaction of the constraint $C_1$.
This is expressed by the following equivalence
in \cite{PottierR05}: we have that
\[
\Let\; x : \forall \omany{X} [C_1]. T \;\In\; C_2
\]
is equivalent
to
\[
\exists \omany{X}. C_1 \;\;\wedge\;\; \Def\; x : \forall \omany{X} [C_1]. T \;\In\; C_2.
\]

\subsubsection{Substituting def away in HM(X)}
\label{Pottier-substituting-def-away}

As stated in \cite{PottierR05}, their def constraints can be substituted away.
Concretely, in their work, the constraint
$\Def\; (x : \forall \omany{X} [C_1].T) \;\In\; C_2$ is equivalent to $C'$,
where $C'$ results from $C_2$ by replacing all $x \preceq T'$ with
$\exists \omany{X}. C_1 \wedge T \le T'$. Here, $\le$ is their subtyping
relation, and we may just use  $\ceq$ for the purposes of this discussion.

\subsubsection{Can we substitute def away?}
\label{substituting-def-away}
We initially did not think this could be possible, however, this is not entirely clear.
In our system, consider the constraint
$\Def\; (x: \forall \omany{a}.H) \;\In\; C_2$.

Straightforwardly adapting the equivalence from \cite{PottierR05}  would yield
a constraint $C'$ that results from $C_2$ by replacing all $x \preceq A$
with $\exists \omany{a}. H \ceq A$ and all $\freeze{x : A}$ with
$\forall \omany{a}. H \ceq A$.

In our system, we then have that $\Def\; (x: \forall \omany{a}.H) \;\In\; C_2$
entails $C'$, but not the other direction.

To understand why, consider
$\Def\; (x : a) \;\In\; x \preceq \forall b. (b \to b)$, where $a$ is some
flexible variable.
Its substituted version $C'$ is simply $a \ceq \forall b. (b \to b)$, which is
clearly satisfiable, whereas the original constraint is not, as we need to pick
the polymorphic type $\forall b. b \to b$ for $a$, which def constraints
prohibit (due to the monomorphism premise in \lab{Sem-Def} \cref{fig:constraint-wellformedness}).

\subsubsection{What about redefining def?}
\label{what-if-redefining-def}
One may now reasonably ask ``is there a \emph{different} notion of def
constraints for which the reverse direction of the entailment holds'' (which
would mean that a def constraint is indeed equivalent to the substituted
version)?  After all, it's
the monomorphism restriction that is preventing the reverse direction from holding.

One option we have considered is to try defining the semantics of def
constraints in terms of the substitution rule directly. Concretely,
consider the following, alternative version of the rule \lab{Sem-Def} from
\cref{paper-fig:constraints-semantics}:
\[
  \inferrule[\lab{Sem-Def-Alt}]
  {
    \Delta; \Xi; \Gamma; \rsubst  \vdash C[x \preceq A \leadsto \exists \omany{a}. H \ceq A][\freeze{x : A} \leadsto \forall \omany{a}.H \ceq A]
  }
  {\Delta; \Xi; \Gamma; \rsubst \vdash \Def\; (x : \forall \omany{a}.H) \;\In\; C}
\]

%
We use $\leadsto$ here to indicate that \emph{all} subconstraints of $C$ of the forms
$x \preceq A$ and $\freeze{x : A}$ are replaced.

We conjecture that this yields a constraint language that still has most general
solutions and we could adapt the solver to behave accordingly.
The reason is that we still only \emph{instantiate}
those variables $\omany{a}$ that were explicitly given in the def constraint. If
the remaining type $H$ contains flexible variables $a$ that are instantiated
with polymorphic types (which this version of def constraints would allow),
those quantifiers could never be instantiated.
This is observable in the earlier example
$\Def\; (x : a) \;\In\; x \preceq \forall b. (b \to b)$, which would now be
\emph{defined} to hold if $a \ceq \forall b. (b \to b)$ holds, and therefore
permits the solution $a \mapsto (\forall b. b \to b)$ without being able to
instantiate $b$.

However, there are several issues with this, which is why we have  not
adopted this idea in the paper.

\begin{enumerate}
\item
It appears rather un-intuitive that with this updated semantics, the \emph{only}
solution for the constraint
$\Def\; (x : a) \;\In\; x \preceq \forall b. (b \to b)$ is the one where we pick
$a \mapsto (\forall b. b \to b)$. When just seeing the constraint, we would
expect $a \mapsto (\forall b.b)$ to work, too.

\item
Similarly, consider the following two constraints:
\[
\ba{rcll}
C_1 \; &:= \; &&\Def (x: \forall a. a \to a) \;\In\; x \preceq c \\
C_2 \; &:= \; &\exists b. b \ceq (\forall a. a \to a) \;\;\wedge &\Def (x: b) \;\In\; x \preceq c \\
\ea
\]
If we use the new version of let constraints (i.e., \lab{Sem-Def-Alt}), we then
have that the solution for $C_1$ requires $c \mapsto d \to d$ for any type $d$,
whereas the solution for $C_2$ requires $c \mapsto (\forall a. a \to a)$. This
seems very un-intuitive.

Using our original semantics for def constraints (i.e., \lab{Sem-Def} from
\cref{paper-fig:constraints-semantics}), we get the same solution for $C_1$ and
$C_2$ is unsatisfiable.

\item
The rule $\lab{Sem-Def-Alt}$ would make def constraints more permissive than
unannotated $\lambda$ terms in FreezeML.
If we used $\lab{Sem-Def-Alt}$, we would be able to conclude that the term
$\lambda x. \lambda (f : (\forall b. b) \to \Int). f\; \freeze{x}$ is
well-typed, because the constraint generated from it would now be satisfiable.
However, this term is not accepted by the FreezeML typing rules.

On the other hand, this problem could be mitigated by generating appropriate
$\monoc{}$ constraints: We would change
\cref{fig:translation}
such that
\[
\congen{\lambda x.M : A}                    =
 \exists a_1, a_2 . (a_1 \to a_2 \ceq A \wedge \Def\; (x : a_1)\; \In \;\congen{M : a_2} \wedge \monoc{}(a_1) )
\]

\end{enumerate}

Instead of following the approach in \ref{what-if-redefining-def}, one may think
that we should make instantiation constraints $A \preceq B$ part of the
constraint language (rather than being syntactic sugar for something else, as in
\cite{PottierR05}).
Given a constraint $\Def (x: A) \;\In\; C_2$, we could then substitute all
$x \preceq B$ in $C_2$ with $A \preceq B$ (and substitute all $\freeze{x : B}$
as in \ref{substituting-def-away}, meaning with $A \ceq B$).

The problem with these constraints is that we need to ensure that the
polymorphism on the left-hand-side is already known:
Without any restrictions, the constraint $a \preceq \Int \to \Int$ would have no
most general solution, as both $a \mapsto (\forall b. b \to b)$ and
$a \mapsto \Int \to \Int$ are feasible.

However, term variables are \emph{the very mechanism} we use in our constraint
language to ensure that the polymorphism on the left-hand-side of an
instantiation constraint is fully determined!
Therefore, using such constraints would not simplify the treatment of
instantiation.

\subsection{Let constraints}

 Just using the semantics of let constraints, we can observe the following properties

\begin{enumerate}
\item
\label{case:let-mono-as-def}
If
$\Delta; \Xi; \Gamma; \rsubst \vdash \Let_\mono\; x = \sqcap a. C_1 \;\In\; C_2$
holds, then for all $\Delta_\mathrm{m}, \rsubst_\mathrm{m}$ such that
$\meta{mostgen}(\Delta, (\Xi, a),\allowbreak \Gamma,\allowbreak C_1, \Delta_\mathrm{m},\allowbreak \rsubst_\mathrm{m})$
we have
$\Delta; \Xi; \Gamma; \rsubst \vdash \exists \Delta_\mathrm{m}. \Def ( x : \rsubst_\mathrm{m}(a)) \;\In\; C_2$.
\item
\label{case:let-poly-as-def}
If
$\Delta; \Xi; \Gamma; \rsubst \vdash \Let_\mono\; x = \sqcap a. C_1 \;\In\; C_2$
holds, then for all
$\Delta_\mathrm{m}, \rsubst_\mathrm{m}, \Delta_\mathrm{o}, \omany{b}$ such that
\[
\bl
\meta{mostgen}(\Delta, (\Xi, a),\allowbreak \Gamma,\allowbreak C_1, \Delta_\mathrm{m},\allowbreak \rsubst_\mathrm{m}) \\
\Delta_\mathrm{o} = \ftv(\rsubst_\mathrm{m}(\Xi)) - \Delta \\
\omany{b} = \ftv(\rsubst_\mathrm{m}(a)) - \Delta, \Delta_\mathrm{o} \\\\
\el
\]
we have
$\Delta; \Xi; \Gamma; \rsubst \vdash \exists \Delta_\mathrm{o}. \Def ( x : \forall \omany{b}. \rsubst_\mathrm{m}(a)) \;\In\; C_2$.
\end{enumerate}
Note that most general types are only unique up to the names of the freshly
introduced variables, which is why we need to universally quantify over
$\rsubst_\mathrm{m}$ and $\Delta_\mathrm{m}$ here.
The extra conditions in (\ref{case:let-poly-as-def}) defining $\Delta_\mathrm{o}$ and $\omany{b}$ are identical to those in \lab{Sem-Let-Poly}.

The properties stated in (\ref{case:let-mono-as-def}) and
(\ref{case:let-poly-as-def}) are not equivalences between let and def constraints,
but only entailments.

It would probably be possible to state some form of equivalence, but that doesn't seem helpful:
Because we don't have constraints in types, a constraint
$\Def\; (x : A) \;\In\; C_2$ with an arbitrary type $A$ entails any constraint
$\Let_R\; x = \sqcap a. C_1 \;\In\; C_2$ where $C_1$ (a constraint pulled out of
thin air) has a most general solution that is somehow related to $A$.

Given the entailments stated in in (\ref{case:let-mono-as-def}) and
(\ref{case:let-poly-as-def}), we could apply the substitution of def constraints
discussed in point \ref{substituting-def-away} in the previous section
\ref{subsection:def-constraint-discussion} to relate def constraints to
constraints without def constraints.
Recall that these are also just entailments, not equivalences.

So far in this section, we have related constraints
$\Let\; x = \sqcap a. C_1 \;\In\; C_2$ to other constraints that effectively
require us to solve $C_1$ first (because we refer to
$\meta{mostgen}(\dots,C_1,\dots)$ in the expansion).
Would it be possible to relate let constraints to other
constraints without solving $C_1$ first? In \cite{GarrigueR99} this is possible,
by using constraints in types.

First, we observe that due to the impredicative nature of our system, we cannot deal with
instantiation constraints
as in \cite{GarrigueR99} (which we discussed Section~\ref{Pottier-substituting-def-away}).
Concretely, given a constraint $\Let_R\; x = \sqcap a. C_1 \;\In\; C_2$, we cannot
simply replace all $x \preceq A$ occuring in $C_2$ with
$\exists a. C_1 \wedge a \ceq A$.
While this may work for examples where generalisation occurs, like
$\Let_\poly\; x = \sqcap a. \exists b. a \ceq b \to b  \;\In\; x \preceq \Int \to \Int$,
it doesn't work in cases such as
$\Let_R\; x = \sqcap a. a \ceq (\forall b. b \to b) \;\In\; x \preceq \Int \to \Int$
(irrespective of the choice of $R$), because the scheme above would yield
$\exists a. a \ceq (\forall b. b \to b) \wedge a \ceq (\Int \to \Int)$, which is
unsatisfiable.

When considering $\Let_\poly$ constraints, we observe that we would need some
form of explicit generalisation constraints to substitute away freeze constraints under a generalising let constraint.

Consider the example
$\Let_\poly\; x = \sqcap a. \exists b. a \ceq b \to b \;\In\; \freeze{x : \forall b. b \to b}$.
Here, we cannot substitute away $\freeze{x : \forall b. b \to b}$ with
$\exists a. \exists b. a \ceq b \to b \;\wedge\; a \ceq \forall b. b \to b$ (as
discussed for handling freeze constraints under $\Def$ in
\ref{substituting-def-away}).

Instead, we would need some kind of generalisation constraint
$\sqcap a. C : A$ that generalises the type for $a$ in $C$
and asserts that the generalised solution for $a$ is equal to $A$.
However, we would have to generalise not any type but the most general type:
Otherwise, $\sqcap a. \exists b. a \ceq b \to b : c$ would have the two solutions
$c \mapsto (\Int \to \Int)$ and $c \mapsto \forall b. b \to b$, depending on
what type we chose for $b$ prior to generalising.
This is the same reason why let terms and constraints have principality
conditions in our system.

Therefore, we would need to push the $\meta{mostgen}$ logic from let constraints
into the semantics of such generalisation constraints, too.

\end{document}

%% file: macros.tex
\usepackage{thmtools}
\usepackage{thm-restate}

\ifcomments
\newcommand{\todo}[1]{{\color{red} {\textsc{TODO:}} #1}}
\newcommand{\draft}[1]{{\color{black!50} #1}}
\newcommand{\outdated}[1]{{\color{black!50} {\textsc{OUTDATED:}} #1}}
\newcommand{\jrc}[1]{{\color{red} {\textsc{JRC:}} #1}}
\newcommand{\fe}[1]{{\color{red}  {\textsc{FE:}}  #1}}
\newcommand{\sam}[1]{{\color{red} {\textsc{SL:}}  #1}}
\newcommand{\js}[1]{{\color{red}  {\textsc{JS:}}  #1}}
\else
\newcommand{\todo}[1]{}
\newcommand{\draft}[1]{}
\newcommand{\outdated}[1]{}
\newcommand{\jrc}[1]{}
\newcommand{\fe}[1]{}
\newcommand{\sam}[1]{}
\newcommand{\js}[1]{}
\fi

\declaretheorem[name=Theorem]{thm}

\declaretheorem[name=Lemma,sibling=thm]{lem}
\declaretheorem[name=Remark]{remark}

\renewcommand{\restriction}[2]{{#1}\mathord{\upharpoonright}_{#2}}
\newcommand{\invertedrestriction}{\mathord{\downharpoonright}}

\newcommand{\freezeml}{FreezeML\xspace}

\newcommand{\dec}[1]{\mathsf{#1}}    
\newcommand{\key}[1]{{\rm {\bf #1}}} 
\newcommand{\meta}{\mathsf}
\newcommand{\freeze}[1]{\lceil{#1}\rceil}

\newcommand{\Let}{\key{let}}
\newcommand{\Def}{\key{def}}
\newcommand{\In}{\key{in}}
\newcommand{\Int}{\dec{Int}}

\newcommand{\List}{\dec{List}}

\newcommand{\letexists}[1]{\sqcap #1.}

\newcommand{\poly}{\star}
\newcommand{\mono}{{\mathord{\bullet}}}

\newcommand{\true}{\dec{true}}
\newcommand{\monoc}{\dec{mono}}

\newcommand{\rsubst}{\delta}  
\newcommand{\fsubst}{\theta}  
\newcommand{\type}[2]{#1 \vdash #2}

\newcommand{\unify}{\meta{unify}}

\newcommand{\freezemlLab}[1]{\text{\scshape{#1}}}

\newcommand{\lab}[1]{\text{\scshape{#1}}}
\newcommand{\slab}[1]{\textrm{#1}}
\newcommand{\parlab}[1]{(\lab{#1})} 
\newcommand{\congen}[1]{\llbracket #1 \rrbracket}

\newcommand{\many}[1]{\tilde{#1}} 
\newcommand{\omany}[1]{\overline{#1}} 
\newcommand{\wmany}[1]{\widetilde{#1}} 

\newcommand{\mgen}{\meta{gen}}

\newcommand{\msplit}{\meta{split}}

\newcommand{\emptystack}{\cdot}
\newcommand{\emptysubst}{\emptyset}

\newcommand{\emptykenv}{\cdot}
\newcommand{\emptygamma}{\cdot}
\newcommand{\emptydelta}{\cdot}

\newcommand{\subst}{\theta}
\newcommand{\substunifier}{\mathcal{U}}

\newcommand{\ceq}{\sim}   
\newcommand{\comp}{\circ} 
\newcommand{\kenv}{\Theta}
\newcommand{\wf}[1]{#1\ \mathbf{ok}} 
\newcommand{\sat}[1]{#1\ \mathbf{sat}} 
\newcommand{\resd}[2]{{#1}\invertedrestriction_{#2}} 
\newcommand{\resk}[2]{{#1}\invertedrestriction_{#2}} 
\newcommand{\resg}[2]{{#1}\invertedrestriction_{#2}} 
\newcommand{\ress}[2]{{#1}\invertedrestriction_{#2}} 
\newcommand{\btv}{\meta{btv}}    
\newcommand{\bv}{\meta{bv}}      
\newcommand{\ftv}{\meta{ftv}}    

\newcommand{\ba}{\begin{array}}
\newcommand{\ea}{\end{array}}
\newcommand{\bl}{\ba[t]{@{}l@{}}}
\newcommand{\el}{\ea}

\crefname{equation}{}{}
\Crefname{equation}{Statement}{Statements}
\crefname{lemma}{Lemma}{Lemmas}
\Crefname{lemma}{Lemma}{Lemmas}
\crefname{lem}{Lemma}{Lemmas}
\Crefname{lem}{Lemma}{Lemmas}
\crefname{figure}{Figure}{Figures}
\Crefname{figure}{Figure}{Figures}
\crefname{section}{Section}{Section}
\Crefname{section}{Section}{Section}
\crefname{thm}{Theorem}{Theorems}
\Crefname{thm}{Theorem}{Theorems}
\crefname{defn}{Definition}{Definitions}
\Crefname{defn}{Definition}{Definitions}
\def\currentprefix{proof:}
\newcounter{premisscounter}
\newcommand{\resetNum}{\setcounter{premisscounter}{0}\setcounter{obligationcounter}{0}}
\newcommand*{\proofContext}[1]{\resetNum\def\currentprefix{proof:#1}}
\newcommand*{\localeqlabel}[1]{\label[premisscounter]{\currentprefix:#1}}

\newcommand*{\localref}[1]{%
\ref*{\currentprefix:#1}}

\let\Item\item
\newenvironment{premisses}
  {\renewcommand{\item}{\refstepcounter{premisscounter} \Item}
   \begin{list}{\textbf{(\arabic{premisscounter})}}{}}
  {\end{list}}

\newcommand*{\premissLabel}[1]{\label{\currentprefix:#1}}
\newcommand*{\premissNum}[1]{\refstepcounter{premisscounter}\localeqlabel{#1}\textbf{(\arabic{premisscounter})}}
\newcommand*{\premissRef}[1]{(\localref{#1})}
\makeatletter
\def\maketag@@@#1{\hbox{\m@th\normalfont\bfseries#1}}
\makeatother

\newcounter{obligationcounter}
\let\Item\item
\newenvironment{obligations}
  {\renewcommand{\item}{\refstepcounter{obligationcounter} \Item}
   \begin{list}{\textbf{(\Roman{obligationcounter})}}{}}
  {\end{list}}

\newcommand*{\obligationLabel}[1]{\label{\currentprefix:#1}}
\newcommand*{\obligationNum}[1]{\refstepcounter{obligationcounter}\premissLabel{#1}{\textbf{(\Roman{obligationcounter})}}}
\newcommand*{\obligationRef}[1]{(\ref{\currentprefix:#1})}
\newcommand*{\obligationFulfilled}[1]{\textit{(\localref{#1})}}

\newcommand{\disjoint}{\mathop{\#}}

\newcommand{\N}{\mathbb{N}}

\newcommand{\Gammaof}[1]{\dec{tc}(#1)}
\newcommand{\Deltaof}[1]{\dec{rc}(#1)}
\newcommand{\Xiof}[1]{\dec{fc}(#1)}

\newcommand{\Uof}[1]{\mathfrak{U}(#1)}

\newcommand{\snoc}{::}

\newcommand{\sysname}{Constraint FreezeML\xspace}

\newcommand{\HMX}{HM($X$)\xspace}
\newcommand{\HM}[1]{HM($#1$)\xspace}


%% file: figures/freezeml-typing-rules.tex
\begin{figure}
\raggedright
$\boxed{\type{\Delta; \Gamma}{M : A}}$

\begin{mathpar}
  \inferrule*[Lab=\freezemlLab{VarFrozen}]
    {x : A \in \Gamma}
    {\type{\Delta; \Gamma}{\freeze{x} : A}}

  \inferrule*[Lab=\freezemlLab{VarPlain}]
    {{x : \forall \omany{a}.H \in \Gamma} \\
      \Delta \vdash \rsubst : \many{a} \Rightarrow_\poly \cdot
    }
    {\type{\Delta; \Gamma}{x : \rsubst(H)}}

  \inferrule*[Lab=\freezemlLab{App}]
    {\type{\Delta; \Gamma}{M : A \to B} \\
     \type{\Delta; \Gamma}{N : A}
    }
    {\type{\Delta; \Gamma}{M\,N : B}} \\

  \inferrule*[Lab=\freezemlLab{LamPlain}]
    {\type{\Delta; (\Gamma, x : S)}{M : B}}
    {\type{\Delta; \Gamma}{\lambda x.M : S \to B}}

  \inferrule*[Lab=\freezemlLab{LamAnn}]
    {\type{\Delta; (\Gamma, x : A)}{M : B}}
    {\type{\Delta; \Gamma}{\lambda (x : A).M : A \to B}}

  \inferrule*[Lab=\freezemlLab{LetPlain}]
    { \omany{a} = \ftv(A') - \Delta \\\\
     (\Delta, \omany{a}, M, A') \Updownarrow A \\
     (\Delta, \many{a}); \Gamma \vdash M : A' \\
     \type{\Delta; (\Gamma, x : A)}{N : B} \\
     \meta{principal}(\Delta, \Gamma, M, \many{a}, A')
    }
    {\type{\Delta; \Gamma}{\Let \; x = M\; \In \; N : B}}

  \inferrule*[Lab=\freezemlLab{LetAnn}]
    {(\many{a}, A') = \msplit(A, M) \\
     \type{(\Delta, \many{a}); \Gamma}{M : A'} \\
     \type{\Delta; (\Gamma, x : A)}{N : B}
    }
    {\type{\Delta; \Gamma}{\Let \; (x : A) = M\; \In \; N : B}}

\end{mathpar}

\caption{\freezeml Typing Rules.}
\label{fig:freezeml-typing}
\end{figure}


%% file: thm-constraint-gen-sound.tex
\begin{restatable}[Constraint generation is sound with respect to the typing judgement]{thm}{constraintgenerationcompleteness}
\label{theorem:constraint-generation-soundness}
Let $\Delta; \Gamma \vdash M : A$ and $a \disjoint \Delta$.
Then
$\Delta; a; \Gamma; [a \mapsto A] \vdash \congen{M : a}$
holds.
\end{restatable}


%% file: thm-constraint-gen-complete.tex
\begin{restatable}[Constraint generation is complete with respect to the typing judgement]{thm}{constraintgenerationsoundness}
\label{theorem:constraint-generation-completeness}
If $\Delta;\Gamma \vdash \wf{M}$ and $\Delta; a; \Gamma; \rsubst \vdash \congen{M : a}$, then $\Delta; \Gamma
\vdash M : \rsubst(a)$.
\end{restatable}


%% file: figures/stack-wellformedness.tex
{
\raggedright
$\boxed{\kenv \vdash \wf{F}}$

\begin{mathpar}

\inferrule
  { }
  {\emptykenv \vdash \wf{\emptystack}}

\inferrule
  {\Deltaof{F}; \kenv; \Gammaof{F} \vdash \wf{C} \\
   \kenv \vdash \wf{F}}
  {\kenv \vdash \wf{F :: \Box \wedge C}}
\\
\inferrule
  {
   \kenv \vdash \wf{F} \\
   a \not\in \btv(F)
  }
  {\kenv \vdash \wf{F :: \forall \;a}}

\inferrule
  {
   \kenv \vdash \wf{F} \\
   a \not\in \btv(F)
  }
  {(\kenv, a : R) \vdash \wf{F :: \exists \; a}}
\\
\inferrule
  {
   \text{for all } a \in \ftv(A) - \Deltaof{F} \mid (a : \mono) \in \Theta  \\\\
    x \not\in \bv(F) \\
   \kenv \vdash \wf{F}}
  {\kenv \vdash \wf{F :: \Def\; (x : A)}}

\inferrule
  {
   \Deltaof{F}; \kenv ; (\Gammaof{F}, x : A) \vdash \wf{C}\\\\
   x \not\in \bv(F) \\
   \kenv \vdash \wf{F} \\
  }
  {(\kenv, a : R) \vdash \wf{F :: \Let_{R'} \, x = \letexists{a}\Box \;\dec{in}\; C}}
\end{mathpar}
}


%% file: figures/stack-machine-rules-new.tex
{

\def\distarrow{0.35cm}
\def\distlines{12pt}
\def\distwhere{3pt}
\def\distlabel{7.15cm}
\def\wherefont{\footnotesize}
\def\letoutxshift{1.7cm}
\definecolor{wheregray}{RGB}{90,90,90}

\newcounter{rulecounter}
\newcommand{\makestep}[4][0pt]{%
  \stepcounter{rulecounter}

  \node[anchor=east,xshift=#1] (left-\therulecounter) at (left-reference) {\ensuremath{#2}};
  \node[anchor=west,xshift=#1] (right-\therulecounter) at (right-reference) {\ensuremath{#3}};

  \node[right=\distarrow of left-\therulecounter.east,anchor=center] (arrow-\therulecounter)  {$\rightarrow$};

  \node[right=\distlabel of right-\therulecounter.north west, anchor=north west,xshift=-#1] (label-\therulecounter) {$(\lab{#4})$};

  \coordinate[below=\distlines of left-\therulecounter.south east,xshift=-#1] (left-reference);
  \coordinate[xshift=-#1] (right-reference) at (left-reference -| right-\therulecounter.west);

}
\newcommand{\makewhere}[1]{
  \node[below=\distwhere of left-\therulecounter.south west,anchor=north west,xshift=.25cm, font=\wherefont] (where-\therulecounter) {\color{wheregray}#1};

  \coordinate (intersect-left) at (left-reference |- where-\therulecounter.south);
  \coordinate (intersect-right) at (right-reference |- where-\therulecounter.south);

  \coordinate[below=\distlines of intersect-left] (left-reference);
  \coordinate[below=\distlines of intersect-right] (right-reference) ;

}

\newcommand{\inlinewhere}[1]{
{\;\wherefont \color{wheregray} #1}
}

\begin{figure}[t]
\begin{tikzpicture}

\tikzset{every node/.style={inner sep=0pt,font=\small}}

\coordinate (left-reference);
\coordinate[right=2*\distarrow of left-reference] (right-reference);

\makestep{(F, \kenv, \subst, A \ceq B)}{(F, \kenv', \subst' \comp \subst, \true)
\inlinewhere{\text{where} (\kenv', \subst') = \substunifier(\Deltaof{F}, \kenv, \subst A, \subst B)}
}{S-Eq}

\makestep{(F, \kenv, \subst, \freeze{x : A})}{(F, \kenv, \subst, \Gammaof{F}(x) \ceq A)}{S-Freeze}

\makestep{(F, \kenv, \subst, x \preceq A)}
{(F, \kenv, \subst, \exists \omany{a}.H \ceq A)
\inlinewhere{\text{where } \forall \omany{a}.H = \Gammaof{F}(x) \quad
\many{a} \disjoint \btv(F)}
}{S-Inst}

\makestep{(F, \kenv, \subst, \monoc(a))}
{(F, \kenv', \subst, \true)
}{S-Mono}
\makewhere{where
$
\bl
\ba[t]{@{}l}
\omany{b} = \ftv(\theta(a)) - \Deltaof{F} \quad
\Theta' = (\Theta - \many{b}) \cup \wmany{b : \mono} \quad
\Theta' \vdash_{\mono} \wf{\theta(a)}
\ea
\el
$
}

\makestep
{(F, \kenv, \subst, C_1 \wedge C_2)}
{(F :: \Box \wedge C_2, \kenv, \subst, C_1)}{S-ConjPush}

\makestep
{(F :: \Box \wedge C_2, \kenv, \subst, \true)}
{(F, \kenv, \subst, C_2)}{S-ConjPop}

\makestep
{(F, \kenv, \subst, \exists a.\, C)}
{(F :: \exists a, (\kenv, a : \poly ), \subst[ a \mapsto a ], C)}{S-ExistsPush}

\makestep
{(F :: f :: \exists \omany{a}, \kenv, \subst, \true)}
{(F :: \exists \many{c} :: f, \kenv', \restriction{\subst}{\Theta'}, \true)}{S-ExistsLower}
\makewhere{
where
$
\bl
\text{$f$ is neither a $\Let$ or $\exists$ frame} \quad
\many{b};\: \many{c} = \dec{partition}(\many{a}, \theta, \Theta) \quad
\Theta' = \Theta - \many{b} \quad
|\many{a}| > 0
\el
$
}

\makestep
{(F, \kenv, \subst, \forall a. C)}
{(F :: \forall\;a, \kenv, \subst, C)}{S-ForallPush}

\makestep{(F :: \forall\;a, \kenv, \subst, \true)}
{(F, \kenv, \subst, \true)   \inlinewhere{\text{where }   a \not\in \ftv(\theta(\Theta))}}{S-ForallPop}

\makestep{(F, \kenv, \subst, \Def\; (x : A) \;\In\; C)}
{(F :: \Def\; (x : A), \kenv', \subst, C)
}{S-DefPush}
\makewhere{where
$
\bl
\ba[t]{@{}l}
\omany{b} = \ftv(\theta(A)) - \Deltaof{F} \quad
\Theta' = (\Theta - \many{b}) \cup \wmany{b : \mono} \quad
\text{for all } a \in \ftv(A) \mid
\Theta' \vdash_{\mono} \wf{\theta(a)}
\ea
\el
$
}

\makestep
{(F :: \Def\; (x : A), \kenv, \subst, \true)}
{(F, \kenv, \subst, \true)}{S-DefPop}

\makestep
{(F, \kenv, \subst, \Let_R\; x = \letexists{b} C_1 \;\In\; C_2)}
{(F :: \Let_R\; x = \letexists{b} \Box \;\In\; C_2, (\kenv, b : \poly), \subst[ b \mapsto b ], C_1)}{S-LetPush}

\makestep[\letoutxshift]
{(F :: \Let_\poly \; x = \letexists{b} \Box \;\In\; C :: \exists \omany{a}, \kenv, \subst, \true)}
{(F :: \exists \wmany{a''}, \kenv', \restriction{\subst}{\Theta'}, \Def\; (x : B) \;\In\; C)}{S-LetPolyPop}
\makewhere{
where
$
\bl
\wmany{a'};\: \wmany{a''} = \dec{partition}((\many{a}, b), \theta, \Theta) \quad
A = \subst(b) \quad
\omany{c}\: = \ftv(A) \cap \wmany{a'} \quad
\Theta' = \Theta - \wmany{a'} \quad
B = \forall \omany{c}.A
\el
$
}

\makestep[\letoutxshift]
{(F :: \Let_\mono \; x = \letexists{b} \Box \;\In\; C :: \exists \omany{a}, \kenv, \subst, \true)}
{(F :: \exists (\many{c}, \wmany{a''}), \kenv', \restriction{\subst}{\Theta'}, \Def\; (x : A) \;\In\; C)}{S-LetMonoPop}
\makewhere{
where
$
\bl
\wmany{a'};\: \wmany{a''} =\dec{partition}((\many{a}, b), \theta, \Theta) \quad
A = \subst(b) \quad
\omany{c} = \ftv(A) \cap \wmany{a'} \quad
\Theta' = \Theta - (\wmany{a'} - \many{c})
\el
$
}

\end{tikzpicture}

\caption{Constraint solving rules.}
\label{paper-fig:constraints-stack-machine}
\end{figure}
}


%% file: thm-preservation.tex
\begin{restatable}[Preservation]{thm}{subjectreduction}
\label{theorem:preservation}
If $\wf{(\forall \Delta :: \exists \Xi :: F_0, \kenv_0, \subst_0, C_0)}$ and
\[
   \premissLabel{premiss:initial-state-wf}
  (\forall \Delta :: \exists \Xi :: F_0, \kenv_0, \subst_0, C_0) \to (\forall \Delta :: \exists \Xi :: F_1, \kenv_1,
   \subst_1, C_1)
\]
then
\[
 \Delta; \Xi; \emptygamma; \rsubst \vdash F_0[C_0 \wedge \mathfrak{U}(\Theta_0, \theta_0)] \;\text{ iff }\;
 \Delta; \Xi; \emptygamma; \rsubst  \vdash F_1[C_1 \wedge \mathfrak{U}(\Theta_1, \theta_1)]
\]
\end{restatable}


%% file: thm-progress.tex
\begin{restatable}[Progress]{thm}{progress}
\label{theorem:progress}
Let $\wf{(F, \kenv, \subst, C)}$ and $F[C] \neq \forall \Delta. \exists \Xi. \true$ for
all $\Delta, \Xi$.
Further, let
$\emptydelta; \emptykenv;  \emptygamma ;\emptysubst \vdash F[C \wedge \mathfrak{U}(\Theta, \theta)]$.
Then there exists a state $s_1$ such that $(F, \kenv, \Gamma, \subst,
C) \to s_1$.
\end{restatable}


%% file: thm-solver-overall-correctness.tex
\begin{restatable}[Correctness of constraint solver]{thm}{constraintssolvableunifybysubstnorank}
\label{theorem:solver-correct}
Let $\Delta \vdash \wf{\Gamma}$ and $\Delta;\Xi;\Gamma  \vdash \wf{C}$.
Then we have
\[
\bl
\Delta; \Xi; \Gamma; \rsubst
\vdash C  \\
\text{iff} \\
\text{there exist } \kenv, \theta', \theta, \Xi' \text{ s.t.\ } \\
\ba{cl}
 \;&(\emptystack, \emptykenv, \emptysubst, \forall \Delta.\, \exists \Xi.\, \Def\; \Gamma \;\In\; C) \to^{*}
    (\forall \Delta :: \exists\; (\Xi,\Xi'), \kenv, \subst, \true) \text{ and }  \\
  &\Delta \vdash \theta' : \Theta \Rightarrow \emptydelta \text{ and }  \\
  &\restriction{(\theta' \comp \theta)}{\Xi} = \rsubst.
\ea
\el
\]
\end{restatable}


%% file: thm-constraint-tc-sound.tex
\begin{restatable}[Constraint-based typechecking is sound]{thm}{constraintbasedtypeinferenceright}
\label{theorem:constraint-based-type-inference-sound}
Let $\Delta \vdash \Gamma$ and $\Delta;\Gamma \vdash \wf{M}$ and $a \disjoint \Delta$.
If $(\emptystack, \emptykenv, \emptysubst, \allowbreak \forall \Delta.\, \exists a.\, \Def\;
\Gamma \;\In\;\allowbreak \congen{M : a})$ $\to^{*}$ $( \forall \Delta :: \exists\; (a, \many{b}), \Theta, \subst,\allowbreak \true)$ and $\Delta \vdash
\subst' : \Theta \Rightarrow \cdot$ then
$\Delta; \Gamma \vdash M : (\subst' \comp \subst)(a)$.
\end{restatable}


%% file: thm-constraint-tc-complete.tex
\begin{restatable}[Constraint-based typechecking is complete and most general]{thm}{constraintbasedtypeinferenceleft}
\label{theorem:constraint-based-type-inference-complete}
Let $a \disjoint \Delta$. If
 $\Delta; \Gamma \vdash M : A$ then there exist
 $\Xi, \Theta, \subst$, $\rsubst$ such that $(\emptystack, \emptykenv, \emptysubst,
 \forall \Delta\,. \exists a\,.\Def\; \Gamma\; \In \; \congen{M : a}) \to^{*}
(\forall \Delta :: \exists\; \Xi, \Theta, \subst, \true)$ and $A =
\rsubst(\subst(a))$.
\end{restatable}


%% file: proofs/principal-vs-mostgen.tex
\principaliffmostgen*

\begin{proof}

Recall the definitons of $\meta{principal}$ and $\meta{mostgen}$:

\begin{align}
&\meta{principal}(\Delta, \Gamma, M, \Delta', A) = \notag \\
&\qquad\Delta, \Delta'; \Gamma \vdash M : A \;\text{ and }
  \label{proof:principal-vs-mostgen:principal-part-one}\\
&\begin{aligned}
\qquad(\text{for all }\Delta'', A'' \mid
       \bl
       \text{if }
       \Delta, \Delta''; \Gamma \vdash M : A'' \\
       \text{then there exists }
       \rsubst
       \text{ such that } \\
       \;\Delta  \vdash \rsubst : \Delta' \Rightarrow_\poly \Delta''
       \text{ and }
       \rsubst(A) = A'' )
       \el
\end{aligned}
\label{proof:principal-vs-mostgen:principal-part-two}
\end{align}

\begin{align}
&\meta{mostgen}(\Delta, a, \Gamma, \rsubst, \congen{M : a}, \Delta', [a \mapsto A]) = \notag \\
&\qquad (\Delta, \Delta'); a; \Gamma; [a \mapsto A] \vdash \congen{M : a} \;\text{ and }\;
  \label{proof:principal-vs-mostgen:mostgen-part-one} \\
&\qquad
\begin{aligned}
  (\text{for all }\Delta'', \rsubst'' \mid
       \bl
       \text{if }
       (\Delta, \Delta'');  a; \Gamma; \rsubst'' \vdash \congen{M : a} \\
       \text{then there exists }
       \rsubst
       \text{ such that } \\
       \;\Delta  \vdash \rsubst : \Delta' \Rightarrow_\poly \Delta''
       \text{ and }
       \rsubst'' = \rsubst \comp [a \mapsto A] )\\
       \el \\
\end{aligned}
\label{proof:principal-vs-mostgen:mostgen-part-two}
\end{align}

\begin{itemize}
\item[$\mathbf{\Rightarrow}$]
We apply \cref{theorem:constraint-generation-soundness} to
$(\ref{proof:principal-vs-mostgen:principal-part-one})$, immediately yielding
the desired property
$(\ref{proof:principal-vs-mostgen:mostgen-part-one})$.

To show $(\ref{proof:principal-vs-mostgen:mostgen-part-two})$, we assume
$(\Delta, \Delta''); a; \Gamma; \rsubst'' \vdash \congen{M : a}$,
which implies that $\rsubst'' = [a \mapsto A'']$ for some $A''$.


By \cref{theorem:constraint-generation-completeness} this gives us
$(\Delta, \Delta''); \Gamma \vdash M : A''$.
According to $(\ref{proof:principal-vs-mostgen:principal-part-two})$, there
exists a $\rsubst$ with the desired properties.

\item[$\mathbf{\Leftarrow}$]

We apply \cref{theorem:constraint-generation-completeness} to
(\ref{proof:principal-vs-mostgen:mostgen-part-one}),
which gives us
satisfaction of property
$(\ref{proof:principal-vs-mostgen:principal-part-one})$.

To show (\ref{proof:principal-vs-mostgen:mostgen-part-two}),
we  assume $\Delta, \Delta''; \Gamma \vdash M : A''$.
\Cref{theorem:constraint-generation-soundness} then gives us
$(\Delta, \Delta'');  a ; \Gamma; [a \mapsto A''] \vdash \congen{M : a}$.
According to $(\ref{proof:principal-vs-mostgen:mostgen-part-two})$ there exists
an appropriate $\rsubst$.
\end{itemize}
\end{proof}


%% file: proofs/constraint-generation-sound.tex

\constraintgenerationcompleteness*

\begin{proof}
By structural induction on $M$, focusing on the let cases.
\fe{Handle other cases, too?}

\begin{description}
\item[Case $\Let\; x = M' \;\In\; N'$, where $M' \in \dec{GVal}$]

The derivation of $\Delta; \Gamma \vdash M : A$ has the following form, for
some $B, B', \omany{a}$.
\[
\inferrule
    { \omany{a} = \ftv(B') - \Delta \\
     (\Delta, \omany{a}, M', B') \Updownarrow B \\
     (\Delta, \many{a}); \Gamma \vdash M' : B' \\
     \type{\Delta; \Gamma, x : B}{N : A} \\
     \meta{principal}(\Delta, \Gamma, M', \many{a}, B')
    }
    {\type{\Delta; \Gamma}{\Let \; x = M'\; \In \; N : A}}
\]
By $M' \in\dec{GVal}$ we have $B = \forall \omany{a}. B'$ and
$\congen{M : a} = \Let_\poly\; x = \letexists{b} \congen{M ' : b} \;\In\; \congen{N' : a}$.
We assume w.l.o.g.\ that $b \disjoint (\Delta, \many{a})$.
By induction we have
$(\Delta, \many{a}); b; \Gamma; [b \mapsto B'] \vdash \congen{M' : b}$ and
$\Delta; a; (\Gamma, x : B); [a \mapsto A] \vdash \congen{N' : a}$.
We can weaken the former to
$(\Delta, \many{a}); (a, b); \Gamma; [a \mapsto A, b \mapsto B'] \vdash \congen{M' : b}$.
By \cref{lemma:aux:principal-iff-mostgen},
$\meta{principal}(\Delta, \Gamma, M', \many{a}, B')$ implies
$\meta{mostgen}(\Delta, b, \Gamma, \congen{M' : b}, \many{a}, [b \mapsto B'])$.
According to \cref{lemma:weakening-mostgen} we can weaken this to
$\meta{mostgen}(\Delta, (a, b), \Gamma, \congen{M' : b}, (\many{a}, c), [a \mapsto c, b \mapsto B'])$
for some fresh $c$.

Let $\Delta_\mathrm{o} \coloneqq c$ and
$\rsubst' \coloneqq [c \mapsto \dec{unit}]$ and
$\rsubst_\mathrm{m} = [a \mapsto c, b \mapsto B']$.
%
Due to  $c \not\in \ftv(B')$ we have $\rsubst'(\rsubst(b)) = \rsubst(B') = B'$.

We can then derive the following
\[
\inferrule
  {
    \meta{mostgen}(\Delta, (a, b), \Gamma, \congen{M' : b}, (\many{a}, c), \rsubst_\mathrm{m}) \\\\
    \Delta_\mathrm{o} = \ftv(\rsubst_\mathrm{m}(a)) - \Delta \\
    \omany{a} = \ftv(\rsubst_\mathrm{m}(b)) - \Delta, \Delta_\mathrm{o} \\\\
    \Delta \vdash \rsubst' : \Delta_\mathrm{o} \Rightarrow_\mono \emptydelta \\
    B' = \rsubst'(\rsubst_\mathrm{m}(b)) \\\\
    (\Delta, \many{a}); (a, b); \Gamma; \rsubst[a \mapsto A, b \mapsto B'] \vdash \congen{M' : b} \\
    \Delta; a; (\Gamma, x : B); \rsubst \vdash \congen{N' : a}
  }
  {\Delta; a; \Gamma; [a \mapsto A] \vdash \Let_\poly\; x = \letexists{b} \congen{M' : b} \;\In\; \congen{N' : a}}
\]

\item[Case $\Let\; x = M' \;\In\; N'$, where $M' \not\in \dec{GVal}$]

We have a derivation of the same shape as in the previous case.
However, by $M' \not\in\dec{GVal}$ we have $B = \delta(B')$ and
$\congen{M : a} = \Let_\mono\; x = \letexists{b} \congen{M ' : b} \;\In\; \congen{N' : a}$,
for some $\delta$ such that
$\Delta \vdash \rsubst : \many{a} \Rightarrow_\mono \emptydelta$.

Let $\Delta_\mathrm{o} $ and
$\rsubst_\mathrm{m}$ as in the previous case
and extend $\rsubst$ to $\rsubst'$ by setting $\rsubst'(a) = c$.
This implies $B = \rsubst(B') = \rsubst'(B') = \rsubst(\rsubst_\mathrm{m}(b))$.

We also extend $\rsubst$ to $\rsubst''$ by setting $\rsubst''(a') = a'$ for all $a' \in \Delta$.

Using similar reasoning as in the previous case we get
$\meta{mostgen}(\Delta, (a, b), \Gamma, \congen{M' : b}, (\many{a}, c), \rsubst_\mathrm{m})$,
$\Delta; a; (\Gamma, x : B); [a \mapsto A] \vdash \congen{N' : a}$, and
$(\Delta, \many{a}); a; \Gamma; [a \mapsto A, b \mapsto B'] \vdash \congen{M' : b}$.
Applying \cref{lemma:refinement-on-constraint-sat} to the latter using
$\delta''$ yields
$\Delta; (a, b); \rsubst''(\Gamma); [a \mapsto \rsubst''(A), b \mapsto \rsubst''(B')] \vdash \congen{M' : b}$,
which is equivalent to
$\Delta; (a, b); \Gamma; [a \mapsto A, b \mapsto B] \vdash \congen{M' : b}$,

We can then derive the following
\[
\inferrule
  {
    \meta{mostgen}(\Delta, (a, b), \Gamma, \congen{M' : b}, (\many{a}, c), \rsubst_\mathrm{m}) \\\\
    \Delta \vdash \rsubst' : (\many{a}, c) \Rightarrow_\mono \emptydelta \\
    B = \rsubst'(\rsubst_\mathrm{m}(a)) \\\\
    \Delta; (a, b); \Gamma; \rsubst[b \mapsto B] \vdash \congen{M' : b} \\
    \Delta; a; (\Gamma,x : B); \rsubst \vdash \congen{N' : a}
  }
  {\Delta; a; \Gamma; [a \mapsto A] \vdash \Let_\mono\; x = \letexists{b} \congen{M' : b} \;\In\; \congen{N' : a}}
\]

\item[Case $\Let (x : B) = M' \;\In\; N'$, where $M' \in \dec{GVal}$:]
Let $\omany{a}$, $H$ such that $B = \forall \omany{a}.H$.

The derivation of $\Delta; \Gamma \vdash M' : A$ has the following form, for some $\Delta', A'$:

\[
\inferrule*[Lab=\freezemlLab{LetAnn}]
    {(\Delta', B') = \msplit(B, M) \\
     \type{(\Delta, \Delta'); \Gamma}{M' : B'} \\
     \type{\Delta; (\Gamma, x : B)}{N' : A}
    }
    {\type{\Delta; \Gamma}{\Let \; (x : B) = M'\; \In \; N' : A}}
\]

By $M' \in \dec{GVal}$, we have $\Delta' = \many{a}$, $B' = H$ and
$\congen{M : A} = (\forall \omany{a}. \congen{M' : H}) \wedge \Def\; (x : B) \;\In\; \congen{N' : A}$.
Let $b$ be fresh.
By induction, we have
$(\Delta, \Delta');  b; \Gamma; [b \mapsto B'] \vdash \congen{M' : b}$
and
$\Delta; a; (\Gamma, x : B); [a \mapsto A] \vdash \congen{N' : A}$.
By \cref{lemma:substitution-on-constraint} we can substitute the former to
$(\Delta, \Delta');  \emptydelta; \Gamma; \emptysubst \vdash \congen{M' : H}$.

Recall that an implicit precondition of $\type{\Delta; (\Gamma, x : B)}{N' : A}$
we have $\Delta \vdash \wf{\Gamma}$, which implies $\Delta \vdash \wf{B}$.
Therefore, we have that $\ftv(B) - \Delta$ is empty and $B[a / A] = B$.

We can now derive the desired judgement.
\[
\inferrule*
    {
    \inferrule*
    {\inferrule*{(\Delta, \Delta');a; \Gamma;[a \mapsto A] \vdash  \congen{M' : H}}{\vdots}}
    {\Delta;a; \Gamma;[a \mapsto A] \vdash (\forall \omany{a}. \congen{M' : H})} \\
    \inferrule*
    {
     \text{for all } a \in \ftv(B) - \Delta \;\mid\;  \Delta; \Xi; \Gamma; \rsubst \vdash \monoc(a) \\\\
    \Delta; \Xi; (\Gamma,x :   B[a/A]); \rsubst \vdash C
    }
    {\Delta;a; \Gamma;[a \mapsto A] \vdash  \Def\; (x : B) \;\In\; \congen{N' : A}}
    }
    {\Delta;a; \Gamma;[a \mapsto A] \vdash (\forall \omany{a}. \congen{M' : H}) \wedge \Def\; (x : B) \;\In\; \congen{N' : A}}
\]

\item[Case $\Let (x : A) = M' \;\In\; N'$, where $M' \not\in \dec{GVal}$:]
Analogous to previous case, with $\Delta' = \emptyset$, $B' = B$.

\end{description}

\end{proof}


%% file: proofs/constraint-generation-complete.tex
\constraintgenerationsoundness*
\begin{proof}
\fe{We are definitelt relying on $\Delta \vdash \wf{M}$ now
with no obvious way to get rid of it. So
\cref{lemma:aux:principal-iff-mostgen} needs it to}
We prove the following, slightly more general property by induction on $M$:
If $\Delta;\Gamma \vdash \wf{M}$ and $\Delta; \Xi \vdash \wf{A}$ and $\Delta; \Xi
; \Gamma; \rsubst \vdash \congen{M : A}$, then $\Delta; \Gamma
\vdash M : \rsubst(A)$.

Note that by the implicit preconditions of
$\Delta; \Xi; \Gamma;\rsubst \vdash \congen{M : A}$ we have
$\Delta \disjoint \Xi$, $\Delta \vdash \wf{\Gamma}$, and
$\Delta \vdash \rsubst : \Xi \Rightarrow_\poly \emptydelta$.
The latter implies $\Delta \vdash \wf{\rsubst(A)}$.

We focus on the let cases.
\begin{description}



\item[Case $\Let\; x = M' \;\In\; N'$,  $M' \in \dec{GVal}$:]
We have
$\Delta; \Xi; \Gamma; \rsubst \vdash \Let_{\poly}\; x = \letexists{b} \congen{M' : b} \;\In\; \congen{N' : A}$.
The derivation of this has the following form for some
$B, \omany{a}, \Delta_\mathrm{m}, \Delta_\mathrm{o}, \rsubst_\mathrm{m}$ and
$\rsubst'$:
\[
\inferrule
  {
\meta{mostgen}(\Delta, (\Xi, b), \Gamma, \congen{M' : b}, \Delta_\mathrm{m}, \rsubst_\mathrm{m}) \\\\
    \Delta_\mathrm{o} = \ftv(\rsubst_\mathrm{m}(\Xi)) - \Delta \\
    \omany{a} = \ftv(\rsubst_\mathrm{m}(b)) - \Delta, \Delta_\mathrm{o} \\\\
    \Delta \vdash \rsubst' : \Delta_\mathrm{o} \Rightarrow_\mono \emptydelta \\
    B = \rsubst'(\rsubst_\mathrm{m}(b)) \\\\
    (\Delta, \many{a}); (\Xi, b); \Gamma; \rsubst[b \mapsto B] \vdash \congen{M' : b} \\
    \Delta; \Xi; (\Gamma,x : \forall \omany{a}. B); \rsubst \vdash \congen{N' : A}
  }
  {\Delta; \Xi; \Gamma; \rsubst \vdash \Let_{\poly}\; x = \letexists{b} \congen{M' : b} \;\In\; \congen{N' : A}}
\]

By induction, this gives us $(\Delta, \many{a}); \Gamma \vdash M' : B$ and
$\Delta; (\Gamma, x : \forall \omany{a}.B) \vdash N' : \rsubst(A)$.
By $M' \in \dec{GVal}$,
we have $\Delta \vdash M'$ and $\Delta \vdash \Gamma$.
According to \cref{lemma:translation-yields-wf-constraint} this implies
$\Delta; b; \Gamma \vdash \wf{\congen{M' : b}}$ (i.e., $b$ is the only free
flexible variable of $\congen{M' : b}$).
This means that $\congen{M' : b}$ leaves all variables of $\Xi$ entirely
unconstrained.
Therefore, we have that $\rsubst_\mathrm{m}$ maps all $a \in \Xi$ to pairwise
different variables $c$ from $\Delta_\mathrm{m}$ and
$c \not\in \ftv(\rsubst_{\mathrm{m}}(b))$ for all such $c$.
This implies $\omany{a} = \ftv(\rsubst_\mathrm{m}(b)) - \Delta,\Delta_\mathrm{o} = \ftv(\rsubst_\mathrm{m}(b)) - \Delta,\Delta_\mathrm{o}$ (i.e., removing
$\Delta_\mathrm{o}$ from the subtracted context has no effect) and $B = \rsubst_\mathrm{m}(b)$ (i.e.,
applying $\rsubst'$ to $\rsubst_\mathrm{m}(b)$ has no effect).

Consequently, we may strengthen
$\meta{mostgen}(\Delta, (\Xi, b), \Gamma, \congen{M' : b}, \Delta_\mathrm{m}, \rsubst_\mathrm{m})$
to
$\meta{mostgen}(\Delta, b, \Gamma, \congen{M' : b}, \many{a}, [b \mapsto \rsubst_\mathrm{m}(b)])$.
We may then apply \cref{lemma:aux:principal-iff-mostgen} to the latter, yielding
$\meta{principal}(\Delta, \Gamma, M', \many{a}, B)$.
Finally, \cref{lemma:principal-type-of-gval-is-guarded} shows us that $B$ is a
guarded type, and we may refer to it as $H$.

\fe{Check wf-ness for applying IH}
\fe{\Cref{lemma:aux:principal-iff-mostgen} is why we need $\Delta;\Gamma \vdash \wf{M}$}

We can therefore derive the desired property
$\Delta; \Gamma \vdash \Let\; x = M' \;\In\; N' : \rsubst(A)$ as follows:
\[
\inferrule*
    {\omany{a} = \ftv(H) - \Delta \\
     (\Delta, \omany{a}, M', H) \Updownarrow \forall \omany{a}.H  \\
     (\Delta, \many{a}); \Gamma \vdash M' : H \\
     \type{\Delta; (\Gamma, x : \forall \omany{a}.H)}{N' : \rsubst(A)} \\\\
     \meta{principal}(\Delta, \Gamma, M', \many{a}, H)
    }
    {\type{\Delta; \Gamma}{\Let \; x = M'\; \In \; N' : \rsubst(A)}}
\]
\fe{Mention well-formedness of $\Gamma, x : \forall \many{a}.H$?}

\item[Case $\Let\; x = M' \;\In\; N'$ if $M' \not\in \dec{GVal}$:]
This case is largely analogous to the previous one.

This time we have a derivation of the form:
\[
\inferrule
  {
    \meta{mostgen}(\Delta, (\Xi, a), \Gamma, \congen{M' : b}, \Delta_\mathrm{m}, \rsubst_\mathrm{m}) \\\\
    \Delta \vdash \rsubst' : \Delta_\mathrm{m} \Rightarrow_\mono \emptydelta \\
    B = \rsubst'(\rsubst_\mathrm{m}(b)) \\\\
    \Delta; (\Xi, a); \Gamma; \rsubst[b \mapsto B] \vdash \congen{M' : b} \\
    \Delta; \Xi; (\Gamma,x : B); \rsubst \vdash \congen{N' : A}
  }
  {\Delta; \kenv; \Gamma; \rsubst \vdash \Let_\mono\; x = \letexists{b} \congen{M' : b} \;\In\; \congen{N' : A}}
\]

Let $A' \coloneqq \rsubst_\mathrm{m}(b)$ and
$\omany{a} \coloneqq \ftv(A') - \Delta$.
By
$\meta{mostgen}(\Delta, (\Xi, a), \Gamma, \congen{M' : b}, \Delta_\mathrm{m}, \rsubst_\mathrm{m})$
we have
$(\Delta, \Delta_\mathrm{m}); (\Xi, a); \Gamma; \rsubst_\mathrm{m} \vdash \congen{M' : b}$.
By induction, this implies $(\Delta, \Delta_\mathrm{m}); \Gamma \vdash M' : A'$,
which we can strengthen to $(\Delta, \many{a}); \Gamma \vdash M' : A'$

We obtain $\Delta; (\Gamma, x : B) \vdash \rsubst(A)$ directly by applying the
induction hypothesis to
$\Delta; \Xi; (\Gamma,x : B); \rsubst \vdash \congen{N' : A}$.
Likewise, we obtain $\meta{principal}(\Delta, \Gamma, M', \many{a}, A')$ using
the same reasoning as in the previous case.

Finally, we observe that $\Delta \vdash \restriction{\rsubst'}{\many{a}} : \many{a} \Rightarrow \emptydelta$ holds
and $B = \restriction{\rsubst'}{\many{a}}(A')$.
Therefore, we have $(\Delta, \omany{a}, M, A') \Updownarrow B$ \\

\fe{Check wf-ness for applying IH}
Thus, we can derive the following:
\[
\inferrule*
    {
     \omany{a} = \ftv(A') - \Delta \\
     (\Delta, \omany{a}, M, A') \Updownarrow B \\
     \Delta, \many{a}; \Gamma \vdash M : A' \\
     \type{\Delta; \Gamma, x : A}{N : \rsubst(A)} \\
     \meta{principal}(\Delta, \Gamma, M, \many{a}, A')
    }
    {\type{\Delta; \Gamma}{\Let \; x = M'\; \In \; N' : \rsubst(A)}}
\]

\item[Case $\Let\; (x : B) = M' \;\In\; N'$ if $M' \in \dec{GVal}$:]
Let $\omany{b}$ and $H$ such that $B = \forall \omany{b}.H$.

We then have
$\Delta; \Xi; \Gamma; \rsubst \vdash (\forall \omany{b}. \congen{M' : H})
  \;\wedge\; \Def\; (x : B) \;\In\; \congen{N' : A}$.
The derivation tree of this contains derivations for
$(\Delta, \many{b}); \Xi; \Gamma; \rsubst \vdash \congen{M' : H}$ and
$\Delta; \Xi; (\Gamma, x : \rsubst(B)); \rsubst \vdash \congen{N' : A}$.
By $\Delta;\Gamma \vdash \wf{M}$ we have $\Delta \vdash B$ (i.e., $B$ contains no flexible variables).
Therefore, $\rsubst(B) = B$ and $\rsubst(H) = H$.
\fe{Here we rely on the fact that $M$ only contains variables from $\Delta$.
Otherwise, it seems like even the monomorphism restriction on the def constraint
wouldn't save us (which means that the toplevel quantifiers don't change): If
$\rsubst(B) \neq B$ then the derivation below doesn't work anymore: Note that it
requires the $B$ in the annotation to be the same as the one used for typing
$N'$.}
By induction, this then gives us
$(\Delta, \many{a}); \Gamma\vdash M' : H$ and
$\Delta; (\Gamma, x : B) \vdash N' : \rsubst(A)$.
\fe{Check wf-ness for applying IH}

Hence, we can derive
\[
  \inferrule*
    {(\many{a}, H) = \msplit(B, M') \\
     \type{(\Delta, \many{a}); \Gamma}{M' : H} \\
     \type{\Delta; \Gamma, x : B}{N' : \rsubst(A)}
    }
    {\type{\Delta; \Gamma}{\Let \; (x : B) = M'\; \In \; N' : \rsubst(A)}}
\]

\item[Case $\Let\; (x : B) = M' \;\In\; N'$ if $M' \not\in \dec{GVal}$]
This case is similar to the previous one:
We  have
$\Delta; \Xi; \Gamma; \rsubst \vdash  \congen{M' : B})
  \;\wedge\; \Def\; (x : B) \;\In\; \congen{N' : A}$ this time
and $\rsubst(B) = B$ due to $\Delta;\Gamma \vdash \wf{M}$.

We get
$\Delta; \Gamma\vdash M' : B$ and
$\Delta; (\Gamma, x : B) \vdash N' : \rsubst(A)$
by induction and can derive
\fe{Check wf-ness for applying IH}
\[
  \inferrule*
    {(\emptydelta, B) = \msplit(B, M') \\
     \type{\Delta; \Gamma}{M' : B} \\
     \type{\Delta; \Gamma, x : B}{N' : \rsubst(A)}
    }
    {\type{\Delta; \Gamma}{\Let \; (x : B) = M'\; \In \; N' : \rsubst(A)}}
\]

\end{description}
\end{proof}


%% file: proofs/preservation.tex



\proofContext{subject-reduction-statement}
\subjectreduction*

\begin{proof}
We carry out the proof by case analysis of the stack machine reduction rules.
Let $s$ be the state before, and $s'$ the state after the step.

\fe{
Consider  case $\lab{S-DefPop}$:
The right-to-left direction only holds, if well-formedness of \emph{stacks}
imposes all type variables in $A$ to be mono!
}
We focus on the let rules.
\begin{itemize}

\item Case $\lab{S-LetPolyPop}$:
We have that $s$ is of the form
$(\forall \Delta \snoc{} \exists \Xi \snoc{}  F :: \Let_\poly\; x = \letexists{a} \Box \;\In\; C :: \exists \many{a}, \Theta_0, \theta_0, \true)$ for some $F, C$, and $\many{a}$
and assume that the following conditions imposed by $\lab{S-LetPolyPop}$ hold:
\[
\bl
\wmany{a'};\: \wmany{a''} = \dec{partition}((\many{b}, a), \theta_0, \Theta_0) \\
A = \subst_0(a) \\
\omany{c}\: = \ftv(A) \cap \wmany{a'} \\
B = \forall \omany{c}.A \\
\Theta_1 = \Theta_0 - \wmany{a'} \\
\theta_1 = \restriction{\theta_0}{\Theta_1}
\el
\]

We need to show
\[
\ba{c}
  \Delta; \Xi; \emptygamma; \rsubst
   \vdash F :: \Let_\poly\; x = \letexists{a} \Box \;\In\; C :: \exists \many{a} [\true \wedge \Uof{\Theta_0, \theta_0}] \\
  \;\text{ iff }\; \\
\Delta; \Xi;
   \emptygamma; \rsubst \vdash F :: \exists{\wmany{a''}}[\Def\; (x: B) \;\In\; C  \wedge \Uof{\Theta_1, \theta_1}],
\ea
\]
which is equivalent to
\[
\ba{c}
  \Delta; \Xi; \emptygamma; \rsubst
   \vdash F[C'_0]  \\
  \;\text{ iff }\; \\
\Delta; \Xi;
   \emptygamma; \rsubst \vdash F[C'_1]
\ea
\]
if we define
$C'_0 := \Let_\poly\; x = \letexists{a} \exists \many{a}. \Uof{\Theta_0, \subst_0} \;\In\; C$
and
$C'_1 := \exists \wmany{a''}. (\Def\; (x: B) \;\In\; C) \wedge \Uof{\Theta_1, \theta_1}$.

We can prove this equivalence directly using \cref{lemma:aux:plugged-constraint-same-stack-irrelevant-sat}.
In order to apply this lemma, we need to show that the following holds:
For all
$\hat{\Delta}, \hat{\rsubst}$ we have
$(\Delta, \Deltaof{F}, \hat{\Delta}); (\Xi, \Xiof{F}); \hat{\rsubst}(\Gammaof{{F}}); \hat{\rsubst} \vdash C'_0$ iff
$(\Delta, \Deltaof{F}, \hat{\Delta}); (\Xi, \Xiof{F}); \hat{\rsubst}(\Gammaof{{F}}); \hat{\rsubst} \vdash C'_1$

To this end, suppose
$\hat{\Delta}$, and $\hat{\rsubst}$ are given.
Let $\Delta' \coloneqq \Delta, \Deltaof{F}, \hat{\Delta}$ and $\Xi' \coloneqq \Xi, \Xiof{F}$ and $\Gamma' \coloneqq \hat{\rsubst}(\Gammaof{F})$.
\begin{itemize}[wide]
\item[$\Longrightarrow$]:
We assume $\Delta'; \Xi'; \Gamma'; \hat{\rsubst} \vdash C'_0$~\premissNum{let-polt:C-p-one-sat}.
 The derivation of
this must have the following form, for some $\Delta_\mathrm{m}, \rsubst_\mathrm{m}, \Delta_\mathrm{o}, \omany{b}, \delta', A'$:
\[
\inferrule
  {
  \meta{mostgen}(\Delta', (\Xi', a), \Gamma', \exists \many{a}.\Uof{\Theta_0, \theta_0}, \Delta_\mathrm{m}, \rsubst_\mathrm{m}) \\\\
    \Delta_\mathrm{o} = \ftv(\rsubst_\mathrm{m}(\Xi')) - \Delta' \\
    \omany{b} = \ftv(\rsubst_\mathrm{m}(a)) - \Delta', \Delta_\mathrm{o} \\\\
    \Delta' \vdash \rsubst' : \Delta_\mathrm{o} \Rightarrow_\mono \emptydelta \\
    A' = \rsubst'(\rsubst_\mathrm{m}(a)) \\\\
    (\Delta', \many{b}); (\Xi', a); \Gamma'; \hat{\rsubst}[a \mapsto A'] \vdash \exists \many{a}.\Uof{\Theta_0, \theta_0} \\
    \Delta'; \Xi'; (\Gamma', x : \forall \omany{b}. A'); \hat{\rsubst} \vdash C
  }
  {\Delta'; \Xi'; \Gamma'; \hat{\rsubst} \vdash \Let_\poly\; x = \letexists{a} \exists \many{a}. \Uof{\Theta_0, \theta_0}  \;\In\; C}
\]
We can assume w.l.o.g.\ that the variables in $\Delta_\mathrm{m}$ are fresh.

By $\wf{s}$ and $\wf{s'}$, we have $\ftv(\Theta_0) = (\Xi, a, \many{a})$ as well as
idempotency of $\theta_0$ and $\theta_1$.

Let $\Xi_\mathrm{f} \coloneqq \ftv(\theta_0) - \Delta'$, which implies
$\Xi_\mathrm{f} \subseteq \Theta$.
By \cref{lemma:mostgen-iff-mapping-free-to-fresh} there exists a bijection
$\rsubst_\mathrm{r}$ such that
$\Delta_\mathrm{m} \vdash \rsubst_\mathrm{r} : \Xi_\mathrm{f} \Rightarrow_\mono \emptydelta$
and
$\rsubst_\mathrm{m} = \restriction{(\rsubst_\mathrm{r} \comp \theta_0)}{(\Xi', a)}$~\premissNum{let-poly:rsubst-r-composition}.

We observe $\rsubst_\mathrm{r}(\wmany{a''}) \subseteq \Delta_\mathrm{o}$.
We can now define $\widehat{\rsubst''}$ such that for all $b \in (\Xi', \wmany{a''})$ we have
\[
\widehat{\rsubst''}(b) =
\begin{cases}
\hat{\rsubst}(b) &\text{if } b \in \Xi' \\
\rsubst'(\rsubst_\mathrm{r}(b)) &\text{if } b \in \wmany{a''}
\end{cases}
\]
which yields
$\Delta' \vdash \widehat{\rsubst''} : (\Xi', \wmany{a''}) \Rightarrow_\poly \emptydelta$~\premissNum{let-poly:hat-theta-pp-wf}.

Next, we show
$\rsubst_\mathrm{r}(\omany{c}) = \omany{b}$~\premissNum{let-poly:omany-c-vs-omany-b}:
\[
\ba{r c l}
\rsubst_\mathrm{r}(\omany{c}) &=& \rsubst_\mathrm{r}(\ftv(\theta_0(a)) \cap \wmany{a'}) \\
&= &\rsubst_\mathrm{r}(\ftv(\theta_0(a)) - \Xi' -  \wmany{a''} - \Delta') \\
&&\;\text{(by $\ftv(\theta_0(a)) \subseteq \Delta', \Xi', \wmany{a'}, \wmany{a''}$)} \\
&= & \rsubst_\mathrm{r}(\ftv(\theta_0(a)) - \ftv(\theta_0(\Xi')) -  \wmany{a''} - \Delta') \\
&&\;(\text{by $\theta_0(b) = b $ for all $b \in \ftv(\theta(a))$ and $\wmany{a'} \disjoint \ftv(\theta(\Xi')) \subseteq \Delta', \Xi', \wmany{a''}$}) \\
&=& \rsubst_\mathrm{r}(\ftv(\theta_0(a)) - \ftv(\theta_0(\Xi')) - \Delta') \\
&&\;(\text{by } \wmany{a''} \subseteq \ftv(\theta_0(\Xi'))) \\
&=&\ftv(\rsubst_\mathrm{m}(a)) - \ftv(\rsubst_\mathrm{m}(\Xi')) - \Delta' \\
&&\;(\text{by \premissRef{let-poly:rsubst-r-composition}}) \\
&=&\ftv(\rsubst_\mathrm{m}(a)) - \Delta',\Delta_\mathrm{o} \\
&=&\omany{b}
\ea
\]

We now show $\widehat{\rsubst''}(c) = \rsubst'(\rsubst_\mathrm{r}(c))$ for all $c \in \ftv(\theta_0(a)) \cap (\Xi', \wmany{a''})$~\premissNum{let-poly:hat-rsubst-pp-on-Xi-p}.
First, we observe that for all such $c \in \wmany{a''}$ this holds by definition of $\widehat{\rsubst''}$.
Next, by
$(\Delta', \many{b}); (\Xi', a); \Gamma'; \widehat{\rsubst'} \vdash \exists \many{a}.\Uof{\Theta_0, \theta_0}$
(see derivation of \premissRef{let-polt:C-p-one-sat}) and
\cref{lemma:satisfying-U} we have that there exists $\rsubst_\mathrm{s}$ such
that
$(\Delta', \many{b}) \vdash \rsubst_\mathrm{s} : (\Xi', a, \many{a}) \Rightarrow_\poly \emptydelta$
and
$\widehat{\rsubst'} = \restriction{(\rsubst_\mathrm{s} \comp \theta_0)}{(\Xi', a)}$~\premissNum{let-poly:rsubst-s-comp}.
We have $\theta_0(b) = b$ for all $b \in \ftv(\theta_0)$ and therefore
$\rsubst_\mathrm{s}(b) = \widehat{\rsubst'}(b)$~\premissNum{let-poly:rsubst-s-vs-hat-rsubst-p} for any such $b$.
Using this, we observe
\[
\ba{r c l l}
&&\widehat{\rsubst'}(a) \\
&=& \rsubst_\mathrm{s}(\theta_0((a)) \qquad&\text{(by \premissRef{let-poly:rsubst-s-comp})}\\
&=& A' &\text{(by def.\ of $\widehat{\rsubst'}$)}\\
&=&\rsubst'(\rsubst_\mathrm{m}(a)) &\text{(by def.\ of $A'$)} \\
&=&\rsubst'(\rsubst_\mathrm{r}(\theta_0(a))) &\text{(by \premissRef{let-poly:rsubst-r-composition})}\\
\ea
\]
This implies that for all $b \in \ftv(\theta_0(a)) \cap \Xi'$ we have
$\rsubst'(\rsubst_\mathrm{r}(b)) = \rsubst_\mathrm{s}(b) \stackrel{\premissRef{let-poly:rsubst-s-vs-hat-rsubst-p}}{=}  \widehat{\rsubst'}(b) = \widehat{\rsubst''}(b)$
and therefore \premissRef{let-poly:hat-rsubst-pp-on-Xi-p} holds.

We now show that $\widehat{\rsubst''}(B)$ is alpha-equivalent to $\forall \many{b}.A'$:
\[
\ba{rcl}
\widehat{\rsubst''}(B) &=& \widehat{\rsubst''}(\forall \omany{c}. \theta_0(a)) \\
&=& \forall \omany{c}. \widehat{\rsubst''}(\theta_0(a)) \\
  &&\; (\text{due to } \premissRef{let-poly:hat-theta-pp-wf}, \omany{c} \disjoint \Xi', \wmany{a''} \text{ and } \omany{c} \disjoint \Delta',\Delta_\mathrm{m}) \\
&=& \forall \omany{b}. \left(\widehat{\rsubst''}(\theta_0(a))[\omany{c} \mapsto \omany{b}] \right) \\
&&\;(\text{due to \premissRef{let-poly:omany-c-vs-omany-b}, $\omany{c} \disjoint \omany{b}$}) \\
&=& \forall \omany{b}. \rsubst'(\rsubst_\mathrm{r}(\theta_0(a))) \\
&&\;(\text{due to \premissRef{let-poly:rsubst-r-composition}, \premissRef{let-poly:hat-rsubst-pp-on-Xi-p}, $\widehat{\rsubst''}(b) = b$ for all $b \in \omany{b}$}) \\
&=& \forall \omany{b}. \rsubst'(\rsubst_\mathrm{m}(a)) \\
&=& \forall \omany{b}. A'
\ea
\]
This means that by
$\Delta'; \Xi'; (\Gamma', x : \forall \omany{b}. A'); \hat{\rsubst} \vdash C$
(see derivation of \premissRef{let-polt:C-p-one-sat}) we also have
$\Delta'; \Xi'; (\Gamma', x : \widehat{\rsubst''}(B)); \hat{\rsubst} \vdash C$,
which we can weaken to
$\Delta'; (\Xi', \wmany{a''}); (\Gamma', x : \widehat{\rsubst''}(B)); \widehat{\rsubst''} \vdash C$~\premissNum{let-poly:C-sat}.

We now show that for all $b \in \ftv(B) - \Delta'$ we have
$\Delta'; \Xi'; \Gamma'; \widehat{\rsubst''} \vdash \monoc(b)$~\premissNum{let-poly:ftvs-mono}:
We have
$\ftv(B) - \Delta' = \ftv(\forall \omany{c}. \theta_0(a)) - \Delta' \subseteq \Xi, \wmany{a''}$.
By \premissRef{let-poly:hat-rsubst-pp-on-Xi-p} we have
$\widehat{\rsubst''}(b) = \rsubst'(\rsubst_\mathrm{r}(b))$ for any such $b$.
By $\Delta' \vdash \rsubst' : \Delta_\mathrm{o} \Rightarrow_\mono \emptydelta$
and $\rsubst_\mathrm{r}$ being a bijection on variables, we have
$\Delta' \vdash_\mono \wf{\widehat{\rsubst''}(b)}$.

We now show
$\Delta'; (\Xi', \wmany{a''}); \Gamma'; \widehat{\rsubst''} \vdash \Uof{\Theta_1, \theta_1}$~\premissNum{let-poly:U-successor-sat}:
Recall that per \premissRef{let-poly:rsubst-s-comp}, we have
$\widehat{\rsubst'} = \restriction{(\rsubst_\mathrm{s} \comp \theta_0)}{(\Xi', a)}$,
where $\widehat{\rsubst'}$ and $\widehat{\rsubst''}$ coincide on $\Xi'$.
Further, $\theta_0(b) = b$ for all $b \in \wmany{a''}$.
We can therefore apply \cref{lemma:satisfying-U} to $\widehat{\rsubst''}$ to obtain
\premissRef{let-poly:U-successor-sat}.


Using \premissRef{let-poly:C-sat}, \premissRef{let-poly:ftvs-mono}, and \premissRef{let-poly:U-successor-sat} we can now derive
$\Delta'; \Xi'; \Gamma'; \hat{\rsubst} \vdash C'_1$ as follows:

\[
\inferrule*
  {
   \inferrule*
    {
      \inferrule*
        {
        \inferrule*
          {
            \Delta'; (\Xi', \wmany{a''}); (\Gamma', x :  \widehat{\rsubst''} B); \widehat{\rsubst''} \vdash C \\
            \text{for all } b \in \ftv(B) - \Delta' \;\mid\;  \Delta'; \Xi'; \Gamma'; \widehat{\rsubst''} \vdash \monoc(b)
          }
          {
            \Delta'; (\Xi', \wmany{a''}); \Gamma'; \widehat{\rsubst''} \vdash \Def\; (x: B) \;\In\; C \\
          } \\
        \Delta'; (\Xi', \wmany{a''}); \Gamma'; \widehat{\rsubst''} \vdash \Uof{\Theta_1, \theta_1}
        }
        {
          \Delta'; (\Xi', \wmany{a''}); \Gamma'; \widehat{\rsubst''} \vdash (\Def\; (x: B) \;\In\; C) \wedge \Uof{\Theta_1, \theta_1}
        }
     }
     {\vdots}
  }
  {\Delta'; \Xi'; \Gamma'; \hat{\rsubst} \vdash \exists \wmany{a''}. (\Def\; (x: B) \;\In\; C) \wedge \Uof{\Theta_1, \theta_1}}
\]

\item[$\Longleftarrow$]:
We assume $\Delta'; \Xi'; \Gamma'; \hat{\rsubst} \vdash C'_1$~\premissNum{let-poly:C-p-two-sat}. The derivation of
this must have the following form for some $\widehat{\rsubst''}$:

\[
\inferrule
  {
   \inferrule
    {
      \inferrule
        {
        \inferrule
          {
            \inferrule
              {\vdots}
              {\Delta'; (\Xi', \wmany{a''}); (\Gamma', x :  \widehat{\rsubst''} B); \widehat{\rsubst''} \vdash C} \\\\
            \text{for all } b \in \ftv(B) - \Delta' \;\mid\;  \Delta'; \Xi'; \Gamma'; \widehat{\rsubst''} \vdash \monoc(b)
          }
          {
            \Delta'; (\Xi', \wmany{a''}); \Gamma'; \widehat{\rsubst''} \vdash \Def\; (x: B) \;\In\; C \\
          } \\
        \inferrule
          {\vdots}
          {
            \Delta'; (\Xi', \wmany{a''}); \Gamma'; \widehat{\rsubst''} \vdash \Uof{\Theta_1, \theta_1}
          }
        }
        {
          \Delta'; (\Xi', \wmany{a''}); \Gamma'; \widehat{\rsubst''} \vdash (\Def\; (x: B) \;\In\; C) \wedge \Uof{\Theta_1, \theta_1}
        }
     }
     {\vdots}
  }
  {\Delta'; \Xi'; \Gamma'; \hat{\rsubst} \vdash \exists \wmany{a''}. (\Def\; (x: B) \;\In\; C) \wedge \Uof{\Theta_1, \theta_1}}
\]
This immedately gives us
$\Delta' \vdash \widehat{\rsubst''} : (\Xi', \wmany{a''}) \Rightarrow_\poly \emptydelta$~\premissNum{let-poly:rtl-delta-pp-wf}.

Let $\Xi_\mathrm{f}$ be defined as in the $\Rightarrow$ case and let
$\rsubst_\mathrm{r}$ be a bijection from $\Xi_\mathrm{f}$ to fresh variables.
Further, let $\Delta_\mathrm{m} \coloneqq \ftv(\rsubst_\mathrm{r})$ and
$\delta_\mathrm{m} \coloneqq \restriction{(\rsubst_\mathrm{r} \comp \theta_0)}{(\Xi', a)}$.
By \cref{lemma:mostgen-iff-mapping-free-to-fresh} we then have
$\meta{mostgen}(\Delta', (\Xi', a),\allowbreak \Gamma',\allowbreak \exists \many{a}.\Uof{\Theta_0, \theta_0},\allowbreak \Delta_\mathrm{m}, \rsubst_\mathrm{m})$.
We may now define $\Delta_\mathrm{o}$ and $\omany{b}$ as in the $\Rightarrow$
case, making each of them a subset of $\Delta_\mathrm{m}$.

We now define $\rsubst'$ for all $b \in \Delta_\mathrm{o}$ as follows:
\[
\rsubst'(b) =
\begin{cases}
\widehat{\rsubst''}(\rsubst_\mathrm{r}^{-1}(b)) &\text{if } \rsubst_\mathrm{r}^{-1}(b) \in \ftv(\theta(a)) - \wmany{a'} - \Delta' \\
\dec{unit} &\text{otherwise}
\end{cases}
\]

We have $\wmany{a'} \subseteq \omany{c}$ and therefore
$\ftv(B) = \ftv(\forall \many{c}. \theta(a)) = \ftv(\theta(a)) - \wmany{a'}$.
Thus, by the definition above and the second premise of the derivation of
$\Delta'; (\Xi', \wmany{a''}); \Gamma'; \widehat{\rsubst''} \vdash \Def\; (x: B) \;\In\; C$,
we have
$\Delta' \vdash \rsubst' : \Delta_\mathrm{o} \Rightarrow_\mono \emptydelta$.

The definition of $\delta'$ immedately yields
$\widehat{\rsubst''}(c) = \rsubst'(\rsubst_\mathrm{r}(c))$ for all
$c \in \ftv(\theta_0(a)) \cap (\Xi', \wmany{a''})$. (which we called
\premissRef{let-poly:hat-rsubst-pp-on-Xi-p} in the $\Rightarrow$ direction).
We define $A' \coloneqq \rsubst'(\rsubst_\mathrm{m}(a))$.
Together with \premissRef{let-poly:rtl-delta-pp-wf}, we can use the same
reasoning as in the $\Rightarrow$ case to obtain
$\omany{b} = \rsubst_\mathrm{r}(\omany{c})$ and the alpha-equivalence of and
$\widehat{\rsubst''}(B)$ and $\forall \omany{b}.A'$.

Hence,
$\Delta'; (\Xi', \wmany{a''}); (\Gamma', x : \widehat{\rsubst''} B); \widehat{\rsubst''} \vdash C$
(see derivation of \premissRef{let-poly:C-p-two-sat}) is equivalent to
$\Delta';\allowbreak (\Xi', \wmany{a''});\allowbreak (\Gamma', x : \forall \omany{b}. A');\allowbreak \widehat{\rsubst''} \vdash C$.
We have $\Delta' \vdash \wf{\forall \omany{b}. A'}$ and
$(\Delta, \Deltaof{F}); \Xi'; (\Gammaof{F}, x : \bot) \vdash \wf{C}$ which means we can
weaken it to
$\Delta'; \Xi'; (\Gamma', x : \forall \omany{b}. A'); \hat{\rsubst} \vdash C$.

Next, we define $\widehat{\rsubst'} \coloneqq \hat{\rsubst}[a \mapsto A']$ as in the
$\Rightarrow$ case and show that
$(\Delta', \many{b}); (\Xi', a); \Gamma'; \widehat{\rsubst'}\vdash \exists \many{a}.\Uof{\Theta_0,\allowbreak \theta_0}$~\premissNum{let-polt:rtl-U-0-sat}
holds.
To this end, we wish to apply \cref{lemma:satisfying-U} to
$\Delta'; (\Xi', \wmany{a''}); \Gamma'; \widehat{\rsubst''} \vdash \Uof{\Theta_1, \theta_1}$
(see derivation of \premissRef{let-poly:C-p-two-sat}), which gives us the
existence of $\theta'_1$ such that
$\Delta' \vdash \theta'_1 : \Theta_1 \Rightarrow \emptydelta$ and
$\widehat{\rsubst''} = \theta'_1 \comp \theta_1$~\premissNum{let-poly:rtl-theta-p-1-comp}.
We now show that the necessary preconditions of the lemma are satisfied.

Recall that $\ftv(\Theta_1) = \ftv(\Theta_0) - \wmany{a'}$ and
$\many{c} \subseteq \wmany{a'}$.
We now define $\theta'_0$ as an extension of $\theta'_1$ by setting
$\theta'_0(c) = \rsubst_\mathrm{r}(c)$ for all $c \in \many{c}$, and
$\theta'_0(c) = \dec{unit}$ for all $c \in \wmany{a'} - \many{c}$.
This implies $\Delta', \many{b} \vdash \theta'_0 : \Theta_0 \Rightarrow \emptydelta$.
We now show that
$\widehat{\rsubst'} = \restriction{(\theta'_0 \comp \theta_0)}{(\Xi', a)}$~\premissNum{let-poly:rtl-theta-p-0-comp}.
For all $b \in \Xi'$ we immediately get
$\theta'_0(\theta_0(b)) = \widehat{\rsubst''}(b)$ by $\Xi' \subseteq \ftv(\Theta_1)$
and \premissRef{let-poly:rtl-theta-p-1-comp}.

It remains to show that $\theta'_0(\theta_0(a)) = \widehat{\rsubst'}(a) \stackrel{\text{def.\ }}{=} A'$.
By definition of $A'$ and $\rsubst_\mathrm{m}$ we have
$A' = \rsubst'(\rsubst_\mathrm{m}(a)) = \rsubst'(\rsubst_\mathrm{r}(\theta_0(a))$.
Therefore, it sufficies to show that for all $b \in \ftv(\theta_0(a)) - \Delta'$ we have
$\theta'_0(b) = \rsubst'(\rsubst_\mathrm{r}(b))$.
If $b \in (\Xi', \wmany{a''})$ we have $\theta'_0(b) = \theta'_1(b)$.
By $b \in \ftv(\theta_0)$ we have $\theta_0(b) = b$, which means
that \premissRef{let-poly:rtl-theta-p-1-comp} imposes $\theta'_1(b) = \widehat{\rsubst''}(b)$.
By definition of $\rsubst'$ we then have
$
\theta'_0(b) =
\theta'_1(b) =
\widehat{\rsubst''}(b) =
\widehat{\rsubst''}(\rsubst^{-1}_\mathrm{r}(\rsubst_\mathrm{r}(b))) =
\rsubst'(\rsubst_\mathrm{r}(b))
$.
Otherwise, if $b \in \wmany{a'}$, we have
$\rsubst_\mathrm{r}(b) \not\in \Delta_\mathrm{o}$ and therefore
$\theta'_0(b) = \rsubst_\mathrm{r}(b) = \rsubst'(\rsubst_\mathrm{r}(b))$.
Finally, having shown \premissRef{let-poly:rtl-theta-p-0-comp} we
may apply
\cref{lemma:satisfying-U}
to obtain \premissRef{let-polt:rtl-U-0-sat}.

In conclusion, we can now derive
$\Delta'; \Xi'; \Gamma'; \hat{\rsubst} \vdash C'_1$ as follows:
\[
\inferrule
  {
  \meta{mostgen}(\Delta', (\Xi', a), \Gamma', \exists \many{a}.\Uof{\Theta_0, \theta_0}, \Delta_\mathrm{m}, \rsubst_\mathrm{m}) \\\\
    \Delta_\mathrm{o} = \ftv(\rsubst_\mathrm{m}(\Xi')) - \Delta' \\
    \omany{b} = \ftv(\rsubst_\mathrm{m}(a)) - \Delta', \Delta_\mathrm{o} \\\\
    \Delta' \vdash \rsubst' : \Delta_\mathrm{o} \Rightarrow_\mono \emptydelta \\
    A' = \rsubst'(\rsubst_\mathrm{m}(a)) \\\\
    (\Delta', \many{b}); (\Xi', a); \Gamma'; \hat{\rsubst}[a \mapsto A'] \vdash \exists \many{a}.\Uof{\Theta_0, \theta_0} \\
    \Delta'; \Xi'; (\Gamma', x : \forall \omany{b}. A'); \hat{\rsubst} \vdash C
  }
  {\Delta'; \Xi'; \Gamma'; \hat{\rsubst} \vdash \Let_\poly\; x = \letexists{a} \exists \many{a}. \Uof{\Theta_0, \theta_0}  \;\In\; C}
\]

\end{itemize}
\end{itemize}
\end{proof}

%% file: proofs/progress.tex
\progress*
\begin{proof}
\proofContext{progress}
Let $s := (F, \Theta, \theta, C)$.
By assumption, we have
$\emptydelta; \emptykenv; \emptygamma ;\emptysubst \vdash F[C \wedge \Uof{\Theta, \theta}]$~\premissNum{sat}.

We first assume $C \neq \true$ and show that one of the rules in
\cref{paper-fig:constraints-stack-machine} is applicable.

\begin{itemize}
\item Case $C_1 \wedge C_2$: Rule $\lab{S-ConjPush}$ is applicable.
\item Case $A \ceq B$: The left-hand-side of $\lab{S-Eq}$ matches.
The derivation of \premissRef{sat}
must contain a subderivation of the form
\[
\inferrule
  {
    \rsubst'(A) = \rsubst'(B)
  }
  {\Delta'; \Xi'; \Gamma'; \rsubst' \vdash A \ceq B}
\]
for some $\Delta', \Xi', \Gamma', \rsubst'$.


Since we are using the same unification algorithm as in \cite{EmrichLSCC20}, we can then use
\cite[Theorem 5]{EmrichLSCC20} to show that unification succeeds.

\item Case $\freeze{x : A}$: The left-hand-side of $\lab{S-Inst}$ matches.
By \premissRef{sat} we have $x \in \Gammaof{F}$, meaning that the rule succeeds.

\item Case $x \preceq A$: Rule $\lab{S-Freeze}$ is applicable (using the same
argument to show $x \in \Gammaof{F}$ as in the previous case).
\item Case $\exists a. C'$: Rule $\lab{S-ExistsPush}$ is applicable.
\item Case $\forall a. C'$: Rule $\lab{S-ForallPush}$ is applicable.
\item Case $\Def\; (x : A) \;\In\; C'$: The left-hand-side of $\lab{S-DefPush}$ matches.

The derivation of \premissRef{sat} must contain a sub-derivation of the form

\[
\inferrule
  {
    (\text{for all } a \in \ftv(A) - (\Deltaof{F}, \Delta') \mid (\Deltaof{F}, \Delta'); \Xiof{F}; \rsubst(\Gammaof{F}); \rsubst \vdash \monoc(a) \\\\
    (\Deltaof{F}, \Delta'); \Xiof{F}; (\rsubst(\Gammaof{F}),x :  \rsubst A); \rsubst \vdash C'
  }
  {(\Deltaof{F}, \Delta'); \Xiof{F}; \rsubst(\Gammaof{F}); \rsubst \vdash \Def\; (x : A) \;\In\; C'}
\]
for some $\rsubst$,
where $C' = C \wedge \Uof{\Theta, \theta}$ and $\ftv(\Theta) = \Xi$.
Note that by $\wf{s}$ we have $\ftv(A) \subseteq (\Deltaof{F},\Xiof{F})$.
\fe{Here we actually use the wf-ness of $C$ part of state wf-ness!}

According to \cref{lemma:satisfying-U}, we have
$\rsubst = \theta' \comp \theta$ for some $\theta'$ with $ (\Deltaof{F}, \Delta') \vdash \theta' : \Theta \Rightarrow \emptydelta$.
Therefore, the monomorphism conditions imposed by $\lab{S-DefPush}$ are satisfied.
\item Case $\monoc(a)$: Analogous to def case; the sub-derivation for $\monoc(a)$ implies that
the monomorphism conditions imposed by $\lab{S-Mono}$ are satisifed, making the rule applicable.
\item Case $\Let_R\; x = \letexists{a} C_1 \;\In\; C_2$: Rule $\lab{S-LetPush}$ is applicable.
\end{itemize}

We now conside the case that $C$ is $\true$.
Due the assumption about the shape of $F[C]$, we know that $F$ is neither emtpy
nor of the shape $\forall \Delta :: \exists \omany{a}$ for any $\omany{a}$.

We perform a case analysis on the topmost stack frame of $F$:
\begin{itemize}
\item Case $\Box \wedge C_2$: Rule $\lab{S-ConjPop}$ is applicable.
\item Case $\Def\; (x : A)$: Rule $\lab{S-DefPop}$ is applicable.
\item
Case $\Let_R\; x = \letexists{a} C_1 \;\In\; C_2$: Rule $\lab{S-LetPolyPop}$ or
$\lab{S-LetMonoPop}$ is applicable, where the sequence $\omany{a}$ mentioned in
the rule's definition is empty.
None of the respective rule's side conditions can fail.

\item Case $\forall a$: Let $F'$ be defined such that $F = F' :: \forall a$.
By assumption $\wf{s}$ we have that the  variables in $\Theta$ are
exactly the variables bound by $\exists$ or $\Let$ frames in $F'$.
Suppose there exists $b \in \ftv(\Theta)$ such that $a \in \ftv(\theta(b))$,
which would cause the rule to fail.
Then there exists a frame $f$ in $F'$ that binds $b$.
Let $F_p$ be the (possibly empty) prefix of $F'$ up to, but not including, frame
$f$ and let $F_s$ be the (possibly empty) suffix from there (i.e., $F' = F_p \snoc{} f :: F_s$).

We distinguish two sub-cases further:
\begin{enumerate}
\item
If $f = \exists b$ then the derivation of \premissRef{sat} must contain a
sub-derivation of the following form:
\[
\inferrule
  {
    (\Deltaof{F_p}, \Delta'); (\Xiof{F_p}, b); \rsubst(\Gammaof{F_p}); \rsubst[b \mapsto A] \vdash C'
  }
  {(\Deltaof{F_p}, \Delta'); \Xiof{F_p}; \rsubst(\Gammaof{F_p}); \rsubst \vdash \exists b.C'}
\]
for some $A$, $\rsubst$, and $\Delta'$,
where $C' = F_s[C \wedge \Uof{\Theta, \theta}]$.

Note that
$(\Deltaof{F_p}, \Delta'); (\Xiof{F_p}, b); \rsubst(\Gammaof{F_p}); \rsubst[b \mapsto A] \vdash C'$
implies $(\Deltaof{F_p}, \Delta') \vdash \wf{A}$~\premissNum{A-wf}.
Further,
as $F_s$ contains the frame $\forall a$, we have $a \not\in \Deltaof{F_p}, \Delta'$~\premissNum{a-not-in-stuff}.

Likewise, there existis a subderivation showing
$(\Deltaof{F}, \Delta''); \Xiof{F}; \rsubst''(\Gammaof{F}); \rsubst'' \vdash \Uof{\Theta, \theta}$
for some $\Delta''$ and $\rsubst''$, where $\Delta'' \supseteq \Delta'$ and $\rsubst''$ is an extension of $\rsubst[b \mapsto A]$.
By \cref{lemma:satisfying-U} we have that there exists
$\theta'$ such that
$\Deltaof{F}, \Delta'' \vdash \theta' : \Theta \Rightarrow \emptykenv$ and
$\rsubst'' = (\theta' \comp \theta)$.
Because $\rsubst''$ is an extension of $\rsubst[b \mapsto A]$, this
implies $A = \theta'(\theta(b))$.


By $a \in \Deltaof{F}$ we have $\theta'(a) = a$. Due to assumption
$a \in \ftv(\theta(b))$ we then have $a \in \ftv(A)$
However, we have $a \not\in \btv(F_p),\Delta'$~\premissRef{a-not-in-stuff} and
$\Deltaof{F_p}, \Delta' \vdash \wf{A}$~\premissRef{A-wf}, yielding the contradiction $a \not\in \ftv(A)$.

\item
If $f$ is of the the form $\Let_\poly\; x = \letexists{b} C_1 \;\In\; C_2$ then
the derivation of \premissRef{sat} must contain a sub-derivation of the
following form:
\[
\inferrule
  {
    \dots \\\\
    (\Deltaof{F_p}, \Delta', \many{a}); (\Xiof{F_p}, b); \rsubst(\Gammaof{F_p}); \rsubst[b \mapsto A] \vdash C_1 \\
  }
  {(\Deltaof{F_p}, \Delta'); \Xiof{F_p}; \rsubst(\Gammaof{F_p}); \rsubst \vdash \Let_\poly\; x = \letexists{b} C_1 \;\In\; C_2}
\]
for some $A, \Delta', \many{a}$  and $\rsubst$,
where $C_1 = F_s[C \wedge \Uof{\Theta, \theta}]$.
We have $(\Deltaof{F_p}, \Delta', \many{a}) \vdash \wf{A}$ and $a \not\in \Deltaof{F_p}, \Delta', \many{a}$
and may therefore obtain the same contradiction as in the previous case $f = \exists b$.



\item The case $\Let_\mono\; x = \letexists{b} C_1 \;\In\; C_2$ is analogous.

\end{enumerate}

\item Case $\exists{a}$:
If the topmost stack frames of $F$ have the shape
$\Let_R x = \letexists{a} \Box \;\In\; C' :: \exists \omany{b}$, then
$\lab{S-LetPolyPop}$ or $\lab{S-LetMonoPop}$ is applicable, as discussed before.
Otherwise, due to our assumption about the shape of $F[C]$, there exists a frame
$f$ in $F$ that isn't an $\exists$ frame and we can apply $\lab{S-ExistsLower}$,
which always succeeds.

\end{itemize}
\end{proof}


%% file: proofs/constraint-based-type-inference-right.tex

\proofContext{constraint-based-type-inference-right}
\constraintbasedtypeinferenceright*

\begin{proof}

Suppose
$(\emptystack, \emptykenv, \emptysubst, \forall \Delta. \exists a. \Def\; \Gamma \;\In\; \congen{M : a})$
$\to^{*}$
$( \forall \Delta :: \exists\; (a, \many{b}), \Theta, \subst,\allowbreak \true)$ and $\Delta \vdash \theta' : \Theta \Rightarrow \emptydelta$.
Let $s$ refer to the first state of the sequence above and $s'$ to its last state.

We apply \cref{lemma:translation-yields-wf-constraint}, which gives us
$\Delta; a ; \Gamma \vdash \wf{\congen{M : a}}$.
We therefore have $\vdash \wf{s}$.
By \cref{lemma:stack-machine-steps-preserve-wf-ness} we then have
$\vdash \wf{s'}$, too, which implies
$\Delta \vdash \Theta \Rightarrow \emptydelta$ and $\ftv(\Theta) = \many{b},a$.

We can therefore define $\rsubst$ as
$\theta' \comp \restriction{\theta}{\{ a \}}$.
This allows us to
apply \cref{theorem:solver-correct}.
We instantiate the right-to-left direction of the theorem such that we need to
show that the following properties hold:
\begin{premisses}
 \item $(\emptystack, \emptykenv, \emptysubst, \forall \Delta.\, \exists a.\, \Def\: \Gamma \:\In\: C) \to^{*}
    (\forall \Delta :: \exists\; (a,\many{b}), \kenv, \subst, \true)$
  \item $\Delta \vdash \theta' : \Theta \Rightarrow \emptydelta$
 \item $\restriction{(\theta' \comp \theta)}{a} = \rsubst$
\end{premisses}

The first two properties follow immediately by assumption, the third one holds
by definition of $\rsubst$.
Therefore, the right-to-left direction of
\cref{theorem:solver-correct} gives us
$\Delta;  a ; \Gamma; \rsubst \vdash \congen{M : a}$.
\Cref{theorem:constraint-generation-completeness} then immediately
yields $\Delta; \Gamma \vdash M : \rsubst(a)$.
By definition of $\rsubst$ this is equivalent to the property to show.

\end{proof}


%% file: proofs/constraint-based-type-inference-left.tex
\proofContext{constraint-based-type-inference-left}
\constraintbasedtypeinferenceleft*

\begin{proof}
We assume $\Delta; \Gamma \vdash M : A$, which implies $\Delta \vdash \Gamma$ and $\Delta;\Gamma \vdash \wf{M}$.
This means that
\cref{theorem:constraint-generation-soundness} gives us
$\Delta; a ; \Gamma; [a \mapsto A] \vdash \congen{M : a}$.

We apply the left-to-right direction of
\cref{theorem:solver-correct} (using $\congen{M : a}$ for
$C$ and $[a \mapsto A]$ for $\rsubst$ in the theorem's statement) which gives us the existence of
$\Theta, \theta, \theta', \many{c}$ s.t.\
\begin{premisses}
 \item $(\emptystack, \emptykenv, \emptysubst, \forall \Delta. \exists a.\Def\;
   \Gamma\; \In \; \congen{ M : a }) \to^* (\forall \Delta :: \exists (a, \many{c}), \kenv, \subst, \true)$ \premissLabel{premiss:machine-steps}
 \item $\Delta \vdash \subst' : \kenv \Rightarrow \emptydelta$
 \item $\restriction{(\subst' \comp \subst)}{\lbrace a \rbrace} = [ a \mapsto
   A ]$ \premissLabel{premiss:substs-compose}
\end{premisses}

Clearly, we have
$(\subst' \comp \restriction{\subst}{\lbrace a \rbrace})(a) = (\subst' \comp \subst)(a) =  [ a \mapsto
   A ](a) = A,
$
which is the second property we need to show.
By choosing $\Xi = (a, \many{c})$, property \premissRef{premiss:machine-steps} becomes
the first property that we needed to show.

\end{proof}
